\documentclass[sigconf]{acmart}

\usepackage{algorithm}
\usepackage{algorithmic,eqparbox}
\usepackage[title]{appendix}
\usepackage{color}
\usepackage{balance}
\usepackage{url}
\usepackage{xspace}
\usepackage{mathtools}
\usepackage{enumitem}
\usepackage{subcaption}

\newcommand{\ie}{\emph{i.e.,}\xspace}
\newcommand{\eg}{\emph{e.g.,}\xspace}

\newcommand{\eat}[1]{}
\DeclarePairedDelimiter\ceil{\lceil}{\rceil}
\def\EndOfProof{\nolinebreak\ \hfill\rule{1.5mm}{2.7mm}}

\setcopyright{none}

\begin{document}
\title[]{Towards Plug-and-Play Visual Graph Query Interfaces: Data-driven Canned Pattern Selection for Large Networks}
\subtitle{[Technical Report]}
\author{Zifeng Yuan $^{\S,\ddag}$\hspace{3ex} Huey Eng Chua$^\ddag$\hspace{3ex} Sourav S Bhowmick$^\ddag$ \hspace{3ex} Zekun Ye $^{\S,\ddag}$\hspace{1ex} Wook-Shin Han$^\spadesuit$\hspace{1ex} Byron Choi$^\dag$\\}
\affiliation{%
\institution{
	$^\ddag$School of Computer Science and Engineering, Nanyang Technological University, Singapore\\
	$^\S$School of Computer Science, Fudan University, China\\
	$^\spadesuit$POSTECH, South Korea\\
	$^\dag$Department of Computer Science, Hong Kong Baptist University, Hong Kong SAR
}
\country{}
}

\email{hechua|assourav@ntu.edu.sg, wshan@dblab.postech.ac.kr, bchoi@comp.hkbu.edu.hk, zfyuan16|zkye16@fudan.edu.cn}


\begin{abstract}
\textit{Canned patterns} (\ie small subgraph patterns) in visual graph query interfaces (a.k.a \textsc{gui}) facilitate efficient query formulation by enabling \textit{pattern-at-a-time} construction mode. However, existing \textsc{gui}s for querying large networks either do not expose any canned patterns or if they do then they are typically selected manually based on domain knowledge. Unfortunately, manual generation of canned patterns is not only labor intensive but may also lack diversity for supporting efficient visual formulation of a wide range of subgraph queries. In this paper, we present a novel generic and extensible framework called \textsc{Tattoo} that takes a data-driven approach to \textit{automatically} selecting canned patterns for a \textsc{gui} from large networks\eat{, thereby taking a concrete step towards the vision of data-driven construction of visual query interfaces for network data}. Specifically, it first \textit{decomposes} the underlying network into \textit{truss-infested} and \textit{truss-oblivious} regions. Then \textit{candidate} canned patterns capturing different real-world query topologies are generated from these regions. Canned patterns based on a user-specified \textit{plug} are then \textit{selected} for the \textsc{gui} from these candidates by maximizing \textit{coverage} and \textit{diversity}, and by minimizing the \textit{cognitive load} of the pattern set. Experimental studies with real-world datasets demonstrate the benefits of \textsc{Tattoo}. Importantly, this work takes a concrete step towards realizing \textit{plug-and-play} visual graph query interfaces for large networks.
\end{abstract}

%
%



\maketitle

\eat{{\fontsize{8pt}{8pt} \selectfont
\textbf{ACM Reference Format:}\\
Kai Huang, Huey Eng Chua, Sourav S Bhowmick, Byron Choi, Shuigeng Zhou. 2021. MIDAS: Towards Efficient and Effective Maintenance of Canned Patterns in Visual Graph Query Interfaces. In Proceedings of the 2021 International Conference on Management of Data (SIGMOD '21), June 20--June 25, 2021, Virtual Event, China. ACM, New York, NY, USA, 13 pages. https://doi.org/10.1145/3448016.3457251
}}

\vspace{-1.5ex}
\section{Introduction}\label{sec:intro}
A recent survey~\cite{SM+17} revealed that graph query languages and usability are considered as some of the top challenges for graph processing. A common starting point for addressing these challenges is the deployment of a visual query interface (a.k.a \textsc{gui}) that can enable an end user to draw a graph query interactively by utilizing \emph{direct-manipulation}~\cite{SP} and visualize the result matches effectively~\cite{PH+17,bloom}. A useful component of such a \textsc{gui} is a panel containing a set of \textit{canned patterns} (\ie small subgraphs) which is beneficial to visual querying in at least three possible ways~\cite{bhowmick2016,catapult,midas}. First, it can potentially decrease the time taken to visually construct a query by facilitating  \textit{pattern-at-a-time} query mode (\ie construct multiple nodes and edges by performing a \emph{single} click-and-drag action) in lieu of \textit{edge-at-a-time} mode. Second, it can facilitate ``bottom-up'' search when a user does not have upfront knowledge of what to search for. Third, canned patterns (patterns for brevity) may alleviate user frustration of repeated edge construction  especially for larger queries.

\begin{figure}[t]
	\centering
	\includegraphics[width=0.8\linewidth]{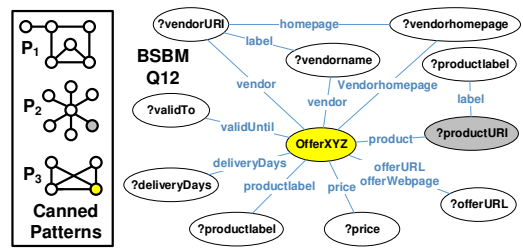}
	\vspace{-1ex}\caption{$Q_{12}$ in BSBM and canned patterns.}\label{fig:query}
	\vspace{-1ex}\end{figure}

\vspace{-1ex}\begin{example}\label{eg:kidneyDisease}
	Consider the real-world subgraph query in Figure~\ref{fig:query} from \textsc{bsbm}~\cite{bsbm} (Query $Q_{12}$). Suppose Wei, a non-programmer, wishes to formulate it using a \textsc{gui} containing a set of canned patterns (a subset of patterns is shown). Specifically, he may drag and drop $p_2$ and $p_3$ on the \textit{Query Canvas}, merge the yellow vertex of $p_3$ with the center vertex of $p_2$, add a vertex and connect it with the grey vertex of $p_2$. Finally, Wei can assign appropriate vertex labels. Observe that it requires five steps to construct the topology. On the other hand, if Wei takes an edge-at-a-time approach to construct the query, it would require 23 steps. Clearly, canned patterns enable more efficient (\ie fewer number of steps or lesser time) formulation of the query.\eat{ Note that canned patterns are \textit{editable} on the \textsc{gui}.}
	
	It is worth noting that Wei may not necessarily have the complete query structure ``in his head'' during query formulation. He may find $p_3$ interesting while browsing the pattern set, which may initiate his bottom-up search leading to the query. Clearly, without the existence of a pattern set, such bottom-up search would be infeasible in practice.
	\EndOfProof\end{example}

\textit{Data-driven} selection of relevant canned patterns for a \textsc{gui} (\eg $p_1$, $p_2$, $p_3$ in Fig.~\ref{fig:query}) is important to facilitate efficient query formulation~\cite{catapult,bhowmick2016}.\eat{ Yet until recently, very little attention has been given to this problem.} In particular, data-driven selection paves the way for \textit{plug-and-play} visual graph query interfaces, which are like a plug-and-play device that can be plugged into any kind of socket (\ie graph data) and used. A plug-and-play \textsc{gui} is dynamically built from a high-level specification of canned pattern properties known as the\textit{ plug} (detailed in Section~\ref{sec:background}). Specifically, given a network $G$ and a plug $b$, the \textsc{gui} is automatically constructed by populating its various components (\eg node/edge attributes, canned patterns) from $G$ without the need for manual \textsc{gui} coding. This enhances portability and maintainability of \textsc{gui}s across different data sources~\cite{bhowmick2016}.\eat{ An example of such plug-and-play \textsc{gui} is the recent \textsc{aurora} system~\cite{kai2020}.}

\eat{Manual selection is not only labour-intensive but the selected patterns may not be \textit{diverse} enough to expedite formulation of a wide range of subgraph queries~\cite{bhowmick2016} or bottom-up exploration of underlying data. Yet until recently, very little attention has been given to this problem. \textsc{Catapult}~\cite{catapult,zhang2015davinci,kai2020} was the first effort that systematically selected canned patterns in a \textit{data-driven} manner. However, \textit{these efforts focused on the collection of small- or medium-sized data graphs instead of large networks}. In particular, as we shall see in Section~\ref{sec:expt}, the clustering-based approach adopted by \textsc{Catapult} is prohibitively expensive for large networks. Furthermore, this approach does not exploit characteristics of real-world subgraph queries for selecting canned patterns.}

In this paper, we present a novel framework called \textsc{Tattoo} (da\underline{T}a-driven c\underline{A}nned pa\underline{T}tern selec\underline{T}i\underline{O}n from netw\underline{O}rks) that takes a data-driven approach to the \textit{canned pattern selection} (\textsc{cps}) problem for \emph{large networks}. Given a network $G$, a user-specified plug specification $b$ which is the number of canned patterns to display and their minimum and maximum permissible sizes\eat{(\ie minimum and maximum permissible sizes of canned patterns and number of patterns to display}, \textsc{Tattoo} automatically selects canned patterns from $G$ that satisfy $b$.

The \textsc{cps} problem is technically challenging. First, it is a NP-hard problem~\cite{catapult}. Second, the availability of query logs can facilitate the selection of relevant patterns as they provide rich information of past queries. In practice, however, such information is often publicly unavailable (\eg none of the networks in \textsc{snap}~\cite{snap} reveal query logs) due to privacy and legal reasons. Hence, we cannot realistically assume the availability of query logs to select patterns. Furthermore, users may demand a \textsc{gui} \textit{prior} to querying a network. Hence, there may not exist any query log prior to the creation of a visual query interface.  Third, it is paramount to find \textit{unlabeled} patterns (\eg Example~\ref{eg:kidneyDisease}) that are potentially useful for query formulation (detailed in Sec.~\ref{sec:problem}). However, the selection of such patterns is challenging as there is an exponential number of them in a large network. Fourth, these selected patterns should not only be \textit{topologically diverse} so that they are useful for a wide variety of queries but they should also impose low \textit{cognitive load}  (\ie mental load to visually interpret a pattern's edge relationships to determine if it is useful for a query) on users. In particular, large graphs overload the human perception and cognitive systems, resulting in poor performance of tasks such as identifying edge relationships~\cite{huang2009,YA+18}.

At this point, a keen reader may wonder why building blocks of real-world networks (\eg paths of length $k$, triangle patterns)~\cite{wang2003,milo2002} cannot be simply utilized as canned patterns since they have high coverage and low cognitive load. However, it may take a larger number of steps to formulate a variety of queries using these patterns due to their small size. For instance, reconsider Example~\ref{eg:kidneyDisease}. Suppose the pattern set consists of an edge, a path of length 2 (\ie 2-path), a triangle, and a rectangle. In this case, $Q_{12}$ may be formulated by dragging and dropping the rectangle once, the 2-path three times, construction of a single node and two edges, along with three node mergers. That is, it takes 10 steps altogether, which is more than using the patterns in Figure~\ref{fig:query}. Furthermore, these patterns do not expose ``interesting'' substructures to facilitate bottom-up search as they occur in almost all large real-world networks.\eat{ For the same reason, \textit{graphlets}~\cite{hocevar2014} cannot be utilized as state-of-the-art graphlet discovery techniques generally limit the size to 5 vertices\eat{\footnote{\scriptsize There are 2 3-node graphlets, 6 4-node graphlets and 21 5-node graphlets.}}~\cite{hocevar2014,ahmed2015}\eat{ as the number of graphlets increases exponentially with the number of vertices}. \eat{ In fact, if we wish to have a canned pattern with 6 vertices (\eg $P_1$ in Example~\ref{eg:kidneyDisease}), then we have to select relevant ones from 112 6-vertex graphlets.} Similarly, as shown in Section~\ref{sec:expt}, \textit{random} selection of subgraphs from $G$ leads to poorer quality canned patterns. Another possible alternative is to utilize frequent subgraphs~\cite{dhiman2016} as they are likely to achieve high coverage. However, they
	may not have low cognitive load or high diversity.}

\begin{figure}[!t]
	\centering
	\includegraphics[width= 0.9\linewidth]{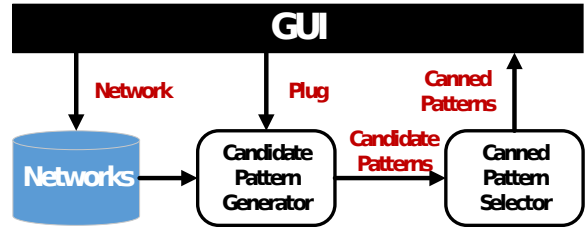}
	\vspace{-3ex}\caption{Overview of \textsc{TATTOO}.}\label{fig:arch}
	\vspace{-2ex} \end{figure}

\textsc{Tattoo} addresses the aforementioned challenges as follows. It exploits a recent analysis of real-world query logs~\cite{Bonifati2017} to \textit{classify} topologies of canned patterns into \textit{categories} that are consistent with the topologies of real-world queries (detailed in Section~\ref{sec:type}). This enables us to reach a middle ground where \textsc{Tattoo} does not need to be restricted by the availability of query logs but yet exploit topological characteristics of real-world queries to guide the selection process. Next, it realizes a novel and efficient \textit{candidate canned pattern generation} technique based on the classified topologies to identify potentially useful patterns. Lastly, canned patterns are \textit{selected} from these candidates for display on the \textsc{gui} based on a novel \textit{pattern set score} that is sensitive to coverage, diversity, and cognitive load of patterns. Specifically, we leverage recent progress in the algorithm community to propose \eat{two selection algorithms that guarantee $\frac{1}{e}$-approximation and $\frac{1}{2}$-approximation~\cite{buchbinder2014,buchbinder2015}, respectively. } a selection algorithm that guarantees $\frac{1}{e}$-approximation~\cite{buchbinder2014}. Figure~\ref{fig:arch} depicts an overview of the \textsc{Tattoo} framework. Experiments with several real-world large networks and users reveal that \textsc{Tattoo} can select canned patterns within few minutes. Importantly, these patterns can reduce the number of steps taken to formulate a subgraph query and query formulation time by up to 9.7X and 18X, respectively, compared to several baseline strategies.\eat{ Interestingly, our study also reveal that across 10 different datasets representing different domains, the number of common patterns is only 32\% when the number of canned patterns on a \textsc{gui} is set to 50. That is, canned patterns overlap only partially across sources and domains, emphasizing that ``one-patternset-fits-all'' approach is ineffective to support efficient visual query formulation.}

In summary, this paper makes the following contributions: (1) We describe \textsc{Tattoo}, an end-to-end canned pattern selection framework for any plug-and-play visual graph query interface for large networks independent of domains and data sources. A video of a plug-and-play interface that incorporates \textsc{Tattoo} can be viewed at \url{https://youtu.be/sL0yHV1eEPw}. (2) We formally introduce the \textit{\textsc{cps} problem for large networks} (Sec.~\ref{sec:problem}) and present a novel categorization of potentially useful canned patterns in Section~\ref{sec:type}. (3) We present an efficient solution to select canned patterns for a \textsc{gui} (Sec.~\ref{sec:candidatePatternGenerator} - \ref{sec:cannedPatternSelector}). Specifically, we present a novel \textit{candidate pattern generation} framework that is grounded on topologies of real-world subgraph queries. Furthermore, for the first time in graph querying literature, we utilize the recent technique in~\cite{buchbinder2014} from the algorithm community to select canned patterns with good theoretical quality guarantees. (4) Using real-world networks, we show the superiority of our proposed framework compared to several baselines (Sec.~\ref{sec:expt}).

\eat{Formal algorithms and selected proofs of theorems and lemmas are provided in~\cite{tech}.}
Proofs of theorems and lemmas are provided in Appendix~\ref{app:proof}.

\vspace{-1ex}
\section{Related Work}\label{sec:relatedWork}
Most germane to our work is our prior efforts on data-driven construction of visual graph query interfaces in~\cite{zhang2015davinci,catapult,kai2020}. The work in~\cite{midas} focuses on the  maintenance of canned patterns for evolving data graphs. Our work differs from these efforts in the following ways. First, we focus on selecting \textit{unlabelled} canned patterns from \emph{large} networks in contrast to labelled patterns from a collection of small- or medium-sized data graphs in~\cite{zhang2015davinci,catapult,kai2020,midas}. Specifically, existing efforts such as \textsc{Catapult}~\cite{catapult} first partitions a collection of data graphs into a set of clusters and summarizes each cluster to a \textit{cluster summary graph} (\textsc{csg}). Then, it selects the canned patterns with the aforementioned characteristics from these \textsc{csg}s using a weighted random walk approach. This clustering-based approach is prohibitively expensive for large networks as detailed in Sec.~\ref{sec:expt}. Second, these approaches do not exploit characteristics of real-world subgraph queries for selecting canned patterns. In contrast, we utilize topological characteristics of real-world queries to guide our solution design.  Third, we present a novel real-world query topology-aware candidate pattern generation technique and a selection technique that provides quality guarantee. No theoretical guarantee is provided in~~\cite{zhang2015davinci,catapult,kai2020} for selecting canned patterns. Lastly, as detailed in Sec.~\ref{sec:cannedPatternSelector}, the computation of \textit{pattern score} to assess the quality of canned patterns is different as the computation of cognitive load and diversity is different from~\cite{catapult} due to the nature of large networks. Furthermore, in this work we provide a theoretical analysis of the pattern score.

Motif discovery techniques~\cite{gurukar2015,milo2002} do not consider diversity and cognitive load. Sizes of these motifs are generally bounded in the range of [3-7] in real applications \cite{gurukar2015,milo2002}\eat{, as opposed to canned patterns which can be larger}. For the same reason, it is difficult to use graphlets \cite{przulj2004,hocevar2014,ahmed2015} as patterns.\eat{ In particular, selection of canned patterns from graphlets generally require graphlet counting which is expensive. Although there are recent approaches \cite{hocevar2014,ahmed2015} developed to improve the efficiency of graphlet counting in large networks, they focus only on graphlets up to size 5 as counting of larger graphlets remains computationally challenging.} Also, frequent subgraphs~\cite{dhiman2016} may not  constitute good canned patterns~\cite{bhowmick2016} and are prohibitively expensive to compute for large networks (detailed in Sec.~\ref{sec:expt}).

\vspace{-1ex}
\section{Background}\label{sec:background}
We first introduce several graph terminologies that we shall be using subsequently. Next, we formally define the notion of \textit{plugs}. Finally, we briefly describe the desirable characteristics of canned patterns as introduced in~\cite{catapult}.

\vspace{-1ex}
\subsection{Terminology}
We denote a graph or network as $G=(V, E)$, where $V$ is a set of nodes/vertices and $E \subseteq V\times V$ is a set of edges. Vertices and edges can have labels as attributes. The \textit{size} of $G$ is defined as $|G|=|E|$. The \textit{degree} of a vertex $v \in V$ is denoted as $deg(v)$.\eat{ We denote the \textit{neighbors} of vertex $v$ as $nb(v)$.} In this paper, we assume that $G$ is an undirected, unweighted graph with labeled vertices.

\eat{Given a graph $G=(V,E)$, a \textit{triangle} is a cycle of length 3 and is denoted as $\triangle_{uvw}$ where $u,v,w\in V$ are the 3 vertices on the cycle. We denote the set of triangles in $G$ as $\triangle_G$ and the \textit{support} of an edge $e=(u,v)\in E$ is defined as $sup(e)=|\{\triangle_{uvw} | \triangle_{uvw}\in\triangle_G\}|$\cite{cohen2008}.}

A \textit{triangle} is a cycle of length 3 in $G$. The \textit{support} of an edge $e=(u,v)\in E$ (denoted by $sup(e)$) is the number of triangles in $G$ containing $u$ and $v$ \cite{voegele2017}. $G_S=(V_S,E_S)$ is a \textit{subgraph} of $G$ (denoted by $G_S\subseteq G$) if $V_S\subseteq V$ and $E_S\subseteq E$. Consider another graph $G^{\prime}=(V^{\prime},E^{\prime})$ where $|V|=|V^{\prime}|$. $G$ and $G^{\prime}$ are \textit{isomorphic} if there exists a bijection $f:V\rightarrow V^{\prime}$ such that $(u,v)\in E$ iff $(f(u),f(v))\in E^{\prime}$. Further, there exists a \textit{subgraph isomorphism} from $G$ to a graph $Q$ if $G$ contains a subgraph $G_S$ that is isomorphic to $Q$. We refer to $G_S$ as the \textit{embedding} of $Q$ in $G$.

Given $G$, the \textit{$k$-truss} of $G$ is the largest subgraph $G^\prime=(V^\prime,E^\prime)$ of $G$ in which every edge $e\in E^\prime$ is contained in at least $k-2$ triangles within the subgraph. A 2-truss is simply $G$ itself. We define the \textit{trussness} of an edge $e$ as $t(e)=\max\{k | e\in E_{T_k}\}$ where $T_k=(V_{T_k},E_{T_k})$ is the $k$-truss in $G$. Further, $k_{max}$ denotes the maximum trussness.

\vspace{-1ex}
\subsection{Plugs}\label{sec:plug}
Recall that data-driven selection of canned patterns facilitates the construction of a plug-and-play visual graph query interface. A \textit{plug} is a high-level specification of the patterns in a \textsc{gui}. Given the specification, \textsc{Tattoo} dynamically generates the canned patterns satisfying it from the underlying network. Formally, it is defined as follows.

\begin{definition}\label{def:plug} \textbf{[Plug]}
	{\em Given a network $G$ and a \textsc{gui} $\mathbb{I}$, a \textbf{plug} $b=(\eta_{min}, \eta_{max},\gamma\eat{,\mathbb{ N}})$ where $\eta_{min} > 2$ (resp. $\eta_{max}$) is the minimum (resp. maximum) size of a pattern, $\gamma > 0$ is the number of patterns to be displayed on $\mathbb{I}$\eat{, and $\mathbb{ N} = \lceil\frac{\gamma}{\eta_{max}-\eta_{min}+1}\rceil$ is the maximum number of allowable patterns for each $k$-sized pattern where $k\in[\eta_{min},\eta_{max}]$}.\/}
\end{definition}
Essentially a plug\footnote{\scriptsize Additional constraint on the distribution of the patterns which is application specific can be included in the plug.} is a collection of attribute-value pairs that specifies the high-level content of a canned pattern panel in a \textsc{gui}. For example, $b=(3, 15, 30\eat{, 3})$ is a plug. Accordingly, the minimum and maximum sizes of patterns in $\mathbb{I}$ are 3 and 15, respectively, and the total number of patterns to be displayed is 30.\eat{ The maximum number of patterns for each valid size is 3.} Observe that there can be multiple plugs for $G$ as well. Similarly, the same plug can be used for different $G$. Hence, different \textsc{gui}s can be constructed by different plug specifications.

A plug should possess the following properties. (a) \textit{Data independence} - A plug should not depend upon a specific network (\ie socket). The specification of plug enables this by not admitting any network-specific information. Observe that this property is important for plug-and-play interfaces as a plug can be used on different network data across different application domains. (b) \textit{Able to select canned patterns with the required specifications} - The resulting canned pattern selection mechanism should select patterns exactly as specified by the plug. 

\vspace{0ex}
\subsection{Characteristics of Canned Patterns}\label{sec:desirableCharacteristics}
Since it is impractical to display a large number of patterns in a visual graph query interface $\mathbb{I}$, the number of patterns should be small and satisfy certain desirable characteristics as introduced in~\cite{catapult}.

\textit{\underline{High coverage}}. A pattern $p \in \mathcal{P}$ \textit{covers} $G$ if $G$ contains a subgraph $s$ that is isomorphic to $p$. Since $p$ may have many embeddings in $G$, the pattern set $\mathcal{P}$ should ideally cover as large portion of $G$ as possible. Then a large number of subgraph queries on $G$ can be constructed by utilizing $\mathcal{P}$.

\textit{\underline{High diversity}}. High coverage of patterns is insufficient to facilitate efficient visual query formulation~\cite{catapult}. In order to make efficient use of the limited display space on $\mathbb{I}$, $\mathcal{P}$ should be \textit{structurally diverse} to serve a variety of queries.\eat{ This ensures that the display space of $\mathbb{I}$ is not wasted by populating with topologically similar and highly redundant patterns.} This also facilitates bottom-up search where a user gets a bird's-eye view of the diverse substructures in $G$.

\textit{\underline{Low cognitive load}}. \textit{Cognitive load} refers to the memory demand or mental effort required to perform a given task \cite{huang2009}. A topologically complex pattern may demand substantial cognitive effort from an end user to  decide if it can aid in her query formulation~\cite{catapult}. Hence, it is desirable for the canned patterns in $\mathcal{P}$ to impose low cognitive load on an end user to make browsing and selecting relevant patterns cognitively efficient during visual query formulation.

\vspace{-1ex}
\section{The CPS Problem}\label{sec:problem}
Given a data graph or network $G=(V,E)$, a visual graph query interface $\mathbb{I}$ and a user-specified plug $b$, the goal of the \textit{canned pattern selection} (\textsc{cps}) problem is to select a set of \textit{unlabelled} patterns $\mathcal{P}$ for display on $\mathbb{I}$, which satisfies the specifications in $b$ and \textit{optimizes} \textit{coverage}, \textit{diversity} and \textit{cognitive load} of $\mathcal{P}$.

Observe that our \textsc{cps} problem differs from~\cite{catapult} in two key ways. First, we focus on a single large network instead of a large collection of small- or medium-sized data graphs. Second, we select \emph{unlabelled} patterns instead of labelled ones. In large networks, a subgraph query may not always contain labels on its vertices or edges. Specifically, unlabelled query graphs are formulated in the \textit{subgraph enumeration} problem~\cite{AFU13} whereas query graphs are labelled in the subgraph matching problem~\cite{SL20}. Hence, by selecting unlabelled patterns \textsc{Tattoo} facilitates visual formulation of both these categories of queries. In particular, one may simply drag-and-drop specific vertex/edge labels from the \textit{Attribute} panel of a \textsc{gui} to add labels to the vertices/edges of a pattern (\eg Example 1).

\eat{We now formally define the \textsc{cps} problem addressed in this paper. We begin by introducing \textit{coverage}, \textit{diversity}, and \textit{cognitive load} of canned patterns. Let $S(p)=\{s_1,\cdots,s_n\}$ be a bag of subgraphs in $G$ isomorphic to $p$ (\ie embeddings of $p$) where vertex labels in $G=(V,E)$ and $p=(V_p,E_p)$ are assumed to be the same and $s_i=(V_i,E_i)$. We say an edge $e\in E_i$ is \textit{covered} by $p$. The \textit{coverage} of $p$ is given as $cov(p)=|\bigcup_{i\in|S(p)|}E_i|/|E|$. Similarly, $cov(\mathcal{P})=|E^{\dag}|/|E|$ where every $e\in E^{\dag}$ is covered by at least one $p\in\mathcal{P}$. Since $|E|$ is constant for a given $G$, \textit{coverage} can be rewritten as $cov(p)=|\bigcup_{i\in|S(p)|}E_i|$ and $cov(\mathcal{P})=|E^{\dag}|$.\eat{The \textit{diversity} $div(p,\mathcal{P}\setminus p)$ of a pattern $p$ with respect to $\mathcal{P}$ can be measured by the minimum number of steps needed to reconstruct $p$ using $\mathcal{P} \setminus p$ (see Section~\ref{sec:cannedPatternSelector}).}
	The \textit{diversity} of $p$ w.r.t to $\mathcal{P}$ is the inverse of \textit{similarity} of $p$.\eat{ Hence, the goal of maximizing diversity of $\mathcal{P}$ is equivalent to minimizing similarity of $\mathcal{P}$.} In particular, the \textit{similarity} of a set of canned patterns $\mathcal{P}$ is denoted as $sim_{(p_i,p_j)\in\mathcal{P}\times\mathcal{P}}(p_i,p_j)$ and is computed using an existing network similarity computation technique (detailed in Section~\ref{sec:cannedPatternSelector}). Finally, we measure \textit{cognitive load} of $p$ (denoted by $cog(p)$) based on the size, density and edge crossings in $p$ (detailed in Section~\ref{sec:cannedPatternSelector}). This follows from the intuition that the cognitive load increases with the density and complexity of a pattern as a user tends to spend more time identifying relationship between different vertices in denser graphs with more edge crossings~\cite{huang2009,huang2010,YA+18}.}

We now formally define the \textsc{cps} problem addressed in this paper. We begin by introducing \textit{coverage}, \textit{diversity}, and \textit{cognitive load} of canned patterns. Let $S(p)=\{s_1,\cdots,s_n\}$ be a bag of subgraphs in $G$ isomorphic to $p$ (\ie embeddings of $p$) where vertex labels in $G=(V,E)$ and $p=(V_p,E_p)$ are assumed to be the same and $s_i=(V_i,E_i)$. We say an edge $e\in E_i$ is \textit{covered} by $p$. The \textit{coverage} of $p$ is given as $cov(p)=|\bigcup_{i\in|S(p)|}E_i|/|E|$. Similarly, $cov(\mathcal{P})=|E^{\dag}|/|E|$ (\ie $f_{cov}(\mathcal{P})$) where every $e\in E^{\dag}$ is covered by at least one $p\in\mathcal{P}$. Since $|E|$ is constant for a given $G$, \textit{coverage} can be rewritten as $cov(p)=|\bigcup_{i\in|S(p)|}E_i|$ and $cov(\mathcal{P})=|E^{\dag}|$.\eat{The \textit{diversity} $div(p,\mathcal{P}\setminus p)$ of a pattern $p$ with respect to $\mathcal{P}$ can be measured by the minimum number of steps needed to reconstruct $p$ using $\mathcal{P} \setminus p$ (see Section~\ref{sec:cannedPatternSelector}).}
The \textit{diversity} of $p$ w.r.t to $\mathcal{P}$ is the inverse of \textit{similarity} of $p$.\eat{ Hence, the goal of maximizing diversity of $\mathcal{P}$ is equivalent to minimizing similarity of $\mathcal{P}$.} In particular, the \textit{similarity} of a set of canned patterns $\mathcal{P}$ is denoted as $f_{sim}(\mathcal{P})=\sum_{(p_i,p_j)\in\mathcal{P}\times\mathcal{P}} sim(p_i,p_j)$ where $sim(p_i,p_j)$ is the similarity between patterns $p_i$ and $p_j$ (detailed in Sec.~\ref{sec:cannedPatternSelector}). Finally, we measure \textit{cognitive load} of $p$ (denoted by $cog(p)$) based on the size, density, and edge crossings in $p$ (detailed in Sec.~\ref{sec:cannedPatternSelector}) as a user tends to spend more time identifying relationships between vertices in denser graphs with more edge crossings~\cite{huang2009,huang2010,YA+18}. The cognitive load of $\mathcal{P}$ (\ie $f_{cog}(\mathcal{P})$) is given as $\sum_{p\in\mathcal{P}}cog(p)$.

\eat{\begin{definition}\label{def:cannedPattern} \textbf{[\textsf{CPS} Problem]}
		{\em Given a network $G$, a \textsc{gui} $\mathbb{I}$, and a plug $b=(\eta_{min}, \eta_{max},\gamma, \mathbb{N})$, the goal of \textbf{canned pattern selection (\textsc{CPS}) problem} is to find a set of unlabelled canned patterns $\mathcal{P}$ from $G$ that satisfies the followings: $\max f_{cov}(\mathcal{P})=cov({\mathcal{P}})$, $\min f_{sim}(\mathcal{P})=\sum_{(p_i,p_j)\in\mathcal{P}\times\mathcal{P}} sim(p_i,p_j)$, and $\min f_{cog}(\mathcal{P})=\sum_{p\in\mathcal{P}}cog(p)$ \textrm{ s.t. }$p \subseteq G$
			where $|\mathcal{P}|=\gamma$.\/}
\end{definition}}

\eat{\begin{definition}\label{def:cannedPattern} \textbf{[\textsf{CPS} Problem]}
		{\em Given a data graph $G$, a \textsc{gui} $\mathbb{I}$, and a plug $b=(\eta_{min}, \eta_{max},\gamma, \mathbb{N})$, the goal of \textbf{canned pattern selection (\textsc{CPS}) problem} is to find a set of unlabelled canned patterns $\mathcal{P}$ from $G$ that satisfies $\max_{f_{cov}}(\mathcal{P})$ and $\min_{i\in\{sim,cog\}}f_i(\mathcal{P})$ where $f_{cov}$ is the coverage of $\mathcal{P}$, $f_{sim}(\mathcal{P})=\sum_{(p_i,p_j)\in\mathcal{P}\times\mathcal{P}} sim(p_i,p_j)$ and $f_{cog}(\mathcal{P})=\sum_{p\in\mathcal{P}}cog(p)$ \textrm{ s.t. }$p\subseteq G$, $|\mathcal{P}|=\gamma$.\/}
\end{definition}}

\vspace{-1ex}\begin{definition}\label{def:cannedPattern} \textbf{[\textsf{CPS} Problem]}
	{\em Given a network $G$, a \textsc{gui} $\mathbb{I}$, and a plug $b=(\eta_{min}, \eta_{max},\gamma\eat{, \mathbb{N}})$, the goal of \textbf{canned pattern selection (\textsc{CPS}) problem} is to find a set of unlabelled canned patterns $\mathcal{P}$ from $G$ that satisfies
		\begin{equation}
			\begin{aligned}
				&\max f_{cov}(\mathcal{P}), -f_{sim}(\mathcal{P}), -f_{cog}(\mathcal{P})\\
				&\textrm{subject to }|\mathcal{P}|=\gamma, \mathcal{P}\in\mathcal{U}
			\end{aligned}
			\vspace{-1ex}\end{equation}
		where $\mathcal{P}$ is the solution; $\mathcal{U}$ is the feasible set of canned pattern sets in $G$; $f_{cov}(\mathcal{P})$, $f_{sim}(\mathcal{P})$ and $f_{cog}(\mathcal{P})$ are the coverage, similarity, and cognitive load of $\mathcal{P}$, respectively.\/}
	\vspace{-1ex}\end{definition}

\textbf{Remark.}
Observe that  \textsc{cps} is a multi-objective optimization problem as our goal is to maximize coverage and diversity (\ie minimize similarity) of canned patterns while minimizing their cognitive load.  Hence, we address it by converting \textsc{cps} into a single-objective optimization problem using a \textit{pattern score} (detailed in Section~\ref{sec:cannedPatternSelector}). Also, observe that we aim to find patterns of size greater than 2 (\ie $\eta_{min} > 2$). Small-size patterns that are basic building blocks of networks~\cite{wang2003,milo2002}  (\eg edge, 2-path, triangle) are provided by \textit{default} for \emph{all} datasets (\ie \textit{default} \textit{patterns}). \eat{\textsc{cps} is a multi-objective optimization problem (\textsc{moop}). In \textsc{moop}, each objective is either minimized or maximized~\cite{oliveira2010}. Also, infinite number of Pareto optimal solutions may exist and it is hard to determine a single suitable solution. The weighted sum method is commonly used to address this \cite{marler2010}. However, since the relative trade-offs between the objectives are not known \textit{a priori}, we shall define a multiplicative score function~\cite{Tofallis2014} (Section~\ref{sec:cannedPatternSelector}, Eq.~\ref{eq:cannedPatternScore}) that is then optimized.}

The \textsc{cps} problem is shown to be NP-hard in~\cite{catapult} by reducing it from the classical maximum coverage problem.

\begin{theorem}\label{thm:cannedPattern}
		\textit{The \textsc{cps} problem is NP-hard.}
\end{theorem}

\eat{\begin{proof} (Sketch).
		The \textsc{cps} is a multi-objective optimization problem which can be reformulated as a constrained single-objective optimization problem where the objective function is $\max f_{cov}$ and the constraints are $\min (f_{sim},f_{cog})$. This reformulated problem (\ie $\max f_{cov}$) can be reduced from the maximum coverage problem, which is a classical NP-hard optimization problem \cite{karp1972}. In particular, given a number $k$ and a collection of sets $S$, the maximum coverage problem aims to find a set $S^{\prime}\subset S$ such that $|S^{\prime}|\leq k$ and the number of covered elements is maximized. In \textsc{cps}, the collection of sets $S$ is the set that consists of all possible subgraphs of the graph dataset $D$. The subset $S^{\prime}$ is the canned pattern set and $k$ is the size of the canned pattern set. The number of covered elements corresponds to the number of covered subgraphs in $D$. Note that the reformulated optimization problem is at least as hard as the maximum coverage problem since optimizing the objective may result in solutions that are sub-optimal with regards to additional imposed constraints.
\end{proof}}

\eat{\begin{proof}
		\vspace{1ex}\noindent\textbf{Proof of Theorem~\ref{thm:cannedPattern} (Sketch).} The \textsc{cps} is a multi-objective optimization problem which can be reformulated as a constrained single-objective optimization problem where the objective function is $\max f_{cov}$ and the constraints are $\min (f_{sim},f_{cog})$. This reformulated problem (\ie $\max f_{cov}$) can be reduced from the maximum coverage problem, which is a classical NP-hard problem \cite{karp1972}. In particular, given a number $k$ and a collection of sets $S$, the maximum coverage problem aims to find a set $S^{\prime}\subset S$ such that $|S^{\prime}|\leq k$ and the number of covered elements is maximized. In \textsc{cps}, the collection of sets $S$ is the set that consists of all possible subgraphs of the graph dataset $D$. The subset $S^{\prime}$ is the canned pattern set and $k$ is the size of the canned pattern set. The number of covered elements corresponds to the number of covered subgraphs in $D$. Note that the reformulated optimization problem is at least as hard as the maximum coverage problem since optimizing the objective may result in solutions that are sub-optimal with regards to additional imposed constraints.
\end{proof}}

\vspace{-1ex}
\section{Categories of Canned Patterns}\label{sec:type}
In theory, numerous different patterns can be selected from a given network. Which of these are ``useful'' for subgraph query formulation in practice? In this section, we provide an answer to this question.

\vspace{-1ex}
\subsection{Topologies of Real-world Queries} \label{sec:topreal}
Although basic building blocks of networks~\cite{wang2003,milo2002} are presented as default patterns in our \textsc{gui}, as remarked earlier, they are insufficient as they do not expose to a user more domain-specific and larger patterns in the underlying data. Such larger substructures not only facilitate more efficient construction of subgraph queries but also guide users for bottom-up search by exposing substructures that are network-specific.\eat{ Note that patterns such as wedge, triangle, and rectangle do not effectively trigger bottom-up search as they appear in almost all real-world networks.} However, which topologies of these substructures should be considered for canned patterns?

\begin{figure}[!t]
	\centering
	\includegraphics[width=0.7\linewidth]{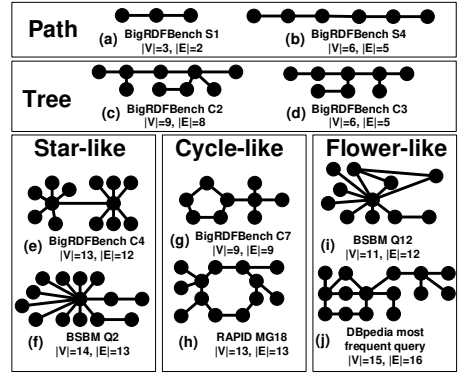}
	\vspace{-2ex}\caption{Examples of real-world query topologies.}\label{fig:topology}
	\vspace{-2ex}\end{figure}

Ideally, real-world subgraph query logs can provide guidance to resolve this challenge. However, as remarked in Section~\ref{sec:intro}, such data may be unavailable. Hence, we leverage results from a recent  study~\cite{Bonifati2017} that analysed a large volume of real-world \textsc{sparql} query logs. It revealed that topologies of many real-world subgraph queries map to chains, trees, stars, cycles, petals, and flowers\footnote{\scriptsize A \textit{petal} is a graph consisting of a source node $s$, target node $t$ and a set of at least 2 node-disjoint paths from $s$ to $t$. A \textit{flower} is a graph consisting of a node $x$ with three types of attachments: chains (stamens), trees that are not chains (the stems), and petals. A \textit{flower set} is a graph in which every connected component is a flower.}~\cite{Bonifati2017}. Figure~\ref{fig:topology} depicts examples of these topologies in real-world subgraph queries extracted from BigRDFBench~\cite{saleem}, BSBM~\cite{bsbm}, Rapid~\cite{rapid}, and DBPedia~\cite{dbpedia}. Consequently, canned patterns in any \textsc{gui} should facilitate efficient construction of these topologies.

\vspace{-1ex}
\subsection{Topologies of Canned Patterns} \label{sec:pattop}
We consider the following types of topological structures of canned patterns in order to facilitate construction of the above query substructures.

\textbf{Path and cycle patterns}. A subgraph query may contain paths of different lengths (\ie chain) and/or cycles.  Figure~\ref{fig:topology} depicts some examples. Hence, our canned patterns should expose representative \textit{$k$-paths} and \textit{$k$-cycles} in the underlying data. Given a graph $G=(V,E)$, a \textit{$k$-path}, denoted as $P_{k}=(V_k,E_k)$, is a walk of length $k$ containing a sequence of vertices $v_1,v_2,\cdots,v_k,v_{k+1}$ where $E_k\subseteq E$, $V_k\subseteq V$ such that all vertices in $V_k$ are distinct. A \textit{$k$-cycle} is simply a closed $(k-1)$-path where $k\geq 3$.

\textbf{Star and asterism patterns}. Intuitively, a \textit{star} is a connected subgraph containing a vertex $r$ where the remaining vertices are connected only to $r$ (\ie neighbors of $r$). A \textit{$k$-star} is a single-level, rooted tree $S_k=(V,E)$ where $V=\{r\}\bigcup L$, $r$ is the root vertex and $L$ is the set of leaves such that $\forall e=\{u,v\}\in E$, $u=r$, $v\in L$ and $|V|=k+1$. We refer to the root as the \textit{center vertex}. Note that $k\geq\epsilon$ where $\epsilon$ is the minimum value of $k$ for which the single-level rooted tree is considered a star.

Real-world queries may contain multiple $k$-stars that are \textit{combined} together. For instance, the query topology in Figure~\ref{fig:topology}(e) is a combination of 6-star and 7-star by merging on a pair of edges. Hence, our canned pattern topology also involves stars that form an \textit{asterism} pattern by \textit{merging} them on a pair of edges. Formally, given $n$ stars $S=\{S_{k_1},\cdots,S_{k_n}\}$ and $n-1$ merged edges $E_m=\{e_{m_1},\cdots,e_{m_{n-1}}\}$ where $S_{k_i}=(V_i,E_i)$ and $e_{m_i}\in E_i$, let $R=\{r_1,\cdots,r_n\}$ be the center vertices such that $r_i\in V_i$. The \textit{asterism} pattern of $S$ is defined as $A_S=(V,E)$ where $e_i=(r_i,v_i)$, $e_{i+1}=(r_{i+1},v_{i+1})$, $E=\bigcup_{1\leq i<n}(\{(r_i,r_{i+1})\}\bigcup(E_i\setminus\{e_i\})\bigcup(E_{i+1}\setminus\{e_{i+1}\}$)), $V=  \bigcup_{1\leq i<n}((V_i\setminus\{v_i\})\bigcup(V_{i+1}\setminus\{v_{i+1}\}))$, $k_i\geq\epsilon$ and $|E|\leq\eta_{max}$.

\begin{figure}[!t]
	\centering
	\includegraphics[width=1\linewidth, height=4cm]{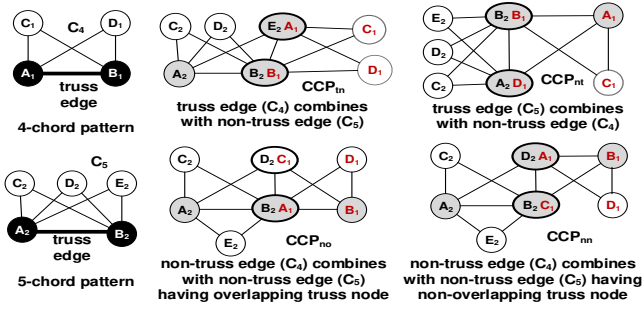}
	\vspace{-4ex}\caption{$k$-chord and composite chord patterns. Grey nodes are truss nodes and oval-shaped nodes are combined nodes.}\label{fig:trussPattern}
	\vspace{-3ex}
\end{figure}

\textbf{$k$-chord and composite chord patterns}. Observe that tree-structured queries can be constructed by combining chains and stars (\eg Figure~\ref{fig:topology}(c)-(d)). However, they are insufficient to construct more complex petal and flower queries efficiently.\eat{ Clearly, since a large network may have numerous distinct petal and flower substructures, it is impractical to expose all of them on a \textsc{gui}.  Hence, we need ``higher-order'' patterns that can facilitate construction of a broad range of these substructures.} In particular, petal and flower queries may often contain \textit{triangle-like} structures. For example, the query in Figure~\ref{fig:topology}(i) contains two triangles. Hence, at first glance it may seem that we can simply select different $k$-trusses ($k>2$) of sizes within the plug specification $b$ as canned patterns. However, a subgraph query may not necessarily always contain $k$-trusses. For instance, the query in Figure~\ref{fig:topology}(j) contains multiple ``triangle-like'' structures as some common edges of triangles are missing. Consequently, the use of only $k$-truss as a canned pattern may make query formulation inefficient as it demands deletion of multiple edges in order to construct a triangle-like query topology. This increases the number of steps required to formulate a query, thereby increase the formulation time. Hence, it is desirable to have ``$k$-truss-like'' substructures as patterns\eat{ to facilitate formulation of petal and flower queries}.\eat{ As we shall see in Section~\ref{sec:expt}, canned patterns based on such ``$k$-truss-like'' substructures provide higher coverage and diversity and lower cognitive load compared to $k$-trusses.}

\begin{table}[!t]\caption{No. of steps for constructing queries.}
	\scriptsize
	\centering
	\vspace{0ex}\begin{tabular}{| p{0.2cm} | p{0.8cm} | p{2.2cm} | p{3.8cm} |}
		\hline
		\textbf{ID} & \textbf{Edge-at-a-time} & \textbf{Default patterns} & \textbf{Canned patterns} \\
		\hline
		(c) & 17 & 6 [2 2-path + 1 square - 1 & 5 [4-path + 2 2-path + 2 merge]\\
		& & edge + 1 edge + 1 merge] & 5 [4-star + 1 2-path + 1 node + 2 edge]\\
		\hline
		(e) & 25 & 11 [5 2-path + 1 node + 2& 1 [A$_{6,7}$]\\
		& &  edge + 3 merge] & 3 [5-star + 6-star + 1 edge]\\
		\hline
		(g) & 18 & 8 [4 2-path + 1 edge + 3& 3 [5-cycle + 4-star + 1 merge]\\
		& &  merge] & 4 [6-path + 2-path + 1 edge + 1 merge]\\
		\hline
		(i) & 23 & 10 [square + 3 2-path + 1 & 5 [4-CP + 6-star + 1 node + 1 edge + 1 merge]\\
		& &  node + 2 edge + 3 merge]& 5 [CCP$_{no}$(4,4) - 2 edge + 5-star + 1 merge]  \\
		\hline
	\end{tabular}\label{tab:step}
	\vspace{-2ex}
\end{table}

To this end, we extract two types of $k$-truss-based structures as canned patterns, namely, \textit{k-chord patterns} ($k$-\textsc{cp})\eat{\footnote{\scriptsize We refer to it as a chord as the truss edge is similar to the horizontal lower chord of a truss used in architecture and structural engineering.}} and \textit{composite chord patterns} (\textsc{ccp}). Intuitively, a $k$-\textsc{cp} is a connected graph containing a \textit{truss edge} $e$ (\ie edge belonging to a $k$-truss) and $k$-2 triangles of $e$. Formally, given a $k$-truss $G_k=(V_k,E_k)$ for $k>2$, the \textit{$k$-chord pattern} ($k$-\textsc{cp}) $C_k=(V_{ck},E_{ck})$ associated with every edge $e=(u,v)\in E_k$ where $u,v\in V_k$ is defined as $V_{ck}=\{u,v\}\bigcup V_{ck}^{\prime}$ and $E_{ck}=\{(u,v)\}\bigcup E_{ck}^{\prime}$ where $V_{ck}^{\prime}=\{w_i: 0\leq i\leq k-2\}$ and $E_{ck}^{\prime}=\{(u,w_i),(w_i,v):0\leq i\leq k-2\}$. $k$-\textsc{cp} can be considered as a building block of $k$-trusses since it is found with respect to each edge in a given $k$-truss. Examples of $k$-\textsc{cp}s (4-\textsc{cp} and 5-\textsc{cp}) are illustrated in Figure~\ref{fig:trussPattern}. We refer to the edge in a $k$-chord pattern that is involved in ($k$-2) triangles as a \textit{truss edge} and the remaining edges as \textit{non-truss edges}. For example, in Figure~\ref{fig:trussPattern}, edges $(A_1,B_1)$ and $(A_2,B_2)$ are truss edges whereas $(A_1,C_1)$ and $(B_2,D_2)$ are non-truss edges. Correspondingly, vertices of a truss edge (\eg $A_1$, $B_1$, $A_2$, $B_2$) are referred to as \textit{truss vertices}. Observe that we can formulate a simple petal query in two steps by selecting the 4-\textsc{cp} pattern and deleting the truss edge.

To select larger canned patterns with greater structural diversity, we \textit{combine} $k$-\textsc{cp}s to yield additional \textit{composite chord patterns} (\textsc{ccp}) that occur in the underlying network. Observe that combining a set of $k$-\textsc{cp}s in different ways results in different patterns as demonstrated in Figure~\ref{fig:trussPattern}. However, this is an overkill as they are not only expensive to compute but also may generate patterns with higher density (higher cognitive load) or are larger than $\eta_{max}$. Hence, we focus on the \textsc{ccp} generated by merging a \emph{single} edge of two $k$-\textsc{cp}s as it not only reduces the complexity of \textsc{ccp} generation, but also produces \textsc{ccp}s with lower density\eat{ due to smaller density}.\eat{ To illustrate, consider the generation of a \textsc{ccp} from $C_4$ and $C_5$. In the case of a single-edge merge, the density of \textsc{ccp} is $2\frac{5+7-1}{(4+5-2)(4+5-2-1)}=0.52$ whereas it is $2\frac{5+7-2}{(4+5-3)(4+5-3-1)}=0.67$ when two-edge merge is used.}\eat{ Furthermore, based on Lemma~\ref{lem:ccpMinSize} the size of a \textsc{ccp} is expected to at least 7.}

\eat{We illustrate with an example on how \textsc{ccp}s may facilitate efficient formulation of flower queries. Consider the flower topology of a query in Figure~\ref{fig:flower} (middle). Edge-at-a-time construction requires 15 steps. In contrast, using pattern-at-time approach, with a \textsc{ccp} and deleting 3 edges (highlighted in red) from it, takes a total of 4 steps. The right figure shows its formulation using four $2$-paths, which consumes 9 steps (4 steps for the $2$-paths and 5 steps to connect the common nodes). Alternatively, it can also be formulated using a rectangle default pattern and two $2$-paths (6 steps).}

\textbf{Unique small graph patterns}. Lastly, we find small connected subgraphs that do not fall under above categories but occur multiple times in the underlying network.

Table~\ref{tab:step} reports the number of steps taken by various modes of query construction of selected query topologies in Fig.~\ref{fig:topology}. Observe that query construction using canned patterns often takes fewer number of steps compared to construction using only default patterns, emphasizing the need for patterns beyond the default ones. One can also formulate a specific query following multiple alternatives, \ie using multiple sets of patterns (canned and default). This gives users the flexibility to formulate a query using these patterns in many ways, all of which often take fewer steps compared to the edge-at-a-time or default pattern-based modes.

\eat{\setlength{\textfloatsep}{2pt}
	\begin{algorithm}[t]
		\algsetup{
			linenosize=\scriptsize
		}
		\begin{algorithmic}[1]
			\scriptsize
			\REQUIRE Data graph $G$, plug $b = (\eta_{min}, \eta_{max}, \gamma, \mathbb{N})$;
			\ENSURE Canned pattern set $\mathcal{P}$;
			
			\STATE $G_T,G_O\leftarrow \textsc{GraphDecomposition}(G)$\label{line:graphDecomposition} \;
			
			\STATE $P_{cp},freq(P_{cp})\leftarrow \textsc{GenChordPatterns}(G_T,t(e))$\label{line:findTrussPatternsFindSimplePatterns} \;
			
			\STATE $P_{ccp},freq(P_{ccp})\leftarrow \textsc{GenCombChordPatterns}(G_T,t(e))$\label{line:findTrussPatternsFindCombinedPatterns} \;
			
			\STATE $P_s,freq(P_s)\leftarrow \textsc{GenStarPatterns}(G_O)$\label{line:findAuxPatternsFindStarLikePatterns} \eat{/*Alg.~\ref{alg:findStarPatterns}*/}\;
			
			\STATE $G_R\leftarrow \textsc{RemoveStarPatternEdges}(G_O,P_s)$\label{line:removeStarPatternEdges}\;
			
			\STATE $P_r,freq(P_r)\leftarrow \textsc{GenSmallPatterns}(G_R,b)$\label{line:findAuxPatternsFindSmallPatterns} \eat{/*Detailed in \cite{tech}*/}\;
			
			\STATE $P_{cand}\leftarrow P_{cp}\bigcup P_{ccp}\bigcup P_s\bigcup P_r$\;\label{line:combinePatterns}
			
			\STATE $freq(P_{cand})\leftarrow freq(P_{cp})\bigcup freq(P_{ccp})\bigcup freq(P_s)\bigcup freq(P_r)$\;\label{line:combineCoverage}
			
			\STATE $\mathcal{P}\leftarrow \textsc{SelectCannedPatterns}(P_{cand},freq(P_{cand}),b)$\label{line:findPatterns} \;
			
		\end{algorithmic}
		\caption{\textsc{Tattoo} framework.} \label{alg:avatar}
\end{algorithm}}

\eat{\begin{algorithm}[t]
		\algsetup{
			linenosize=\small
		}
		\begin{algorithmic}[1]
			\small
			\REQUIRE Data graph $G=(V,E)$, maximum pattern size $\eta_{max}$;
			\ENSURE \textsc{tir} graph $G_T=(V_T,E_T)$, \textsc{tsr} graph $G_O=(V_O,E_O)$, trussness of all edges $T(e)=\{t(e_i)|e_i\in E\}$;
			
			\STATE $k\leftarrow 2$\;
			
			\STATE $G_T\leftarrow\phi$\;
			
			\STATE $G_O\leftarrow\phi$\;
			
			\STATE compute $\textsc{sup}(e)$ for each edge\;\label{line:graphDecompositionEdgeSupport}
			
			\STATE sort all edges in ascending order of support\;
			
			\WHILE{$\exists e$ such that $\textsc{sup}(e)\leq(k-2)$}\label{line:graphDecompositionExtractRegionStart}
			
			\STATE $e=(u,v)\leftarrow \textsc{GetEdgeWithLowestSupport}(E)$\;
			
			\FOR{each $w\in nb(u)\setminus\{v\}$}
			
			\IF{$(v,w)\in E$}
			
			\STATE $\textsc{sup}((u,w))\leftarrow sup((u,w))-1\;$\label{line:graphDecompositionUpdateEdgeSupportStart}
			
			\STATE $\textsc{sup}((v,w))\leftarrow sup((v,w))-1\;$\label{line:graphDecompositionUpdateEdgeSupportEnd}
			
			\STATE reorder $(u,w)$ and $(v,w)$ according to new support\;
			
			\ENDIF
			
			\ENDFOR
			
			\STATE $t(e)\leftarrow \textsc{min}(k,\eta_{max})$\;\label{line:graphDecompositionTrussness}
			
			\IF{$k=2$}
			
			\STATE update $G_O$ with $u$, $v$ and $(u,v)$\;\label{line:graphDecompositionGAux}
			
			\ELSE
			
			\STATE update $G_T$ with $u$, $v$ and $(u,v)$\;\label{line:graphDecompositionGTruss}
			
			\ENDIF
			
			\STATE remove $(u,v)$ from $G$\;
			
			\ENDWHILE
			
			\IF{not all edges in $G$ are removed}
			
			\STATE $k\leftarrow(k+1)$\;
			
			\STATE go to line 6\;
			
			\ENDIF\label{line:graphDecompositionExtractRegionEnd}
			
		\end{algorithmic}
		\caption{$\textsc{GraphDecomposition}$.} \label{alg:graphDecomposition}
\end{algorithm}}

\eat{\begin{figure*}[!t]
		\centering
		\includegraphics[width=7in]{graphDecompositionSteps.eps}
		\vspace{-2ex}\caption{Graph decomposition. Numbers annotating the edges correspond to the support of the edges. Edges whose trussness are updated are highlighted in red. Figure is best viewed in colour.}\label{fig:graphDecompositionStep}
\end{figure*}}

\eat{\textbf{Remark.} Note that the efficiency of Algorithm~\ref{alg:avatar} can be further improved by using multi-threading. In particular, the following are executed using 6 threads in our implementation: edge support computation (Algorithm~\ref{alg:graphDecomposition}) and computation of candidate pattern scores (Algorithm~\ref{alg:patternMining}).}
\vspace{-2ex}
\section{Candidate Patterns Generation}\label{sec:candidatePatternGenerator}
In the preceding section, we classified the topologies of canned patterns broadly into ``$k$-truss-like'' and ``non-$k$-truss-like'' structures. In this section, we describe how candidate canned patterns conforming to these topological categories are extracted from the underlying network $G$. \eat{To this end, we first \textit{decompose} $G$ into \textit{truss-infested}  and \textit{truss-oblivious regions} and then generate ``$k$-truss-like'' and ``non-$k$-truss-like'' candidate patterns from these regions, respectively. We discuss these two steps in turn.}

\begin{algorithm}[t]
	\algsetup{
		linenosize=\scriptsize
	}
	\begin{algorithmic}[1]
		\scriptsize
		\REQUIRE Data graph $G$, plug $b = (\eta_{min}, \eta_{max}, \gamma\eat{, \mathbb{N}})$;
		\ENSURE Canned pattern set $\mathcal{P}$;
		
		\STATE $G_T,G_O\leftarrow \textsc{GraphDecomposition}(G)$\label{line:graphDecomposition} \;
		
		\STATE $P_{cp},freq(P_{cp})\leftarrow \textsc{GenChordPatterns}(G_T,t(e))$\label{line:findTrussPatternsFindSimplePatterns} /*Alg.~\ref{alg:findSimpleTrussPatterns}*/\;
		
		\STATE $P_{ccp},freq(P_{ccp})\leftarrow \textsc{GenCombChordPatterns}(G_T,t(e))$\label{line:findTrussPatternsFindCombinedPatterns} /*Alg.~\ref{alg:findCombinedTrussPatterns}*/ \;
		
		\STATE $P_s,freq(P_s)\leftarrow \textsc{GenStarPatterns}(G_O)$\label{line:findAuxPatternsFindStarLikePatterns} /*Alg.~\ref{alg:findStarPatterns}*/\;
		
		\STATE $G_R\leftarrow \textsc{RemoveStarPatternEdges}(G_O,P_s)$\label{line:removeStarPatternEdges}\;
		
		\STATE $P_r,freq(P_r)\leftarrow \textsc{GenSmallPatterns}(G_R,b)$\label{line:findAuxPatternsFindSmallPatterns} /*Alg.~\ref{alg:findSmallPatterns}*/\;
		
		\STATE $P_{cand}\leftarrow P_{cp}\bigcup P_{ccp}\bigcup P_s\bigcup P_r$\;\label{line:combinePatterns}
		
		\STATE $freq(P_{cand})\leftarrow freq(P_{cp})\bigcup freq(P_{ccp})\bigcup freq(P_s)\bigcup freq(P_r)$\;\label{line:combineCoverage}
		
		\STATE $\mathcal{P}\leftarrow \textsc{SelectCannedPatterns}(P_{cand},freq(P_{cand}),b)$\label{line:findPatterns} \;
		
	\end{algorithmic}
	\caption{The \textsc{Tattoo} algorithm.} \label{alg:avatar}
\end{algorithm}

We begin by providing an overview of the \textsc{Tattoo} algorithm. Algorithm~\ref{alg:avatar} outlines the procedure. It first \textit{decomposes} $G$ into \textit{truss-infested}  and \textit{truss-oblivious regions} (Line 1) and then \textit{generates} ``$k$-truss-like'' and ``non-$k$-truss-like'' candidate patterns from these regions, respectively (Lines 2-8). Finally, it \textit{selects} the canned pattern set from these candidate patterns based on the plug specification  (Line 9). We discuss the decomposition of $G$ and candidate pattern generation in turn. In the next section, we shall elaborate on the selection of canned patterns from the candidate patterns.

\begin{table}[!t]\caption{TIR and TOR graphs in real networks.}
	\scriptsize
	\centering
	\vspace{0ex}\begin{tabular}{| l | l | l | l | l | l |}
		\hline
		\textbf{Data} & \textbf{Name} & \textbf{$|V|$} & \textbf{$|E|$} & \textbf{\% ($G_T$)} & \textbf{\% ($G_O$)}\\
		\hline
		$BK$ & loc-Brightkite & 58K & 214K & 67.3 & 32.7\\
		\hline
		$GO$ & loc-Gowalla & 197K & 950K & 78.2 & 21.8\\
		\hline
		$DB$ & com-DBLP & 317K & 1.05M & 93 & 7\\
		\hline
		$AM$ & com-Amazon & 335K & 926K & 77.2 & 22.8\\
		\hline
		$RP$ & RoadNet-PA & 1.09M & 1.54M & 12.7 & 87.3\\
		\hline
		$YT$ & com-Youtube & 1.13M & 2.99M & 46.8 & 53.2\\
		\hline
		$RT$ & RoadNet-TX & 1.38M & 1.92M & 12.5 & 87.5\\
		\hline
		$SK$ & as-Skitter & 1.7M & 11M & 79.1 & 20.9\\
		\hline
		$RC$ & RoadNet-CA & 1.97M & 2.77M & 12.6 & 87.4\\
		\hline
		$LJ$ & com-LiveJournal & 4M & 34.7M & 83.2 & 16.8\\
		\hline
	\end{tabular}\label{tab:graphDecomposition}
	\vspace{-2ex}
\end{table}

\vspace{-1ex}
\subsection{Truss-based Graph Decomposition}\label{sec:graphDecomposition}
\eat{Since real-world query logs contain \textit{triangle-like} and \textit{non-triangle-like} substructures (detailed in Section~\ref{sec:type}), in the \textit{truss-based graph decomposition} phase (Line~\ref{line:graphDecomposition}), the input data graph is decomposed into a dense \textit{truss-infested region} ($G_T$) and a sparse \textit{truss-oblivious region} ($G_O$) by leveraging the notion of $k$-truss. The \textit{candidate pattern generation} phase (Lines~\ref{line:findTrussPatternsFindSimplePatterns}-\ref{line:findAuxPatternsFindSmallPatterns}) efficiently generates candidate patterns from $G_T$ and $G_O$ by finding instances of different \textit{categories} of query topology that appear in real-world query log data~\cite{Bonifati2017}. Lastly, the \textit{canned pattern selection} phase (Line~\ref{line:findPatterns}) selects canned patterns for the \textsc{gui} from these candidates based on Definition~\ref{def:cannedPattern}.\eat{In particular, two approaches are proposed that guarantee $\frac{1}{e}$-approximation and $\frac{1}{2}$-approximation, respectively.}\eat{ In the sequel, we shall justify these steps\eat{ and elaborate on them in turn}.} In particular, an approach that guarantees $\frac{1}{e}$-approximation is proposed.}

In order to extract ``non-$k$-truss-like'' and ``$k$-truss-like'' structures as candidate patterns, we first decompose a network $G$ into \textit{sparse} (containing non-trusses) and \textit{dense} (containing trusses) regions. The latter region is referred to as \textit{truss-infested region} (\textsc{tir} graph) and the former \textit{truss-oblivious region} (\textsc{tor} graph), and are denoted by $G_T$ and $G_O$, respectively. Table~\ref{tab:graphDecomposition} reports the sizes of $G_T$ and $G_O$ in several real-world networks measured as the percentage of the total number of edges. We observe $G_T$ basically consists of relatively large connected subgraphs that comprise multiple $k$-trusses. On the other hand, $G_O$ mainly consists of chains (\ie paths), stars, cycles, and small connected components. Furthermore, although some networks have small $G_O$ (\eg \textsf{com-DBLP}), there are networks where $G_O$ is large (\eg \textsf{RoadNet-CA}), encompassing up to $87.5\%$ of the total number of edges. Consequently, by decomposing a network into $G_T$ and $G_O$, we can improve efficiency\eat{~\cite{tech}} by limiting the search for $k$-truss-like patterns in $G_T$ instead of the entire network and extract non-truss-like patterns from $G_O$. Additionally, generating candidate patterns of aforementioned topological categories from \emph{both} \textsc{tir} and \textsc{tor} graphs enables us to select a \textit{holistic} collection of patterns having higher coverage and diversity. Cognitive load of the pattern set is often reduced when patterns from both regions are considered due to the sparse structure of \textsc{tor}\eat{ in lieu of selecting only from \textsc{tir}}.\eat{ We shall demonstrate this in Section~\ref{sec:expt}.}

\textsc{Tattoo} utilizes the state-of-the-art truss decomposition approach in \cite{wang2012}\eat{\footnote{\scriptsize The choice of truss decomposition technique is orthogonal to our framework. Any superior technique can be used.}} to decompose $G$ into $G_T$ and $G_O$. Briefly, this approach identifies $k$-trusses ($k\in[2-k_{max}]$) in $G$ iteratively by removing edges with support less than $k-2$ from $G$.\eat{ These extracted edges form the $k$-truss where each $k$-truss is stored as an individual graph.} Hence, our graph decomposition algorithm adapts it to assign 2-truss as $G_O$ and the remaining $k$-trusses as $G_T$. \eat{Note that the choice of truss decomposition technique is orthogonal to our framework.}\eat{ Any superior technique can be used.}

\eat{\begin{figure*}[!t]
		\centering
		\includegraphics[width=\linewidth, height=5cm]{Gaux.eps}
		\vspace{-3ex}\caption{Visualization of portions of \textsc{tor} graphs of \textit{HepPh}, \textit{Amazon} and \textit{Skitter} datasets.}\label{fig:Gaux}
		\vspace{-2ex}\end{figure*}}

\eat{Algorithm~\ref{alg:graphDecomposition} outlines the graph decomposition procedure.
	The variable $t(e)$ helps to keep track of the edge trussness in $G_T$. Since the goal is to select canned patterns with maximum size $\eta_{max}$, the upper bound of edge trussness is set to this value. The algorithm first identifies the support of each edges (Line~\ref{line:graphDecompositionEdgeSupport}). Then, regions of the data graph are iteratively extracted by removing the edges with the lowest support, starting from the sparsest (\ie $sup(e)=0$) to the densest (Lines~\ref{line:graphDecompositionExtractRegionStart} to~\ref{line:graphDecompositionExtractRegionEnd}). In particular, \textsc{Tattoo} considers all edges with $sup(e)=0$ as sparse regions and these edges form the \textsc{tor} graph $G_O$. The remaining edges form the \textsc{tir} graph $G_T$.

	In summary, Algorithm~\ref{alg:graphDecomposition} makes the following two simple modifications on \textsc{TrussDecomposition} in \cite{wang2012}: (1) Instead of storing each $k$-truss as a separate graph, Algorithm~\ref{alg:graphDecomposition} stores 2-truss as $G_O$ and the remaining $k$-trusses are combined as a single graph $G_T$, (2) Algorithm~\ref{alg:graphDecomposition} assigns a trussness value $t(e)$ to every edge in $G_T$ and $G_O$.}

We keep track of the edge trussness (denoted as $t(e)$) in $G_T$. Since the goal is to select canned patterns with maximum size $\eta_{max}$, the upper bound of edge trussness is set to this value. The algorithm first identifies the support of each edge. Then, regions of the data graph are iteratively extracted by removing edges with the lowest support, starting from the sparsest (\ie $sup(e)=0$) to the densest. In particular, \textsc{Tattoo} considers all edges with $sup(e)=0$ as sparse regions and these edges form the \textsc{tor} graph $G_O$. The remaining edges form the \textsc{tir} graph $G_T$.

In summary, the above approach makes the following two simple modifications to the truss decomposition technique in~\cite{wang2012}: (1) instead of storing each $k$-truss as a separate graph, it stores 2-truss as $G_O$ and the remaining $k$-trusses are combined as a single graph $G_T$; (2) it assigns a trussness value $t(e)$ to every edge in $G_T$ and $G_O$. The worst-case time and space complexities of this algorithm are $O(|E|^{1.5})$ and $O(|V|+|E|)$, respectively~\cite{wang2012}.

\eat{\begin{example}
		Figure~\ref{fig:graphDecompositionStep} illustrates the first few steps of truss decomposition of the data graph $G$. Observe that when $k=2$, edges with $\textsc{sup}(e)=0$ (\textit{e.g.}, edge (3,7)) are extracted and stored in $G_O$. The extracted edges are assigned trussness of 2. In the updated data graph, edges with $sup(e)=0$ are now removed and the lowest edge support becomes 1. In the next round (\ie $k=3$), edges with support of 1 are removed. However, instead of storing them in $G_O$, these edges form part of $G_T$. Note that when an edge is removed (\textit{e.g.}, edge (15,16)), the support of neighbouring edges (\textit{e.g.}, edge (2,15) and edge (2,16)) may be affected if the removed edge forms a triangle with these neighbouring edges. In order to maintain the consistency of the support of edges in $G$, the edge support of these neighbouring edges are decremented by 1 (Lines~\ref{line:graphDecompositionUpdateEdgeSupportStart} to~\ref{line:graphDecompositionUpdateEdgeSupportEnd}). Graph decomposition continues until no more edges remain in the data graph.
		\EndOfProof
\end{example}}

\eat{\begin{lemma}\label{lem:graphDecompositionComplexity}
		The worst-case time and space complexities of the truss-based graph decomposition phase are $O(|E|^{1.5})$ time and $O(|V|+|E|)$, respectively. \end{lemma}}

\eat{\textbf{Remark.}\eat{ A natural question to ask regarding the graph decomposition strategy is why we do not simply consider $G_T$ and ignore the \textsc{tor} (\ie $G_O$) as it is expected to be small in size. That is, since $G_O$ is small we can focus on selecting canned patterns from $G_T$ only. However, our investigation with several real-world networks shows that $G_O$ is not necessarily small and its size typically varies based on application domains. Table~\ref{tab:graphDecomposition} reports this result. Observe that although some networks have small $G_O$ (\eg \textsf{com-DBLP}), there are networks where $G_O$ is large (\eg \textsf{RoadNet-CA}), encompassing up to $87.7\%$ of the total number of edges. Hence, \textsc{tor} may contain patterns that are useful for query formulation.}
	The goal of our graph decomposition is different from graph partitioning strategies used for distributed graph processing~\cite{AO18}. The latter focuses on generating partitions of nearly equal sizes to support parallel processing which is different from \textsc{Tattoo}'s goal of canned pattern selection. In particular, our technique aims to retain the structural complexity (\textit{e.g.}, subgraph density) of the graphs after decomposition such that the topology of the generated patterns are still applicable to the original graph.}

\vspace{-1ex}
\subsection{Patterns from a TIR Graph}
Next, we generate $k$-\textsc{cp}s and \textsc{ccp}s as candidate patterns from a \textsc{tir} graph. For each pattern we also compute its frequency as it will be used subsequently to measure its coverage. We discuss them in turn. \eat{The formal algorithms are given in~\cite{tech}.}

\textbf{Generation of $k$-chord patterns.} Algorithm~\ref{alg:findSimpleTrussPatterns} describes generation of the $k$-\textsc{cp}s. In particular, we can find $k$-\textsc{cp}s with respect to each edge in a given $k$-truss. For instance, every edge in a 4-truss and a 5-truss is part of at least 2 and 3 triangles, respectively. Observe that the $2$-chord pattern of an edge $e$ is simply the edge itself. Hence, \textsc{Tattoo} generates $k$-\textsc{cp}s for $k\geq 3$. The \textit{frequency} of a $k$-\textsc{cp} is measured by the frequency of the pattern occurring in $G_T$, which is essentially the number of edges having trussness greater than or equals to $k$ (Lines~\ref{line:simpleTrussPatternComputeTrussnessPerEdgeStart} to~\ref{line:simpleTrussPatternComputeTrussnessPerEdgeEnd}). Formally, given a \textsc{tir} graph $G_T=(V_T,E_T)$ and a $k$-chord pattern $C_k=(V_{ck},E_{ck})$, the \textit{frequency} of $C_k$ is defined as $freq(C_k)=|\{e\in E|t(e)\geq k\}|$. Then, the set of $k$-\textsc{cp}s of a $G_T$ is simply the set of patterns $C_k$ whose frequency is greater than 0.\eat{ Note that the edge trussness of all edges in $G_T$ is already computed during the truss-based graph decomposition phase.} We first generate $k$-chord patterns in $G_T$ and then compute their frequencies using edge trussness.

\begin{figure}[!t]
	\centering
	\includegraphics[width=0.9\linewidth, height=1.3cm]{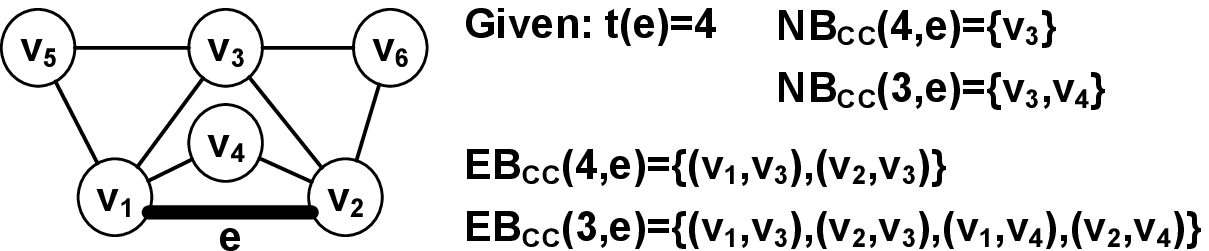}
	\vspace{-2ex}\caption{$k$-CCP node and edge neighbourhoods.}\label{fig:neighbourhoodExample}
	\vspace{-1ex}\end{figure}

\eat{\vspace{-1ex}\begin{lemma}\label{lem:coverageOfSimpleTrussAntiMonotone}
		The frequency of a $k$-chord pattern is anti-monotonic.
\end{lemma}}

\vspace{-1ex}\begin{lemma}\label{lem:complexityFindSimpleTrussPatterns}
	\textit{The worst-case time and space complexities of $k$-\textsc{cp} generation are $O(k_{max}|E_T|^{1.5})$ and $O(|V_T|+|E_T|)$, respectively.}
	\vspace{-1ex}\end{lemma}

\begin{algorithm}[t]
	\algsetup{
		linenosize=\scriptsize
	}
	\begin{algorithmic}[1]
		\scriptsize
		\REQUIRE \textsc{tir} graph $G_T=(V_T, E_T)$, trussness of all edges $T(e)$;
		\ENSURE Set of $k$-chord patterns $P_{cp}=\{C_k|3\leq k\leq k_{max}\}$ and frequency $freq(P_{cp})$;
		
		\FOR{$k=3$\textrm{ to }$k_{max}/* \textrm{\textsf{generate $k$-chord patterns}}*/\;$}\label{line:simpleTrussPatternGetPatternStart}
		
		\STATE $C_k=(V_{ck},E_{ck})\leftarrow\phi$\;
		
		\STATE $V_{ck}\leftarrow\{u,v\}$\;
		
		\STATE $E_{ck}\leftarrow\{(u,v)\}$\;
		
		\STATE $i\leftarrow k$\;
		
		\WHILE{$i\geq 3$}
		
		\STATE $V_{ck}\leftarrow\{w_{i-2}\}$\;
		
		\STATE $E_{ck}\leftarrow\{(u,w_{i-2}),(w_{i-2},v)\}$\;
		
		\STATE $i\leftarrow i-1$\;
		
		\ENDWHILE
		
		\STATE $freq(C_k)\leftarrow 0$\;
		
		\ENDFOR\label{line:simpleTrussPatternGetPatternEnd}
		
		\FOR{\textrm{each }$e\in E_T/* \textrm{\textsf{compute frequencies using edge trussness}}*/\;$}\label{line:simpleTrussPatternComputeTrussnessPerEdgeStart}
		
		\STATE $k\leftarrow t(e)$\;
		
		\WHILE{$k\geq 3$}
		
		\STATE $cov(C_k)\leftarrow freq(C_k)+1$\;\label{line:simpleTrussPatternCoverage}
		
		\STATE $P_{cp}\leftarrow P_{cp}\bigcup C_k\;$\label{line:simpleTrussPatternPattern}
		
		\STATE $k\leftarrow k-1$\;
		
		\ENDWHILE
		
		\ENDFOR\label{line:simpleTrussPatternComputeTrussnessPerEdgeEnd}
		
	\end{algorithmic}
	\caption{$\textsc{GenChordPatterns}$.} \label{alg:findSimpleTrussPatterns}
\end{algorithm}

\textbf{Generation of composite chord patterns.} Next, we generate the \textsc{ccp}s. Specifically, we generate the following categories of \textsc{ccp}s based on different ways of merging truss and non-truss edges.

\vspace{-1ex}\begin{definition} \label{def:ccp}{\em Let $C_{k_1}=(V_{ck_1},E_{ck_1})$ and $C_{k_2}=(V_{ck_2},E_{ck_2})$ be two $k$-chord patterns where $s,t\in V_{ck_1}$ and $u,v\in V_{ck_2}$ are truss vertices. Then, we can generate the following categories of \textbf{composite chord patterns} of $C_{k_1}$ and $C_{k_2}$ by merging $C_{k_1}$ and $C_{k_2}$ as follows:\/}
	\begin{enumerate} \itemsep = -0.5ex {\em
			\item $CCP_{tn}(k_1, k_2)$: merge the truss edge of $C_{k_1}$ with a non-truss edge of $C_{k_2}$.
			\item $CCP_{nt}(k_1, k_2)$: merge the truss edge of $C_{k_2}$ with a non-truss edge of $C_{k_1}$.
			\item $CCP_{no}(k_1, k_2)$: merge a non-truss edge of $C_{k_1}$ with a non-truss edge of $C_{k_2}$ such that there is an overlapping truss vertex.
			\item $CCP_{nn}(k_1, k_2)$: merge a non-truss edge of $C_{k_1}$ with a non-truss edge of $C_{k_2}$ such that there is no overlapping truss vertex.\/}
	\end{enumerate}
\end{definition}
Figure~\ref{fig:trussPattern} depicts examples of these four categories of \textsc{ccp}s. When the context is clear, we shall simply refer to a \textsc{ccp} as $CCP_i$. A keen reader may observe that it is possible to create another \textsc{ccp} by merging the truss edge of $C_{k_1}$ with the truss edge of $C_{k_2}$. However, this \textsc{ccp} is in fact a $k$-\textsc{cp} where $k=k_1+k_2-2$. For instance, when $C_4$ and $C_5$ in Figure~\ref{fig:trussPattern} are merged on their truss edges, the resultant pattern is a $7$-\textsc{cp}. Also, combining two $3$-\textsc{cp}s always yields a $4$-\textsc{cp} (Lemma~\ref{lem:combinedSimple3TrussPattern}). Since $k$-\textsc{cp}s have already been handled earlier, these combinations are ignored.

\eat{\vspace{-1ex}\begin{lemma}\label{lem:ccpMinSize}
		\textit{The minimum size of a \textsc{ccp} is 7.}
\end{lemma}}

\begin{lemma}\label{lem:combinedSimple3TrussPattern}
	\textit{Two $3$-\textsc{cp}s always yield a \textsc{ccp} that is $4$-\textsc{cp}.}
\end{lemma}

\eat{\begin{proof} (Sketch.) The 3-truss pattern $C_3=(V_{c3},E_{c3})$ is a triangle. Then, $\forall e=(u,v)\in E_{c3}$, there is a vertex $w$ that is adjacent to both $u$ and $v$.\eat{ That is, all edges in $C_3$ have similar structure.} Hence, all different types of single edge merger between two $C_3$ produces a pattern with a merged edge $e_m=(x,y)$ and vertices $x$ and $y$ have two common adjacent vertices $w_1$ and $w_2$. This is essentially $C_4$ where its truss edge corresponds to the merged edge of the two $C_3$.
\end{proof}}

We now elaborate on how the \textsc{ccp}s and their \textit{frequencies} are computed in \textsc{Tattoo} efficiently. We shall introduce two terminologies related to \textit{node} and \textit{edge neighbourhoods} of a \textsc{ccp} to facilitate exposition. Given an edge $e=(u,v)$ in a $k$-truss, the \textit{$k^{\prime}$-\textsc{ccp} node neighbourhood} (denoted as $NB_{cc}(k^{\prime}, e)$) of $e$ is a set of vertices $W$ adjacent to $u$ and $v$ such that $\forall w\in W$, $t((u,w))\geq k^{\prime}$ and $t((w,v))\geq k^{\prime}$ where $k^{\prime}\leq k$. The \textit{$k^{\prime}$-\textsc{ccp} edge neighbourhood} (denoted as $EB_{cc}(k^{\prime}, e)$) of $e$ is the set of edges $S$ adjacent to $e$ such that $\forall (u,x_1),(x_2,v)\in S$, $x_1,x_2\in NB_{cc}(k^{\prime}, e)$ where $k^{\prime}\leq k$. Figure~\ref{fig:neighbourhoodExample} illustrates examples of $k^{\prime}$-\textsc{ccp} node and edge neighborhoods. For instance, $NB_{cc}(4, e)$ consists of $v_3$ since $t(v_1,v_3)\geq 4$ and $t(v_2,v_3)\geq 4$.

\begin{lemma}\label{lem:cnn}
	\textit{Given a truss edge $e$, there is at least a $k$-chord pattern $C_k$ on $e$ if $|NB_{cc}(k, e)|\geq (k-2)$.}
\end{lemma}

\eat{\begin{proof} (Sketch). Observe that $k$-chord pattern on an edge $e=(u,v)$ implies that $k$-2 triangles in the graph contain $e$. Since $\textsc{nb}_{cc}(k,e)$ is the set of nodes $W$ adjacent to $u$ and $v$ such that $\forall w\in W$, $t((u,w))\geq k$ and $t((w,v))\geq k$, $|\textsc{nb}_{cc}(k,e)|$ is equivalent to the number of triangles around $e$. Hence, when $|\textsc{nb}_{cc}(k,e)|\geq (k-2)$, a $k$-chord pattern must exist on $e$.
\end{proof}}

\eat{\begin{algorithm}[t]
		
		\algsetup{
			linenosize=\small
		}
		\begin{algorithmic}[1]
			\small
			\REQUIRE Truss graph $G_T$, edge $e=(u,v)$ and its trussness $k$ where $3\leq k\leq k_{max}$;
			\ENSURE $k$-CT node neighbourhood of edge $e$ $CNN_k(e)$;
			
			\STATE $CNN_k(e)\leftarrow\phi\;$
			
			\STATE $W\leftarrow GetAdjacentNodes(G_T,u)\bigcap GetAdjacentNodes(G_T,v)\;$
			
			\FOR{$w\in W$}
			
			\IF{$t((u,w))\geq k$ and $t((w,v))\geq k$}
			
			\STATE $CNN_k(e)\leftarrow CNN_k(e)\bigcup\{w\}\;$
			
			\ENDIF
			
			\ENDFOR
		\end{algorithmic}
		\caption{$GetCNN$.} \label{alg:getCNN}
	\end{algorithm}
	
	\eat{\begin{algorithm}[t]
			\algsetup{
				linenosize=\small
			}
			\begin{algorithmic}[1]
				\small
				\REQUIRE Truss graph $G_T$, edge $e=(u,v)$ and its $k$-CT node neighbour $CNN_k(e)$;
				\ENSURE $k$-CT edge neighbourhood of edge $e$ $CEN_k(e)$;
				
				\STATE $CEN_k(e)\leftarrow\phi\;$
				
				\STATE $S\leftarrow GetAdjacentEdges(G_T,e)\;$
				
				\FOR{$s=(x,y)\in S$}
				
				\IF{($x\in\{u,v\}$ and $y\in CNN_k(e)$) or ($y\in\{u,v\}$ and $x\in CNN_k(e)$)}
				
				\STATE $CEN_k(e)\leftarrow CEN_k(e)\bigcup\{s\}\;$
				
				\ENDIF
				
				\ENDFOR
			\end{algorithmic}
			\caption{$GetCEN$.} \label{alg:getCEN}
	\end{algorithm}}
}

\underline{\textit{Frequencies of $CCP_{tn}(k_1, k_2)$ and $CCP_{nt}(k_1, k_2)$}}. Consider two different $k$-\textsc{cp}s. $CCP_{tn}$ and $CCP_{nt}$ involve merger of a truss edge belonging to one $k$-\textsc{cp} with a non-truss edge belonging to another $k$-\textsc{cp}. Given two $k$-\textsc{cp}s $C_{k_1}$ and $C_{k_2}$, let edges $e_1$ and $e_2$ be the truss edges of $C_{k_1}$ and $C_{k_2}$, respectively. Intuitively, a pattern is a $CCP_{tn}(k_1, k_2)$ if it contains an embedding of $C_{k_1}$ and of $C_{k_2}$ whereby there is an edge $e_m$ in the pattern that belongs to the two embeddings such that $e_m$ is a truss edge of $C_{k_1}$'s embedding and is a non-truss edge of $C_{k_2}$'s embedding, respectively. In other words, $C_{k_1}$ and $C_{k_2}$ can form a \textsc{ccp} ($CCP_{tn}(k_1, k_2)$) by merging a truss edge $e_1$ from $C_{k_1}$ with a non-truss edge from $C_{k_2}$ if the following conditions are satisfied: (a) \textit{Condition 1:} There is a $C_{k_1}$ pattern on $e_1$ containing $e_2$. (b) \textit{Condition 2:} There is a $C_{k_2}$ pattern on $e_2$ where $e_2\neq e_1$.

Note that due to Lemma~\ref{lem:cnn}, \textit{Condition 1} holds if \linebreak $|NB_{cc}(k_2, e_2)\setminus\{u,v\}|\geq (k_2-2)$ where $e_1=(u,v)$. Further, if $|NB_{cc}(k_1, e_1) \bigcup NB_{cc}(k_2, e_2)\setminus\{u,v\}|\geq (k_1-2)+(k_2-2))$, then the pattern $CCP_{tn}(k_1, k_2)$ must exist. Hence, \textsc{Tattoo} checks the conditions iteratively on decreasing $k_2$ and skips checks for $k_2^{\prime}<k_2$ if the conditions are satisfied for $k_2$\eat{ (due to Lemma~\ref{lem:coverageOfSimpleTrussAntiMonotone})}. The \textit{frequency} of $CCP_{tn}(k_1, k_2)$ is simply the number of such $e_1$ edges. For $CCP_{nt}(k_1, k_2)$, the approach is the same by swapping $C_{k_1}$ with $C_{k_2}$.

\eat{\vspace{-1ex}\begin{definition}\label{def:coverageOfCTTN} {\em
			Given a \textsc{tir} graph $G_T=(V_T,E_T)$, let the \textsc{ccp} $CCP_{tn}(k_1, k_2)$ merge the truss edge of $C_{k_1}$ with a non-truss edge of $C_{k_2}$ where $e_1$ and $e_2$ are the truss edges of $C_{k_1}$ and $C_{k_2}$, respectively. Then, the \textbf{frequency} of $CCP_{tn}(k_1, k_2)$ is defined as $freq(CCP_{tn})=|\{e_1\in E_T\}|$ where $e_1=(u,v)$ and $e_2$ satisfies the following conditions: (1) $t(e_1)\geq k_1$ and $t(e_2)\geq k_2$, (2) $e_2\in EB_{cc}(k_1, e_1)$, (3) $|NB_{cc}(k_2, e_2)\setminus\{u,v\}|\geq (k_2-2)$, and (4) $|NB_{cc}(k_1, e_1)\bigcup NB_{cc}(k_2, e_2)\setminus\{u,v\}|\geq (k_1-2)+(k_2-2)$.\/}
\end{definition}}

\eat{\vspace{-1ex}\begin{example}
		Consider the graph in Figure~\ref{fig:nonTrussEdgeMerge}(a). Observe that every edge is in a $4$-truss. Let us compute the frequency of $CCP_{tn}(4,5)$, which is formed by combining $C_4$ and $C_5$ where $k_1=4$ and $k_2=5$. Each edge $e$ is considered in turn as the truss edge (\ie, $e_1$) of $C_4$. Suppose $e_1=(A,B)$ and $e_2 = (B, D)$ are truss edges. Then, based on Definition~\ref{def:coverageOfCTTN}, $EB_{cc}(4, (A, B))=\{(A,C),(B,C),(A,D),(B,D)\}$ contains $(B, D)$ (Condition 2) and only $t((B,D))\geq 5$ (Condition 1). $|NB_{cc}(5,(B,D))\setminus\{A,B\}|=|\{E,F,G\}|=3\geq(k_2-2)$ (Condition 3). Further, $NB_{cc}(4,(A,B))=\{C,D\}$. Hence, $|NB_{cc}(4,(A,B))\bigcup NB_{cc}(5,(B,D))\setminus\{A,B\}|=|\{C,D,E,F,G\}|=5\geq(k_1-2)+(k_2-2)$ (Condition 4). Consequently, $CCP_{tn}(4,5)$ can be obtained by considering $(A,B)$ as $e_1$. Similarly, $CCP_{tn}(4,5)$ can be found by considering $(A,D)$, $(B,C)$ or $(C,D)$ as $e_1$.
		
		However, when $(A,C)$ is considered as $e_1$, no edges in $EB_{cc}(4, (A, C))$ is in 5-truss as $EB_{cc}(4, (A, C))=\{(C,D), \linebreak (A,D),(B,C),(A,B)\}$. Similarly, when $(D,E)$ is considered as $e_1$, $EB_{cc}(4, (D, E))=\{(C,D),(A,D),(B,D),(D,G),\linebreak (D,F),(B,E), (E,G),(E,F)\}$ where $t((B,D))=t((D,G))=t((D,F))=t((B,E))=t((E,G))=t((E,F))$ and $t((B,D))\geq 5$. However, $|NB_{cc}(5,(B,D))\setminus\{D,E\}|=|\{F,G\}|=2<(k_2-2)$, which is the case for all $e\in EB_{cc}(4, (D, E))$ where $t(e)\geq 5$. Hence, $CCP_{tn}(4,5)$ cannot be obtained by considering $(D,E)$ as $e_1$. The same situation occurs when any edge in the clique consisting of vertices $B$, $D$, $E$, $F$ and $G$ is chosen as $e_1$. Thus, $freq(CCP_{tn})=4$.
		\EndOfProof
\end{example}}

\eat{\begin{lemma}\label{lem:complexityOfGetCTTN}
		The worst-case time and space complexities of computing $CCP_{tn/nt}$ are $O(k_{max})$ and $O(|E_T|+|V_T|)$, respectively.
\end{lemma}}

\eat{
	\begin{algorithm}[t]
		\algsetup{
			linenosize=\small
		}
		\begin{algorithmic}[1]
			\small
			\REQUIRE Truss graph $G_T$, edge $e_1=(u_1,v_1)$ and its trussness $k_1$, edge $e_2=(u_2,v_2)$ and its trussness $k_2$;
			\ENSURE Combined truss patterns $CT_{nnc}$ and frequency $freq(CT_{nnc})$;
			
			\WHILE{$k_2\geq 3$}
			
			\STATE $CNN_{k_1}(e_1)\leftarrow GetCNN(G_T,e_1,k_1)\;$\label{line:getCTNNCConditionStart}
			
			\STATE $CEN_{k_1}(e_1)\leftarrow GetCEN(G_T,e_1,CNN_{k_1}(e_1))\;$
			
			\STATE $CNN_{k_2}(e_2)\leftarrow GetCNN(G_T,e_2,k_2)\;$
			
			\STATE $CEN_{k_2}(e_2)\leftarrow GetCEN(G_T,e_2,CNN_{k_2}(e_2))\;$
			
			\STATE $u\leftarrow GetCommonNode(u_1,v_1,u_2,v_2)\;$
			
			\STATE $E_c\leftarrow\phi /*E_c\textrm{ holds the set of candidate }e_3*/\;$
			
			\FOR{$e^{\prime}=(x,y)\in CEN_{k_1}(e_1)$}
			
			\IF{$u\in\{x,y\}$ and $e^{\prime}\in CEN_{k_2}(e_2)$ and $e^{\prime}\neq e_1$ and $e^{\prime}\neq e_2$}
			
			\STATE $E_c\leftarrow E_c\bigcup\{e^{\prime}\}\;$
			
			\ENDIF
			
			\ENDFOR
			
			\FOR{$e=(u,w)\in E_c$}
			
			\IF{$|GetCNN(G_T,e_1,k_1)\setminus\{u,w\}|\geq (k_1-3)$ and $|GetCNN(G_T,e_2,k_2)\setminus\{u,w\}|\geq (k_2-3)$ and $|GetCNN(G_T,e_1,k_1)\setminus\{u,w\}\bigcup GetCNN(G_T,e_2,k_2)\setminus\{u,w\}|\geq(k_1-3+k_2-3)$}\label{line:getCTNNCConditionEnd}
			
			\STATE $k_3\leftarrow k_2\;$
			
			\WHILE{$k_3\geq 3$}
			
			\STATE $CT_{nnc}\leftarrow CT_{nnc}\bigcup CT_{nnc(k_1,k_3)}\;$
			
			\STATE $freq(CT_{nnc(k_1,k_3)})\leftarrow freq(CT_{nnc(k_1,k_3)})+1\;\label{line:getCTNNCCoverage}$
			
			\STATE $k_3\leftarrow k_3-1\;$
			
			\ENDWHILE
			
			\BREAK
			
			\ENDIF
			
			\ENDFOR
			
			\STATE $k_2\leftarrow k_2-1\;$
			
			\ENDWHILE
			
		\end{algorithmic}
		\caption{$GetNNC$.} \label{alg:getCTNNC}
	\end{algorithm}
}

\eat{
	\begin{algorithm}[t]
		\algsetup{
			linenosize=\small
		}
		\begin{algorithmic}[1]
			\small
			\REQUIRE Truss graph $G_T$, edge $e_1=(u_1,v_1)$, its trussness $k_1$, its $k_1$-CT node and edge neighbourhood $CNN_{k_1}(e_1)$ and $CEN_{k_1}(e_1)$, edge $e_2=(u_2,v_2)$ and its trussness $k_2$;
			\ENSURE Combined truss patterns $CT_{nnc}$ and frequency $freq(CT_{nnc})$;
			
			\WHILE{$k_2\geq 3$}
			
			\STATE $CNN_{k_2}(e_2)\leftarrow GetCNN(G_T,e_2,k_2)\;$\label{line:getCTNNCConditionStart}
			
			\STATE $CEN_{k_2}(e_2)\leftarrow GetCEN(G_T,e_2,CNN_{k_2}(e_2))\;$
			
			\STATE $u\leftarrow GetCommonNode(u_1,v_1,u_2,v_2)\;$
			
			\STATE $E_c\leftarrow\phi /*E_c\textrm{ holds the set of candidate }e_3*/\;$
			
			\FOR{$e^{\prime}=(x,y)\in CEN_{k_2}(e_2)$}
			
			\IF{$u\in\{x,y\}$ and $e^{\prime}\neq e_1$}
			
			\STATE $E_c\leftarrow E_c\bigcup\{e^{\prime}\}\;$
			
			\ENDIF
			
			\ENDFOR
			
			\FOR{$e=(u,w)\in E_c$}
			
			\IF{$|GetCNN(G_T,e_1,k_1)\setminus\{u,w\}|\geq (k_1-3)$ and $|GetCNN(G_T,e,k_2)\setminus\{u,w\}|\geq (k_2-3)$ and $|GetCNN(G_T,e_1,k_1)\setminus\{u,w\}\bigcup GetCNN(G_T,e,k_2)\setminus\{u,w\}|\geq(k_1-3+k_2-3)$}\label{line:getCTNNCConditionEnd}
			
			\STATE $k_3\leftarrow k_2\;$
			
			\WHILE{$k_3\geq 3$}
			
			\STATE $CT_{nnc}\leftarrow CT_{nnc}\bigcup CT_{nnc(k_1,k_3)}\;$
			
			\STATE $freq(CT_{nnc(k_1,k_3)})\leftarrow freq(CT_{nnc(k_1,k_3)})+1\;\label{line:getCTNNCCoverage}$
			
			\STATE $k_3\leftarrow k_3-1\;$
			
			\ENDWHILE
			
			\BREAK
			
			\ENDIF
			
			\ENDFOR
			
			\STATE $k_2\leftarrow k_2-1\;$
			
			\ENDWHILE
			
		\end{algorithmic}
		\caption{$GetNNC$.} \label{alg:getCTNNC}
	\end{algorithm}
}

\begin{figure}[!t]
	\centering
	\includegraphics[width=3.3in]{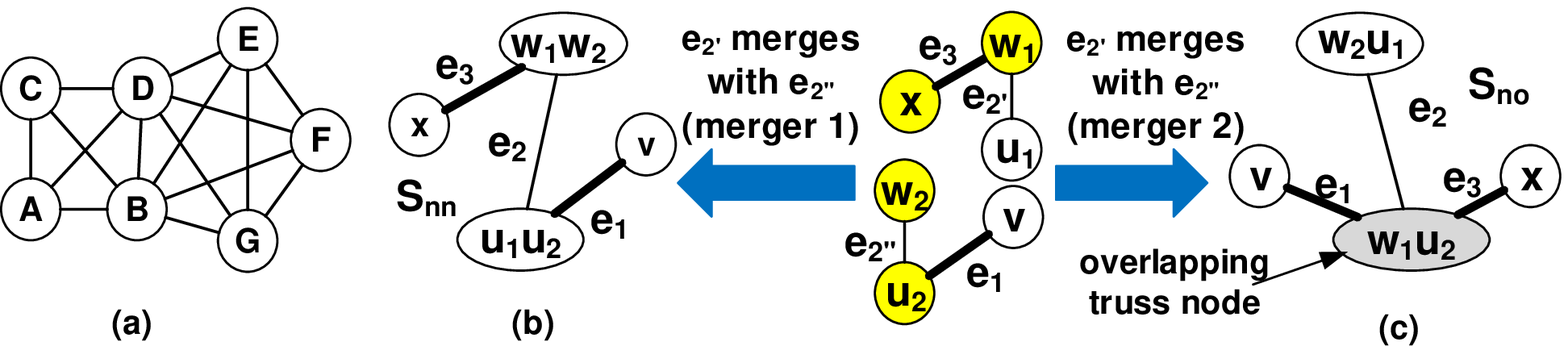}
	\vspace{-2ex}\caption{(a) A $G_T$; (b) Skeleton structure of $CCP_{nn}$; (c) Skeleton structure of $CCP_{no}$. $e_1$ and $e_3$ are truss edges.}\label{fig:nonTrussEdgeMerge}
	\vspace{-3ex}
\end{figure}

\eat{\underline{\textit{Frequencies of $CCP_{nn}(k_1, k_2)$ and $CCP_{no}(k_1, k_2)$}}.
	Recall that\linebreak (Def.~\ref{def:ccp}) a single-edge merge can also involve the merger of two non-truss edges, each from a different $k$-\textsc{cp}. Each non-truss edge contains a truss vertex. There are two ways in which two non-truss edges can merge as shown in Figures~\ref{fig:nonTrussEdgeMerge}(b) and (c). In the former (resp. latter), vertex pairs ($w_1,w_2$) (resp. ($w_2,u_1$)) and ($u_1,u_2$) (resp. ($w_1,u_2$)) are merged. Hence, \eat{Recall from Definition~\ref{def:ccp}, $CCP_{nn}$ involves the merger of a non-truss edge $e$ from a $k_1$-\textsc{cp} with a non-truss edge $e^{\prime}$ from another $k_2$-\textsc{cp} resulting in a merged edge $e_2$ containing a common truss vertex between $e$ and $e^{\prime}$.}a pattern is a $CCP_{nn}$ if it contains embeddings of $C_{k_1}$ and of $C_{k_2}$ whereby there is an edge $e_m$ in the pattern that belongs to both embeddings such that (1) $e_m=(u_m,v_m)$ is a non-truss edge of both embeddings and (2) either $u_m$ or $v_m$ is a truss vertex of both embeddings. Such a merger produces a structure shown in Figure~\ref{fig:nonTrussEdgeMerge}(b) which we refer to as the \textit{skeleton structure} of $CCP_{nn}$ (denoted as $\mathbb{S}_{nn}$). Observe that such a structure is present in \emph{all} $CCP_{nn}$. Hence, we can search for the $\mathbb{S}_{nn}$ of a $CCP_{nn}$ in a \textsc{tir} graph to compute its occurrence and frequency. Specifically, a $CCP_{nn}$ can be obtained if the followings are satisfied: (a) \textit{Condition 1}: There is a $C_{k_1}$ pattern on its truss edge $e_1=(u,v)$ which contains $e_2=(u,w)$. (b) \textit{Condition 2}: There is a $C_{k_2}$ pattern on its truss edge $e_3=(u,x)$ which contains $e_2$.}

\underline{\textit{Frequencies of $CCP_{nn}(k_1, k_2)$ and $CCP_{no}(k_1, k_2)$}}.
Recall that (Def.~\ref{def:ccp}) a single-edge merge can also involve the merger of two non-truss edges, each from a different $k$-\textsc{cp}. Each non-truss edge contains a truss vertex. There are two ways in which two non-truss edges can merge as shown in Figures~\ref{fig:nonTrussEdgeMerge}(b) and (c). In the former (resp. latter), vertex pairs ($w_1,w_2$) (resp. ($w_2,u_1$)) and ($u_1,u_2$) (resp. ($w_1,u_2$)) are merged. Hence, a pattern is a $CCP_{nn}$ if it contains at least one embedding of a structure shown in Figure~\ref{fig:nonTrussEdgeMerge}(b) which we refer to as the \textit{skeleton structure} of $CCP_{nn}$ (denoted as $\mathbb{S}_{nn}$). Hence, we can search for the $\mathbb{S}_{nn}$ of a $CCP_{nn}$ in a \textsc{tir} graph to compute its occurrence and frequency. Specifically, a $CCP_{nn}$ can be obtained if the followings are satisfied: (a) \textit{Condition 1}: There is a $C_{k_1}$ pattern on its truss edge $e_1=(u_1u_2,v)$ which contains $e_2=(u_1u_2,w_1w_2)$. (b) \textit{Condition 2}: There is a $C_{k_2}$ pattern on its truss edge $e_3=(w_1w_2,x)$ which contains $e_2$.

Note that \textit{Condition 1} holds if $|NB_{cc}(k_1, e_1)\setminus\{u_1u_2,w_1w_2\}|\geq (k_1-3)$ (Lemma~\ref{lem:cnn}). Similarly, \textit{Condition 2} holds if \linebreak $|NB_{cc}(k_2, e_3)\setminus\{u_1u_2,w_1w_2\}|\geq (k_2-3)$. Further, if \linebreak $|NB_{cc}(k_1, e_1)\setminus\{u_1u_2,w_1w_2\}\bigcup NB_{cc}(k_2, e_3)\setminus\{u_1u_2,w_1w_2\}|\geq (k_1-3)+(k_2-3)$, then the pattern $CCP_{nn}$ must exist. The \textit{frequency} of a $CCP_{nn}$ is simply the number of skeleton structures $\mathbb{S}_{nn}$ in a \textsc{tir} graph.

\eat{\begin{definition}\label{def:coverageOfCTNNC} {\em
			Given a \textsc{tir} graph $G_T=(V_T,E_T)$, let the composite chord pattern $CCP_{nn}$ merges the non-truss edge of $C_{k_1}$ with a non-truss edge of $C_{k_2}$ where $e_1=(u,v)$ and $e_3=(u,x)$ are the truss edges of $C_{k_1}$ and $C_{k_2}$, respectively. Then, the \textbf{frequency} of a $CCP_{nn}$ is defined as $freq(CCP_{nn})=|\{e_1\in E\}|$ and satisfies the following conditions:\/}
		\begin{enumerate} {\em
				\item $t(e_1)\geq k_1$ and $t(e_3)\geq k_2$,
				\item $e_2=(u,w)\in EB_{cc}(k_1, e_1)$ and $e_2\in EB_{cc}(k_2, e_3)$,
				\item $|NB_{cc}(k_1, e_1)\setminus\{u,w\}|\geq (k_1-3)$,
				\item $|NB_{cc}(k_2, e_3)\setminus\{u,w\}|\geq (k_2-3)$ and
				\item $|NB_{cc}(k_1, e_1)\setminus\{u,w\}\bigcup NB_{cc}(k_2, e_3)\setminus\{u,w\}|\geq (k_1-3)+(k_2-3)$.\/}
		\end{enumerate}
	\end{definition}
	
	\begin{example}
		Reconsider Figure~\ref{fig:nonTrussEdgeMerge}(b). The frequency of $CCP_{nn(5,4)}$ where $k_1=5$ and $k_2=4$ can be computed as follows. Observe that every edge in the clique involving vertices $B$, $D$, $E$, $F$, and $G$ is in a $5$-truss. Suppose $(D,E)$ is $e_1$. Then, $EB_{cc}(5, (D, E))=\{(B,D),(D,G),(D,F),(B,E),\linebreak(E,G),(E,F)\}$. Suppose the merged edge is $(B,D)$. Then $(C,D)$ must be $e_2$ (\ie forming the skeleton structure) if $CT_{nnc(5,4)}$ exists. Since $|NB_{cc}(5, (D,E))\setminus\{B,D\}|=|\{F,G\}|=2\geq(k_1-3)$ (\textit{Condition 3} in Definition~\ref{def:coverageOfCTNNC}), $|NB_{cc}(4, (C, D))\setminus\{B,D\}|=|\{A\}|=1\geq(k_2-3)$ (\textit{Condition 4}) and $|NB_{cc}(5, (D, E))\setminus\{B,D\}\bigcup NB_{cc}(4, (C, D))\setminus\{B,D\}|=|\{F,G,A\}|=3\geq (k_1-3)+(k_2-3)$ (\textit{Condition 5}), then $(B,D)$ and $(C,D)$ can be the merged edge and $e_2$, respectively. Hence, $CCP_{nn}(4,5)$ can be found in this data graph. The same applies when $(D,F)$, $(D,G)$, $(B,E)$, $(B,F)$ or $(B,G)$ is chosen as $e_1$. Selecting other edges as $e_1$ yield no suitable candidate for the merged edge and $e_2$. Hence, $freq(CCP_{nn}(4,5))=6$.
		\vspace{-0.5ex}\EndOfProof\end{example}}

$CCP_{no}$ is very similar to $CCP_{nn}$ except that the truss vertices of the merged edges are not combined during the merger. Figure~\ref{fig:nonTrussEdgeMerge}(c) illustrates the \textit{skeleton structure} of a $CCP_{no}$ ($\mathbb{S}_{no}$), which occurs in all $CCP_{no}$. \eat{Hence, conditions for \textsc{ccp} $CCP_{no}$ are identical to those discussed above except that in this case $e_3=(w_1u_2,x)$ instead of $(w_1w_2,x)$.} The \textit{frequency} of a $CCP_{no}$ is the number of skeleton structures $\mathbb{S}_{no}$.

Observe that $freq(CCP_{nn}(k_1,k_2))=freq(CCP_{no}(k_2,k_1))$ since $k_1$ and $k_2$ can be swapped. The same is true for $CCP_{tn}$ and $CCP_{nt}$. Hence, when combining two $k$-\textsc{cp}s, we only consider the case when $k_1\geq k_2$.

\eat{\vspace{-1ex}\begin{lemma}\label{lem:complexityOfGetCTNN}
		\textit{The worst-case time and space complexities of computing $CCP_{nn/no}$ are $O(k_{max}|EB_{max}|)$ and $O(|E_T|+|V_T|)$, respectively, where $EB_{max}$ is the largest $k$-\textsc{ccp} edge neighbourhood.}
\end{lemma}}

\begin{algorithm}[t]
	\algsetup{
		linenosize=\scriptsize
	}
	\begin{algorithmic}[1]
		\scriptsize
		\REQUIRE \textsc{tir} graph $G_T=(V_T,E_T)$, trussness of all edges $T(e)$;
		\ENSURE Composite chord patterns $P=\{CCP_{tn}\bigcup CCP_{nn}\bigcup CCP_{no}\}$ and frequency $freq(P)$ where $CCP_{tn}=\{CCP_{tn}(k_1,k_2) | 3<k_1\leq k_{max}, 3\leq k_2\leq k_{max}\}$, $CCP_{nn}=\{CCP_{nn}(k_1,k_2) | 3<k_1\leq k_{max}, 3\leq k_2\leq k_{max}\}$ and $CCP_{no}=\{CCP_{no}(k_1,k_2) | 3<k_1\leq k_{max}, 3\leq k_2\leq k_{max}\}$;
		
		\STATE $CCP_{tn}\leftarrow\phi$, $CCP_{nn}\leftarrow\phi$, $CCP_{no}\leftarrow\phi$\;
		
		\FOR{$e_1\in E_T$}
		
		\STATE $k_1\leftarrow t(e_1)$\;
		
		\STATE Compute $NB_{cc}(k_1, e_1)$ \label{line:combinedTrussPatternGetCNN}/* \textrm{\textsf{compute $k$-\textsc{ccp} node neighbourhood}}*/\;
		
		\STATE Compute $EB_{cc}(k_1, e_1)$ \label{line:combinedTrussPatternGetCEN}/* \textrm{\textsf{compute $k$-\textsc{ccp} edge neighbourhood}}*/\;
		
		\WHILE{$k_1\geq 4/* \textrm{\textsf{find composite chord patterns}}*/$} \label{line:combinedTrussPatternSkipSimple4TrussPattern}
		
		\FOR{$e_2\in EB_{cc}(k_1, e_1)$}
		
		\STATE $k_2\leftarrow \textsc{Min}(t(e_2),k_{max}-k_1)\;$
		
		\STATE $CCP_{tn}, freq(CCP_{tn})\leftarrow GetTN(G_T,e_1,k_1,e_2,k_2) \label{line:combinedTrussPatternGetCTTN}\;$
		
		\STATE $CCP_{nn}, freq(CCP_{nn})\leftarrow GetNN(G_T,e_1,k_1,NB_{cc}(k_1, e_1),\linebreak EB_{cc}(k_1, e_1),e_2,k_2,\textsf{NN}) \label{line:combinedTrussPatternGetCTNNC}\eat{/*Alg.~\ref{alg:getCTNN}, Appendix~\ref{app:pseudo}*/}\;$
		
		\STATE $CCP_{no}, freq(CCP_{no})\leftarrow GetNN(G_T,e_1,k_1,NB_{cc}(k_1, e_1),\linebreak EB_{cc}(k_1, e_1),e_2,k_2,\textsf{NO}) \label{line:combinedTrussPatternGetCTNNN}\eat{/*Alg.~\ref{alg:getCTNN}*/}\;$
		
		\ENDFOR
		
		$k_1\leftarrow k_1-1\;$
		
		\ENDWHILE \label{line:end}
		
		\ENDFOR
	\end{algorithmic}
	\caption{$\textsc{GenCombChordPatterns}$.} \label{alg:findCombinedTrussPatterns}
\end{algorithm}

\underline{\textbf{Algorithm}}. Putting the above strategies together (outlined in Algorithm~\ref{alg:findCombinedTrussPatterns}), the \textsc{ccp}s are computed as follows. For each edge in $G_T$, compute the $k_1$-\textsc{ccp} node and edge neighbourhoods (Lines~\ref{line:combinedTrussPatternGetCNN}-\ref{line:combinedTrussPatternGetCEN}). Next, it computes the four types of \textsc{ccp}s (Lines~\ref{line:combinedTrussPatternSkipSimple4TrussPattern}-\ref{line:end}) based on the aforementioned strategies. Note that the smallest \textsc{ccp} generated is a \textsc{ccp}(3,4) due to Lemma~\ref{lem:combinedSimple3TrussPattern}. \eat{Hence, $k_1\geq 4$ in Line~\ref{line:combinedTrussPatternSkipSimple4TrussPattern}.} Also, we only compute $CCP_{tn}(k_1,k_2)$ instead of both $CCP_{tn}(k_1,k_2)$ and $CCP_{nt}(k_1,k_2)$ as $CCP_{nt}(k_1,k_2)$ is covered when $k_2$ and $k_1$ are swapped.

\vspace{-1ex}\begin{theorem}\label{lem:complexityFindCombinedTrussPatterns}
	\textit{The worst-case time and space complexities of the \textsc{ccp} generation technique are $O(k_{max}^2|E_T||EB_{max}|^2)$ and $O(k_{max}|E_T|+|V_T|)$, respectively.
}\end{theorem}

\vspace{-2ex}
\subsection{Patterns from a TOR Graph} \label{sec:torg}
\eat{We now present candidate pattern generation from the \textsc{tor} graph $G_O$. Recall that $G_O$ comprises of subgraphs that are stars, paths, cycles, and small connected subgraphs of unique topologies.} Generation of candidates from a \textsc{tor} graph consists of two phases: \textit{star pattern extraction} and \textit{small pattern extraction}. The former extracts star and asterism patterns. Subsequently, the edges involved in these patterns are removed from $G_O$ resulting in further decomposition of the \textsc{tor} graph. The resultant graph is referred to as the \textit{remainder graph} ($G_R$). Then, the second phase extracts paths, cycles, and small connected subgraphs from $G_R$.\eat{ The formal algorithms are given in~\cite{tech}.}

\textbf{Extraction of star and asterism patterns}. The frequencies of these patterns can be derived directly from their definitions (Sec.~\ref{sec:pattop}). Specifically, $freq(S_k)=|\{v | v\in V_O, deg(v)=k\}|$ and $freq(A_S)=freq(\{E_m=\{e_{m_1},\ldots,e_{m_{n-1}}\})$
where $e_{m_i}=(r_i,r_{i+1})\in E_O$, $\{k,k_i\}\geq\epsilon$, $deg(r_i)=k_i$ and $deg(r_{i+1})=k_{i+1}$.
Algorithm~\ref{alg:findStarPatterns} outlines the procedure. The star and asterism patterns are extracted in Lines~\ref{line:findStarPatternsStarStart} to~\ref{line:findStarPatternsStarEnd} and Lines~\ref{line:findStarPatternsCombinedStarStart} to~\ref{line:findStarPatternsCombinedStarEnd}, respectively. Briefly, asterism patterns are found using breadth-first search (\textsc{bfs}). A vector of vertices is used to keep track of star centers in an asterism pattern. We ``grow'' the pattern by adding a neighbouring vertex $z$ of the current star center being considered only if $deg(z)\geq\epsilon$ and when the size of the grown pattern is less than or equals to $\eta_{max}$.

\eat{\begin{algorithm}[t]
		\algsetup{
			linenosize=\scriptsize
		}
		\begin{algorithmic}[1]
			\scriptsize
			\REQUIRE \textsc{tor} graph $G_O=(V_O, E_O)$
			\ENSURE Stars and asterisms $P_s$ and frequency $freq(P_s)$;
			
			\STATE $P_s\leftarrow\phi\;$
			
			\FOR{$v\in V_O$}\label{line:findStarPatternsStarStart}
			
			\IF{$deg(v)\geq\epsilon$} \label{line:findStarPatternsStarEpsilon}
			
			\STATE $P_s\leftarrow P_s\bigcup S_{deg(v)}\;$
			
			\STATE $freq(S_{deg(v)})\leftarrow freq(S_{deg(v)})+1\;$
			
			\ENDIF
			
			\ENDFOR\label{line:findStarPatternsStarEnd}
			
			\FOR{$e=(u,v)\in E_O$}\label{line:findStarPatternsCombinedStarStart}
			
			\IF{$deg(u)\geq\epsilon$ and $deg(v)\geq\epsilon$}\label{line:findStarPatternsCombinedStarEpsilon}
			
			\STATE $k_1\leftarrow \textsc{Min}(deg(u),deg(v))\;$
			
			\STATE $k_2\leftarrow \textsc{Max}(deg(u),deg(v))\;$
			
			\STATE $P_s\leftarrow P_s\bigcup A_{k_1,k_2}\;$
			
			\STATE $freq(A_{k_1,k_2})\leftarrow freq(A_{k_1,k_2})+1\;$
			
			\ENDIF
			
			\ENDFOR\label{line:findStarPatternsCombinedStarEnd}
			
		\end{algorithmic}
		\caption{$\textsc{GenStarPatterns}$.} \label{alg:findStarPatterns}
	\end{algorithm}
}

\begin{algorithm}[t]
	\algsetup{
		linenosize=\scriptsize
	}
	\begin{algorithmic}[1]
		\scriptsize
		\REQUIRE \textsc{tor} graph $G_O=(V_O, E_O)$
		\ENSURE Stars and asterisms $P_s$ and frequency $freq(P_s)$;
		
		\STATE $P_s\leftarrow\phi\;$
		
		\FOR{$v\in V_O$}\label{line:findStarPatternsStarStart}
		
		\IF{$deg(v)\geq\epsilon$} \label{line:findStarPatternsStarEpsilon}
		
		\STATE $P_s\leftarrow P_s\bigcup S_{deg(v)}\;$
		
		\STATE $freq(S_{deg(v)})\leftarrow freq(S_{deg(v)})+1\;$
		
		\STATE $Q\leftarrow\phi/*\textrm{Q is a queue}*/\;$\label{line:findStarPatternsCombinedStarStart}
		
		\STATE $SC\leftarrow \textsc{InsertLast}(SC,v)/* \textrm{\textsf{$v$ is appended to $SC$, a vector of nodes}}*/\;$
		
		\STATE $Q\leftarrow\textsc{Enqueue}(Q,SC)\;$
		
		\WHILE{$Q\neq\phi$}
		
		\STATE $SC_{curr}\leftarrow\textsc{Dequeue}(Q)$
		
		\STATE $u\leftarrow\textsc{GetLast}(SC_{curr})/* \textrm{\textsf{retrieve last element in $SC_{curr}$}}*/\;$
		
		\FOR{$z\in\textsc{Neighbours}(u)$}
		
		\IF{$z\notin SC_{curr}$ and $deg(z)\geq\epsilon$ and $\textsc{Size}(SC_{curr})+\textsc{Size}(S_{deg(z)})-1\leq\eta_{max}$}
		
		\STATE $SC_{curr}\leftarrow\textsc{InsertLast}(SC_{curr},z)\;$
		
		\STATE $P_s\leftarrow P_s\bigcup A_{SC_{curr}}\;$
		
		\STATE $freq(A_{SC_{curr}})\leftarrow freq(A_{SC_{curr}})+1\;$
		
		\STATE $Q\leftarrow\textsc{Enqueue}(Q,SC_{curr})\;$
		
		\ENDIF
		
		\ENDFOR
		
		\ENDWHILE\label{line:findStarPatternsCombinedStarEnd}
		
		\ENDIF
		
		\ENDFOR\label{line:findStarPatternsStarEnd}
		
	\end{algorithmic}
	\caption{$\textsc{GenStarPatterns}$.} \label{alg:findStarPatterns}
\end{algorithm}

\vspace{-1ex} \begin{lemma}\label{lem:complexityFindStarLikePatterns}
	\textit{The worst-case time and space complexities of star and asterism pattern extraction are $O(|V_O|^2)$ and $O(|E_O|+|V_O|)$, respectively.}
	\vspace{-0.5ex}\end{lemma}

\textbf{Extraction of small patterns.} The remainder graph $G_R$ is primarily composed of small connected components such as paths, cycles, and subgraphs with unique topology. Algorithm~\ref{alg:findSmallPatterns} outlines the extraction of these small patterns and we denote $k$-cycle as $Y_k$ and subgraphs with unique topology as $U$. Given a graph $G=(V,E)$, a \textit{$k$-path}, denoted as $P_{k}=(V_k,E_k)$, is a walk of length $k$ containing a sequence of vertices $v_1,v_2,\cdots,v_k,v_{k+1}$ where $E_k\subseteq E$, $V_k\subseteq V$ such that all vertices in $V_k$ are distinct. A \textit{$k$-cycle} is simply a closed $(k-1)$-path where $k\geq 3$.\eat{ That is, there is an edge $(v_k,v_1)$ in $k$-cycle connecting the ending node ($v_k$) of $(k-1)$-path to starting node ($v_1$) of $(k-1)$-path.} We refer to small subgraph patterns as connected components in $G_R$ that are neither $k$-paths nor $k$-cycles. Note that 1-path, 2-path, 3-cycle and 4-cycle are basic building blocks of real-world networks~\cite{milo2002}.\eat{ Hence, they are likely to be useful for constructing queries.} Recall that in \textsc{Tattoo}, we consider them as \textit{default} patterns\eat{\footnote{\scriptsize Default patterns are permanent fixture on the \textsc{gui} regardless of the dataset.}} and they are not part of the candidate canned pattern set. Hence, we extract all $k$-paths for $k > 2$ (Lines~\ref{line:findSmallPatternsFindPathStart}-\ref{line:findSmallPatternsFindPathEnd}) and $k$-cycles for $k>4$ (Lines~\ref{line:findSmallPatternsFindCycleStart}-\ref{line:findSmallPatternsFindCycleEnd}) and their frequencies. After that, small connected subgraphs and their corresponding frequencies are extracted.\eat{ Since it is possible to have multiple occurrences of a subgraph with the same topology, \textsc{Tattoo} handles this by checking if each newly discovered subgraph pattern is isomorphic to an existing one and updates its frequency accordingly.}

\begin{algorithm}[t]
	\algsetup{
		linenosize=\small
	}
	\begin{algorithmic}[1]
		\small
		\REQUIRE Remainder graph $G_R=(V_R, E_R)$, pattern budget $b=(\eta_{min}, \eta_{max}, \gamma)$
		\ENSURE Small patterns $P_r=\{P\bigcup Y\bigcup U\}$ and frequency $freq(P_r)$ where $P=\{P_k | k\geq 3\}$, $Y=\{Y_k | k\geq 5\}$;
		
		\STATE $P_r\leftarrow\phi$, $P\leftarrow\phi$, $Y\leftarrow\phi$, $U\leftarrow\phi\;$
		
		\STATE $UC_{maxID}\leftarrow 0\;/*\textrm{maximum ID of UC}*/$
		
		\FOR{$v\in V_R$}
		
		\STATE \textrm{set $v$ as unvisited}\;
		
		\ENDFOR
		
		\FOR{$v\in V_R$}
		
		\IF{$v$\textrm{ is not visited}}
		
		\STATE \textrm{find component $C=(V_C,E_C)$ containing $v$}\;
		
		\STATE $n_{deg1}\leftarrow 0\;/*\textrm{num of nodes with deg=1}*/$
		
		\STATE $n_{deg2}\leftarrow 0\;/*\textrm{num of nodes with deg=2}*/$
		
		\FOR{$u\in V_C$}\label{line:findSmallPatternsFindNumDegStart}
		
		\IF{$deg(u)=1$}
		
		\STATE $n_{deg1}\leftarrow n_{deg1}+1$\;
		
		\ELSIF{$deg(u)=2$}
		
		\STATE $n_{deg2}\leftarrow n_{deg2}+1$\;
		
		\ENDIF
		
		\STATE \textrm{set $u$ as visited}\;
		
		\ENDFOR\label{line:findSmallPatternsFindNumDegEnd}
		
		\IF{$n_{deg1}=2$ and $n_{deg2}=|V_C|-2$ and $|V_C|\neq 2$ and $|V_C|\neq 3$}\label{line:findSmallPatternsFindPathStart}
		
		\STATE $P\leftarrow P\bigcup P_{|V_C|-1}\;$
		
		\STATE $freq(P_{|V_C|-1})\leftarrow freq(P_{|V_C|-1})+1\;$\label{line:findSmallPatternsFindPathEnd}
		
		\ELSIF{$n_{deg1}=0$ and $n_{deg2}=|V_C|$ and $|V_C|\neq 3$ and $|V_C|\neq 4$}\label{line:findSmallPatternsFindCycleStart}
		
		\STATE $Y\leftarrow Y\bigcup Y_{|V_C|}\;$
		
		\STATE $freq(Y_{|V_C|})\leftarrow freq(Y_{|V_C|})+1\;$\label{line:findSmallPatternsFindCycleEnd}
		
		\ELSIF{$|E_C|\geq b.\eta_{min}$ and $|E_C|\leq b.\eta_{max}$}\label{line:findSmallPatternsFindUncommonStart}
		
		\IF{$\textsc{IsIsomorphic}(C,U)=true$}\label{line:findSmallPatternsFindUncommonIsomorphicStart}
		
		\STATE $U_{currID}\leftarrow \textsc{GetID}(C,U)\;$
		
		\STATE $freq(U_{currID})\leftarrow freq(U_{currID})+1\;$\label{line:findSmallPatternsFindUncommonIsomorphicEnd}
		
		\ELSE\label{line:findSmallPatternsFindUncommonNotIsomorphicStart}
		
		\STATE $U\leftarrow U\bigcup (C,U_{maxID})\;$
		
		\STATE $freq(U_{maxID})\leftarrow 1\;$
		
		\STATE $U_{maxID}\leftarrow U_{maxID}+1\;$
		
		\ENDIF\label{line:findSmallPatternsFindUncommonNotIsomorphicEnd}
		
		\ENDIF\label{line:findSmallPatternsFindUncommonEnd}
		
		\ENDIF
		
		\ENDFOR
	\end{algorithmic}
	\caption{$\textsc{GenSmallPatterns}$.} \label{alg:findSmallPatterns}
\end{algorithm}

\vspace{-1ex}\begin{lemma}\label{lem:complexityFindSmallPatterns}
	\textit{Worst-case time and space complexities \eat{of Alg.~\ref{alg:findSmallPatterns}}to find small patterns are $O(\eta_{max}|V_R|\eta_{max}!)$ and $O(|E_R|+|V_R|)$, respectively.}
\end{lemma}

\textbf{Remark.} Exponential time complexity of the small pattern extraction phase is due to the \eat{graph }isomorphism check. The time cost is small in practice due to the small size of candidate patterns and their number is typically small in $G_R$.

\eat{\begin{algorithm}[!t]
		\algsetup{
			linenosize=\scriptsize
		}
		\begin{algorithmic}[1]
			\small
			\REQUIRE Candidate pattern set $P_{all}$ and its frequency $freq(P_{all})$, plug $b=(\eta_{min}, \eta_{max}, \gamma, \mathbb{N})$;
			\ENSURE Canned pattern set $\mathcal{P}$;
			
			\STATE $\mathcal{P}\leftarrow \textsc{GetDefaultPatterns}()\;$
			
			\FOR{$p=(V_p,E_p)\in P_{all}\;$}\label{line:patternMiningPruneStart}
			
			\IF{$|E_p|<\eta_{min}$ or $|E_p|>\eta_{max}$ or $freq(p)<\delta$}
			
			\STATE $P_{all}\leftarrow P_{all}\setminus p\;$
			
			\ENDIF
			
			\ENDFOR\label{line:patternMiningPruneEnd}
			
			\STATE $\mathcal{P}\leftarrow\textsc{SelectCandidates}(P_{all},freq(P_{all}),b)\;$\label{line:selectCandidates}
			
		\end{algorithmic}
		\caption{$\textsc{SelectCannedPatterns}$.} \label{alg:patternMining}
\end{algorithm}}

\vspace{-1ex}
\section{Selection of Canned Patterns}\label{sec:cannedPatternSelector}
In this section, we describe the algorithm to select canned pattern set $\mathcal{P}$ from the generated candidate patterns. \eat{The algorithm selects patterns based on the pattern set score (Definition~\ref{def:patternScore}) derived from the characteristics of canned patterns (\ie coverage, cognitive load and similarity/diversity). Hence, we begin by introducing quantification of these characteristics.\eat{ In particular, computation of coverage, cognitive load and similarity are performed in Lines~\ref{line:patternMiningCovCogStart},~\ref{line:patternMiningCog} and ~\ref{line:patternMiningDiversity}, respectively, in Algorithm~\ref{alg:greedy}. Note that the multilinear function $F$ (Lines~\ref{line:dGreedyMultilinear1} and~\ref{line:dGreedyMultilinear2}) in Algorithm~\ref{alg:dGreedy} is derived from these characteristics as well.} As we shall see, the computation of these three measures are different from~\cite{catapult} due to the nature of large networks.} We begin by presenting the theoretical underpinning that influences the design of our algorithm.

\vspace{-1ex}
\subsection{Theoretical Analysis} \label{sec:soln}
Due to the hardness of the \textsc{cps} problem, we design an approximation algorithm to address it. We draw on insights from a related problem, \textit{team formation problem} (\textsc{tfp})~\cite{chen2004,bhowmik2014}, which aims to hire a team of individuals $T$ from a group of experts $S$ for a specific project where $T\subseteq S$. Bhowmik \textit{et al.}~\cite{bhowmik2014} proposed that several aspects should be considered in \textsc{tfp}, namely, skill coverage (\textit{skill}), social compatibility (\textit{social}), teaming cost (\textit{team}) and miscellaneous aspects such as redundant skills avoidance (\textit{red}) and inclusion of selected experts (\textit{exp}).\eat{ In \textsc{tfp}, a project demands particular set of skills and a \textit{skill} is \textit{covered} by an individual $n$ if $n$ is in the team. \textit{Social} considers the communication cost between pairs of skills for a particular project whereas \textit{team} is the cost of forming a team and includes consideration of team size and personnel cost. \textit{Red} is related to skill redundancy of individuals in the team due to overlapping skill set. In contrast, \textit{exp} facilitates inclusion of a predefined set of experts in the team.} The formulation of \textsc{tfp} is given as $s(T^{\prime})=\alpha_{skill}f_{skill}(T^{\prime})-\alpha_{social}f_{social}(T^{\prime})-\alpha_{team}f_{team}(T^{\prime})-\alpha_{red}f_{red}(T^{\prime})+\alpha_{exp}f_{exp}(T^{\prime})$ were $\alpha_{skill}$, $\alpha_{social}$, $\alpha_{team}$, $\alpha_{red}$ and $\alpha_{exp}$ are non-negative coefficients that represent the relative importance of each aspect of team formation~\cite{bhowmik2014}. The goal is to find a team $T^{\prime}\subseteq S$ where the \textit{non-negative} and \textit{non-monotone} function $s(T^{\prime})$ is maximized. According to ~\cite{bhowmik2014}, this formulation can be posed as an \textit{unconstrained submodular function maximization problem} which is NP-hard for arbitrary submodular functions.

Selecting a set of canned patterns in \textsc{cps} is akin to hiring a team of individuals in \textsc{tfp} where $f_{skill}$, $f_{red}$, $f_{team}$ correspond to $f_{cov}$, $f_{sim}$ and $f_{cog}$, respectively. Hence, \textsc{cps} can be formulated in the form $s(P^{\prime})=\alpha_{f_{cov}}f_{cov}(P^{\prime})-\alpha_{f_{sim}}f_{sim}(P^{\prime})-\alpha_{f_{cog}}f_{cog}(P^{\prime})$ (Definition~\ref{def:patternScore}) where $P^{\prime}$ is the set of\emph{ candidate} patterns which yields an optimized $s(P^{\prime})$.\eat{  The constants in $s(P^{\prime})$ ensure non-negative and non-monotone behaviour (Theorem~\ref{thm:scoreNonNegativeNonMonotone}).}

\vspace{-1ex} \begin{definition}\label{def:patternScore} \textbf{[Pattern Set Score]}
	{\em Given a pattern set $\mathcal{P}^{\prime}$, the \textbf{score} of $\mathcal{P}^{\prime}$ is  $s(\mathcal{P}^{\prime})=\frac{1}{3|\mathcal{P}^{\prime}|}(f_{cov}(\mathcal{P}^{\prime})-f_{sim}(\mathcal{P}^{\prime})-f_{cog}(\mathcal{P}^{\prime})+2|\mathcal{P}^{\prime}|)$ where $f_{cov}$, $f_{sim}$ and $f_{cog}$ are the coverage, similarity and cognitive load of $\mathcal{P}^{\prime}$, respectively.}
\end{definition}
\vspace{-1ex}\begin{definition}\label{def:goodPattern} \textbf{[Good Candidate Pattern]}
	{\em Given a pattern set $\mathcal{P}^{\prime}$ and two candidate patterns $p_1$ and $p_2$, $p_1$ is considered a \textbf{good candidate pattern} if $s(\mathcal{P}^{\prime}\bigcup p_1)>s(\mathcal{P}^{\prime}\bigcup p_2)$ and is added to $\mathcal{P}^{\prime}$ instead of $p_2$.}
\end{definition}

Note that Definition~\ref{def:goodPattern} can be utilized for determining inclusion of a candidate pattern in $\mathcal{P}$.
Next, we analyze the  properties of $f_{cov}$, $f_{sim}$, $f_{cog}$, and the pattern score.

\vspace{-1ex} \begin{lemma}\label{lem:coverageSubmodular}
	\textit{Coverage of a pattern set $\mathcal{P}$, $f_{cov}(\mathcal{P})$, is submodular.}
\end{lemma}

\eat{\begin{proof}
	Given a set of $n$ elements ($N$), a function $f(.)$ is submodular if for every $A\subseteq B\subseteq N$ and $j\notin B$, $f(A\bigcup\{ j\})-f(A)\geq f(B\bigcup \{j\})-f(B)$. Given a graph $G$ and canned pattern sets $\mathcal{P}_A$ and $\mathcal{P}_B$ where $\mathcal{P}_A\subseteq\mathcal{P}_B$, let the coverage of $\mathcal{P}_A$ and $\mathcal{P}_B$ be $f_{cov}(\mathcal{P}_A)$ and $f_{cov}(\mathcal{P}_B)$, respectively. Observe that $\mathcal{P}_B$ consists of $\mathcal{P}_A$ and additional patterns (\ie $\mathcal{P}^{\prime}=\mathcal{P}_B\setminus\mathcal{P}_A$). For each canned pattern $p\in\mathcal{P}^{\prime}$, we let $s=\min(|f_{cov}(p)|,|f_{cov}(\mathcal{P}_A)|)$ and $K$ denotes the overlapping set $f_{cov}(p)\bigcap f_{cov}(\mathcal{P}_A)$. The coverage of $p$ falls under one of the four possible scenarios: (1) $K=f_{cov}(p)$ if $s=|f_{cov}(p)|$, (2) $K=f_{cov}(\mathcal{P}_A)$ if $s=|f_{cov}(\mathcal{P}_A)|$, (3) $K$ is an empty set and (4) otherwise (\ie $0<|f_{cov}(p)\bigcap f_{cov}(\mathcal{P}_A)|<s$).
	
	In the case where coverage of every $p$ falls under scenario 1, $f_{cov}(\mathcal{P}_A)=f_{cov}(\mathcal{P}_B)$. If any $p$ falls under scenario 2, 3 or 4, then $f_{cov}(\mathcal{P}_A)\subset f_{cov}(\mathcal{P}_B)$. Hence, $f_{cov}(\mathcal{P}_A)\subseteq f_{cov}(\mathcal{P}_B)$. Consider a pattern $p^{\prime}\notin\mathcal{P}_B$, let $t=\min(|f_{cov}(p^{\prime})|,\linebreak |f_{cov}(\mathcal{P}_A)|)$. Suppose $f_{cov}(p^{\prime})\bigcap f_{cov}(\mathcal{P}_A)=f_{cov}(p^{\prime})$ where $|f_{cov}(p^{\prime})|<|f_{cov}(\mathcal{P}_A)|$ (Scenario 1), then $f_{cov}(\mathcal{P}_A\bigcup \{p^{\prime}\})-f_{cov}(\mathcal{P}_A)$ is an empty set. Note that we use the minus and set minus operator interchangeably in this proof. Since $f_{cov}(\mathcal{P}_A)\subseteq f_{cov}(\mathcal{P}_B)$, $f_{cov}(\mathcal{P}_B\bigcup \{p^{\prime}\})=f_{cov}(\mathcal{P}_B)$. Hence, $f_{cov}(\mathcal{P}_A\bigcup \{p^{\prime}\})-f_{cov}(\mathcal{P}_A)= f_{cov}(\mathcal{P}_B\bigcup\{ p^{\prime}\})-f_{cov}(\mathcal{P}_B)$.
	
	Now, consider $f_{cov}(p^{\prime})\bigcap f_{cov}(\mathcal{P}_A)=f_{cov}(\mathcal{P}_A)$ where $|f_{cov}(p^{\prime})|\linebreak>|f_{cov}(\mathcal{P}_A)|$ (Scenario 2). $f_{cov}(\mathcal{P}_A\bigcup\{ p^{\prime}\})-f_{cov}(\mathcal{P}_A)=f_{cov}(p^{\prime})-f_{cov}(\mathcal{P}_A)$ where $f_{cov}(\mathcal{P}_A)\subset f_{cov}(p^{\prime})$. Let $L$ and $M$ be $f_{cov}(p^{\prime})\setminus f_{cov}(\mathcal{P}_A)$ and $f_{cov}(\mathcal{P}_B)\setminus f_{cov}(\mathcal{P}_A)$, respectively. Observe that, similar to previous observation, it is possible for (1) $L$ to be fully contained in $M$ if $|L|<|M|$, (2) $M$ to be fully contained in $L$ if $|M|<|L|$, (3) $L\bigcap M$ to be empty or (4) otherwise (\ie $0<|L\bigcap M|<t$ where $t=\min(|L|,|M|)$). Hence, $|L\bigcap M|\in[0,t]$. When $|L\bigcap M|=0$, $f_{cov}(\mathcal{P}_A\bigcup \{p^{\prime}\})-f_{cov}(\mathcal{P}_A)= f_{cov}(\mathcal{P}_B\bigcup \{p^{\prime}\})-f_{cov}(\mathcal{P}_B)$. Otherwise, there are some common graphs covered by $L$ and $M$, resulting in $f_{cov}(\mathcal{P}_B\bigcup \{p^{\prime}\})-f_{cov}(\mathcal{P}_B)=L\setminus (L\bigcap M)$. Hence, $|f_{cov}(\mathcal{P}_A\bigcup \{p^{\prime}\})-f_{cov}(\mathcal{P}_A)|>|f_{cov}(\mathcal{P}_B\bigcup \{p^{\prime}\})-f_{cov}(\mathcal{P}_B)|$. Taken together, for scenario 2, $|f_{cov}(\mathcal{P}_A\bigcup \{p^{\prime}\})-f_{cov}(\mathcal{P}_A)|\geq|f_{cov}(\mathcal{P}_B\bigcup \{p^{\prime}\})-f_{cov}(\mathcal{P}_B)|$.
	
	Scenario 3 is similar to scenario 2 where $L$ is $f_{cov}(p^{\prime})$\eat{ instead of $f_{cov}(p^{\prime})\setminus f_{cov}(\mathcal{P}_A)$}. $f_{cov}(\mathcal{P}_A\bigcup \{p^{\prime}\})-f_{cov}(\mathcal{P}_A)=L$ and $f_{cov}(\mathcal{P}_B\bigcup \{p^{\prime}\})-f_{cov}(\mathcal{P}_B)\linebreak =L\setminus (L\bigcap M)$. Since $|L\bigcap M|\in[0,t]$, $|f_{cov}(\mathcal{P}_A\bigcup \{p^{\prime}\})-f_{cov}(\mathcal{P}_A)|\geq |f_{cov}(\mathcal{P}_B\bigcup \{p^{\prime}\})-f_{cov}(\mathcal{P}_B)|$.
	
	Scenario 4 is the same as scenario 3 except that $L=f_{cov}(p^{\prime})\setminus (f_{cov}(\mathcal{P}_A)\bigcap f_{cov}(p^{\prime}))$. Observe that $|f_{cov}(\mathcal{P}_A\bigcup \{p^{\prime}\})-f_{cov}(\mathcal{P}_A)|\geq |f_{cov}(\mathcal{P}_B\bigcup \{p^{\prime}\})-f_{cov}(\mathcal{P}_B)|$ as $|L\bigcap M|\in[0,t]$.
	
	Hence, in all cases, $|f_{cov}(\mathcal{P}_A\bigcup \{p^{\prime}\})-f_{cov}(\mathcal{P}_A)|\geq \linebreak|f_{cov}(\mathcal{P}_B\bigcup \{p^{\prime}\}) -f_{cov}(\mathcal{P}_B)|$ and $f_{cov}(.)$ is submodular.
\end{proof}}

\eat{\begin{proof}Submodular functions satisfies the property of diminishing marginal returns. That is given a set of $n$ elements ($N$), a function $f(.)$ is submodular if for every $A\subseteq B\subseteq N$ and $j\notin B$, $f(A\bigcup j)-f(A)\geq f(B\bigcup j)-f(B)$. Given a graph $G$ and canned pattern sets $\mathcal{P}_A$ and $\mathcal{P}_B$ where $\mathcal{P}_A\subseteq\mathcal{P}_B$, let the coverage of $\mathcal{P}_A$ and $\mathcal{P}_B$ be $f_{cov}(\mathcal{P}_A)$ and $f_{cov}(\mathcal{P}_B)$, respectively. Observe that $\mathcal{P}_B$ consists of $\mathcal{P}_A$ and additional patterns (\ie $\mathcal{P}^{\prime}=\mathcal{P}_B\setminus\mathcal{P}_A$). For each canned pattern $p\in\mathcal{P}^{\prime}$, we let $s=\min(|f_{cov}(p)|,|f_{cov}(\mathcal{P}_A)|)$ and $K$ denotes the overlapping set $f_{cov}(p)\bigcap f_{cov}(\mathcal{P}_A)$. The coverage of $p$ falls under one of four possible scenarios, namely, (1) $K=f_{cov}(p)$ if $s=|f_{cov}(p)|$, (2) $K=f_{cov}(\mathcal{P}_A)$ if $s=|f_{cov}(\mathcal{P}_A)|$, (3) $K$ is an empty set and (4) otherwise (\ie $0<|f_{cov}(p)\bigcap f_{cov}(\mathcal{P}_A)|<s$).
		
		In the case where coverage of every $p$ falls under scenario 1, then $f_{cov}(\mathcal{P}_A)=f_{cov}(\mathcal{P}_B)$. Should any $p$ falls under scenario 2, 3 or 4, then $f_{cov}(\mathcal{P}_A)\subset f_{cov}(\mathcal{P}_B)$. Hence, $f_{cov}(\mathcal{P}_A)\subseteq f_{cov}(\mathcal{P}_B)$. Consider a canned pattern $p^{\prime}\notin\mathcal{P}_B$, let $t=\min(|f_{cov}(p^{\prime})|,|f_{cov}(\mathcal{P}_A)|)$. Suppose $f_{cov}(p^{\prime})\bigcap f_{cov}(\mathcal{P}_A)=f_{cov}(p^{\prime})$ where $|f_{cov}(p^{\prime})|<|f_{cov}(\mathcal{P}_A)|$ (Scenario 1), then $f_{cov}(\mathcal{P}_A\bigcup p^{\prime})-f_{cov}(\mathcal{P}_A)$ is an empty set. Note that we use the minus and set minus operator interchangeably in this proof. Since $f_{cov}(\mathcal{P}_A)\subseteq f_{cov}(\mathcal{P}_B)$, $f_{cov}(\mathcal{P}_B\bigcup p^{\prime})=f_{cov}(\mathcal{P}_B)$. Hence, $f_{cov}(\mathcal{P}_A\bigcup p^{\prime})-f_{cov}(\mathcal{P}_A)= f_{cov}(\mathcal{P}_B\bigcup p^{\prime})-f_{cov}(\mathcal{P}_B)$.
		
		Now, consider $f_{cov}(p^{\prime})\bigcap f_{cov}(\mathcal{P}_A)=f_{cov}(\mathcal{P}_A)$ where $|f_{cov}(p^{\prime})|>|f_{cov}(\mathcal{P}_A)|$ (Scenario 2). $f_{cov}(\mathcal{P}_A\bigcup p^{\prime})-f_{cov}(\mathcal{P}_A)=f_{cov}(p^{\prime})-f_{cov}(\mathcal{P}_A)$ where $f_{cov}(\mathcal{P}_A)\subset f_{cov}(p^{\prime})$. Let $L$ and $M$ be $f_{cov}(p^{\prime})\setminus f_{cov}(\mathcal{P}_A)$ and $f_{cov}(\mathcal{P}_B)\setminus f_{cov}(\mathcal{P}_A)$, respectively. Observe that, similar to previous observation, it is possible for (1) $L$ to be fully contained in $M$ if $|L|<|M|$, (2) $M$ to be fully contained in $L$ if $|M|<|L|$, (3) $L\bigcap M$ to be empty or (4) otherwise (\ie $0<|L\bigcap M|<t$ where $t=\min(|L|,|M|)$). Hence, $|L\bigcap M|\in[0,t]$. When $|L\bigcap M|=0$, $f_{cov}(\mathcal{P}_A\bigcup p^{\prime})-f_{cov}(\mathcal{P}_A)= f_{cov}(\mathcal{P}_B\bigcup p^{\prime})-f_{cov}(\mathcal{P}_B)$. Otherwise, there are some common graphs covered by $L$ and $M$, resulting in $f_{cov}(\mathcal{P}_B\bigcup p^{\prime})-f_{cov}(\mathcal{P}_B)=L\setminus (L\bigcap M)$. Hence, $|f_{cov}(\mathcal{P}_A\bigcup p^{\prime})-f_{cov}(\mathcal{P}_A)|>|f_{cov}(\mathcal{P}_B\bigcup p^{\prime})-f_{cov}(\mathcal{P}_B)|$. Taken together, for scenario 2, $|f_{cov}(\mathcal{P}_A\bigcup p^{\prime})-f_{cov}(\mathcal{P}_A)|\geq |f_{cov}(\mathcal{P}_B\bigcup p^{\prime})-f_{cov}(\mathcal{P}_B)|$.
		
		For scenario 3, it is similar to scenario 2 where $L$ is $f_{cov}(p^{\prime})$ instead of $f_{cov}(p^{\prime})\setminus f_{cov}(\mathcal{P}_A)$. $f_{cov}(\mathcal{P}_A\bigcup p^{\prime})-f_{cov}(\mathcal{P}_A)=L$ and $f_{cov}(\mathcal{P}_B\bigcup p^{\prime})-f_{cov}(\mathcal{P}_B)=L\setminus (L\bigcap M)$. Since $|L\bigcap M|\in[0,t]$, $|f_{cov}(\mathcal{P}_A\bigcup p^{\prime})-f_{cov}(\mathcal{P}_A)|\geq |f_{cov}(\mathcal{P}_B\bigcup p^{\prime})-f_{cov}(\mathcal{P}_B)|$.
		
		For scenario 4, it is the same as scenario 3 except that $L=f_{cov}(p^{\prime})\setminus (f_{cov}(\mathcal{P}_A)\bigcap f_{cov}(p^{\prime}))$. Observe that $|f_{cov}(\mathcal{P}_A\bigcup p^{\prime})-f_{cov}(\mathcal{P}_A)|\geq |f_{cov}(\mathcal{P}_B\bigcup p^{\prime})-f_{cov}(\mathcal{P}_B)|$ due to $|L\bigcap M|\in[0,t]$.
		
		Hence, in all cases, $|f_{cov}(\mathcal{P}_A\bigcup p^{\prime})-f_{cov}(\mathcal{P}_A)|\geq |f_{cov}(\mathcal{P}_B\bigcup p^{\prime})-f_{cov}(\mathcal{P}_B)|$ applies and $f_{cov}(.)$ is submodular.
\end{proof}}

\begin{lemma}\label{lem:similaritySubmodular}
	\textit{The similarity (resp. cognitive load) of a pattern set $\mathcal{P}$, $f_{sim}(\mathcal{P})$ (resp. $f_{cog}(\mathcal{P})$), is supermodular.}
\end{lemma}

\eat{\begin{proof}
	Given a submodular function $f(.)$, for every $\mathcal{P}_A\subseteq\mathcal{P}_B\subseteq D$ and every $p\subset D$ s.t $p\notin\mathcal{P}_A,\mathcal{P}_B$, the first order difference states that $f(\mathcal{P}_A\bigcup \{p\})-f(\mathcal{P}_A)\geq f(\mathcal{P}_B\bigcup \{p\})-f(\mathcal{P}_B)$. Given a graph $G$, a canned pattern $p\notin\mathcal{P}_B$ and canned pattern sets $\mathcal{P}_A$ and $\mathcal{P}_B$ where $\mathcal{P}_A\subseteq\mathcal{P}_B$, let the similarity of $\mathcal{P}_A$ and $\mathcal{P}_B$ be $f_{sim}(\mathcal{P}_A)$ and $f_{sim}(\mathcal{P}_B)$, respectively. $f_{sim}(\mathcal{P}_B\bigcup \{p\})-f_{sim}(\mathcal{P}_B)=\sum_{p_i\in\mathcal{P}_B}sim(p,p_i)$ and $f_{sim}(\mathcal{P}_A\bigcup \{p\})-f_{sim}(\mathcal{P}_A)=\sum_{p_i\in\mathcal{P}_A}sim(p,p_i)$. Since \linebreak $sim(p_i,p_j)\geq 0$ $\forall p_i,p_j\subset G$, $\mathcal{P}_A\subseteq\mathcal{P}_B$ and by definition of the first order difference, $f_{sim}(.)$ is supermodular.
	
	The proof is similar for $f_{cog}(.)$.
\end{proof}}

\eat{\begin{proof}
		We begin by stating the \textit{first order difference}. Given a submodular function $f(.)$, for every $\mathcal{P}_A\subseteq\mathcal{P}_B\subseteq D$ and every $p\subset D$ such that $p\notin\mathcal{P}_A,\mathcal{P}_B$, the first order difference states that $f(\mathcal{P}_A\bigcup p)-f(\mathcal{P}_A)\geq f(\mathcal{P}_B\bigcup p)-f(\mathcal{P}_B)$.
		
		Given a graph $G$, a canned pattern $p\notin\mathcal{P}_B$ and canned pattern sets $\mathcal{P}_A$ and $\mathcal{P}_B$ where $\mathcal{P}_A\subseteq\mathcal{P}_B$, let the similarity of $\mathcal{P}_A$ and $\mathcal{P}_B$ be $f_{sim}(\mathcal{P}_A)$ and $f_{sim}(\mathcal{P}_B)$, respectively. $f_{sim}(\mathcal{P}_B\bigcup p)-f_{sim}(\mathcal{P}_B)=\linebreak\sum_{p_i\in\mathcal{P}_B}sim(p,p_i)$ and $f_{sim}(\mathcal{P}_A\bigcup p)-f_{sim}(\mathcal{P}_A)=\sum_{p_i\in\mathcal{P}_A}sim(p,p_i)$. Since $sim(p_i,p_j)\geq 0$ $\forall p_i,p_j\subset G$, $\mathcal{P}_A\subseteq\mathcal{P}_B$ and by definition of the first order difference, $f_{sim}(.)$ is supermodular.
\end{proof}}

\eat{\begin{lemma}\label{lem:cognitiveSubmodular}
		The cognitive load of a pattern set $\mathcal{P}$ denoted as $f_{cog}(\mathcal{P})$ is supermodular.
	\end{lemma}
	
	\begin{proof}
		The proof is similar to that of Lemma~\ref{lem:similaritySubmodular} by replacing $f_{sim}(.)$ with $f_{cog}(.)$
\end{proof}}

\begin{theorem}\label{thm:scoreNonNegativeNonMonotone}
	\textit{The pattern set score  $s(\mathcal{P}^\prime)$ in Definition~\ref{def:patternScore} is a non-negative and non-monotone submodular function.}
\end{theorem}

\eat{\begin{proof} (Sketch).
	Consider a partial pattern set $\mathcal{P}^{\prime}$ and a candidate pattern $p$. Suppose $p$ does not improve the set coverage of $\mathcal{P}^{\prime}$ and adds a high cost in terms of cognitive load and diversity. Then, $s(\mathcal{P}^{\prime})>s(\mathcal{P}^{\prime}\bigcup \{p\})$. Hence, $s(.)$ is \textit{non-monotone}. Since $f_{cov}(\mathcal{P}^{\prime}),f_{sim}(\mathcal{P}^{\prime}), f_{cog}(\mathcal{P}^{\prime})\in[0,|\mathcal{P}^{\prime}|]$, $f_{cov}(\mathcal{P}^{\prime})-f_{sim}(\mathcal{P}^{\prime})-f_{cog}(\mathcal{P}^{\prime})$ is in the range [-2$|\mathcal{P}^{\prime}|$,$|\mathcal{P}^{\prime}|$]. Hence, $\frac{1}{3|\mathcal{P}^{\prime}|}(f_{cov}(\mathcal{P}^{\prime})-f_{sim}(\mathcal{P}^{\prime})-f_{cog}(\mathcal{P}^{\prime})+2|\mathcal{P}^{\prime}|)$ is in the range $[0,1]$ and is \textit{non-negative}. Since supermodular functions are negations of submodular functions and non-negative weighted sum of submodular functions preserve submodular property \cite{fujishige2005}, $s(\mathcal{P}^{\prime})$ is submodular. Note that adding a positive constant (\ie $\frac{2}{3}$) does not change the submodular property~\cite{bhowmik2014} and ensures that $s(\mathcal{P}^{\prime})$ is non-negative. The scaling factors of $\alpha_{f_{cov}}=\alpha_{f_{sim}}=\alpha_{f_{cog}}=\frac{1}{3|\mathcal{P}^{\prime}|}$ further bound $s(\mathcal{P}^{\prime})$ within the range $[0,1]$.
\end{proof}}

\eat{\begin{proof}(Sketch). Consider a partial pattern set $\mathcal{P}^{\prime}$ and a candidate pattern $p$. Suppose $p$ does not improve the set coverage of $\mathcal{P}^{\prime}$ and adds a high cost in terms of cognitive load and diversity. Then, $s(\mathcal{P}^{\prime})>s(\mathcal{P}^{\prime}\bigcup p)$. Hence, the score function $s(.)$ is \textit{non-monotone}. Since $f_{cov}(\mathcal{P}^{\prime}),f_{sim}(\mathcal{P}^{\prime}),f_{cog}(\mathcal{P}^{\prime})\in[0,|\mathcal{P}^{\prime}|]$, $f_{cov}(\mathcal{P}^{\prime})-f_{sim}(\mathcal{P}^{\prime})-f_{cog}(\mathcal{P}^{\prime})$ is in the range [-2$|\mathcal{P}^{\prime}|$,$|\mathcal{P}^{\prime}|$]. Hence, $\frac{1}{3|\mathcal{P}^{\prime}|}(f_{cov}(\mathcal{P}^{\prime})-f_{sim}(\mathcal{P}^{\prime})-f_{cog}(\mathcal{P}^{\prime})+2|\mathcal{P}^{\prime}|)$ (Definition~\ref{def:patternScore}) is in the range $[0,1]$ and is \textit{non-negative}. Since supermodular functions are negations of submodular functions and that non-negative weighted sum of submodular functions preserve submodular property \cite{fujishige2005}, $s(\mathcal{P}^{\prime})$ is submodular. Note that adding a constant (\ie $\frac{2}{3}$) does not change the submodular property~\cite{bhowmik2014} and ensures that $s(\mathcal{P}^{\prime})$ is non-negative. The scaling factors of $\alpha_{f_{cov}}=\alpha_{f_{sim}}=\alpha_{f_{cog}}=\frac{1}{3|\mathcal{P}^{\prime}|}$ further bounds $s(\mathcal{P}^{\prime})$ within the range $[0,1]$.
\end{proof}}

\eat{Note that $s(\mathcal{P}^{\prime})$ is non-monotone in general as exemplified by adding a new candidate pattern that does not improve the set coverage, but adds a high cost in terms of cognitive load and diversity. Further, observe that using appropriate lower bounds for the coefficients allow us to transform $s(\mathcal{P}^{\prime})$ into a non-negative function. Since $f_{cov},f_{sim},f_{cog}\in[0,1]$, $f_{cov}(P^{\prime})-f_{sim}(P^{\prime})-f_{cog}(P^{\prime})$ is in the range [-2,1]. We transform $s(P^{\prime})$ into $s(P^{\prime})=\frac{1}{3}(f_{cov}(P^{\prime})-f_{sim}(P^{\prime})-f_{cog}(P^{\prime})+2)$ (Definition~\ref{def:patternScore}) to ensure that $s(P^{\prime})\in[0,1]$ is \textit{non-negative} and \textit{non-monotone}. Note that adding a constant (\ie $\frac{2}{3}$) does not change the submodular property~\cite{bhowmik2014} and ensures that $s(P^{\prime})$ is non-negative. The scaling factors of $\alpha_{f_{cov}}=\alpha_{f_{sim}}=\alpha_{f_{cog}}=\frac{1}{3}$ further bounds $s(P^{\prime})$ within the range $[0,1]$. Note that the pattern set score can be used to compare a given set of candidate patterns to determine which ones to be selected for inclusion in the canned pattern set (Definition~\ref{def:goodPattern}). We shall revisit these in Section~\ref{sec:cannedPatternSelector}.
}

\vspace{-1ex}Similar to $s(T^{\prime})$ in \textsc{tfp}, $s(\mathcal{P}^{\prime})$ in \textsc{cps} is non-negative and non-monotone. However, unlike \textsc{tfp}, \textsc{cps} imposes a cardinality constraint where $|\mathcal{P}|$ is at most $\gamma$. Thus, \textsc{cps} can be posed instead as a maximization of submodular function problem subject to cardinality constraint~\cite{buchbinder2014}.\eat{ Hardness of such problems is $\frac{1}{2}$ when $k=\frac{|\mathcal{P}_{cand}|}{2}$~\cite{vondrak2013} where $\mathcal{P}_{cand}$ is the set of candidate patterns.\eat{ That is, unless \textsc{p}=\textsc{np}, there is no polynomial-time approximation for \textsc{cps} that achieves better than $\frac{1}{2}$-approximation factor.}\eat{ In particular, two approaches have been proposed to address such problems~\cite{buchbinder2014}, namely, simple randomized greedy approach and continuous double greedy approach. These approaches guarantee $\frac{1}{e}$-approximation and $\frac{1}{2}$-approximation, respectively.}
	
	\begin{theorem}\label{thm:cannedPatternApproximation}
		The hardness of the \textsc{cps} approximation problem is $\frac{1}{2}$.
	\end{theorem}
	
	The detailed proof of the above theorem is provided in ~\cite{vondrak2013} and summarized in~\cite{tech}.}

\eat{\begin{proof}
		We first introduce the concepts of \textit{multilinear relaxation}~\cite{vondrak2013}, \textit{symmetry}~\cite{vondrak2013} and \textit{symmetry gap}~\cite{vondrak2013}.
		
		Consider a discrete optimization problem $\max\{s(\mathcal{P}):\mathcal{P}\in\mathcal{F}\}$ where $s:2^{|P_{all}|}\rightarrow R$ is the objective function and $\mathcal{F}\subset 2^{|P_{all}|}$ is the collection of feasible solutions. In case $s$ is a linear function, $s(\mathcal{P})=\sum_{j\in\mathcal{P}}w_j$, the problem can be replaced by a linear programming problem. For general set function $s(\mathcal{P})$, \textit{multilinear relaxation} is used~\cite{vondrak2013}: For $\textbf{x}\in[0,1]^{|P_{all}|}$, let $\textbf{\^{x}}$ denote a random vector in $\{0,1\}^{|P_{all}|}$ where each coordinate of $x_i$ is rounded independently to 1 with probability $x_i$ and 0 otherwise. Let $F(\textbf{x})=\mathbb{E}[s(\textbf{\^{x}})]=\sum_{\mathcal{P}\subseteq P_{all}}s(\mathcal{P})\prod_{i\in \mathcal{P}}x_i\prod_{j\notin\mathcal{P}}(1-x_j)$. Note that $F(\textbf{x})$ can be computed approximately by random sampling. Consider a continuous optimization problem $\max\{F(\textbf{x}):\textbf{x}\in P(\mathcal{F})\}$ where $P(\mathcal{F})=\{\sum_{\mathcal{P}\in\mathcal{F}}\alpha_{\mathcal{P}}
		\textbf{1}_{\mathcal{P}}:\sum_{\mathcal{P}\in\mathcal{F}}\alpha_{\mathcal{P}}=1,\alpha_{\mathcal{P}}\geq 0\}$ is the convex hull of characteristic vectors corresponding to $\mathcal{F}$.
		
		We denote $\alpha$ as the naturally induced mapping of subsets of $P_{all}:\alpha(\mathcal{P})=\{\alpha(i):i\in \mathcal{P}\}$. Observe that $\alpha$ is applied to sets and must be derived from a permutation on $P_{all}$. For $\textbf{x}\in[0,1]^{|P_{all}|}$, the \textit{symmetrization of \textbf{x}} is $\overline{\textbf{x}}=\mathbb{E}_{\alpha\in\mathcal{G}}[\alpha(\mathbf{x})]$ where $\alpha\in\mathcal{G}$ is uniformly random, $\mathcal{G}$ is the symmetry group and $\alpha(\textbf{x})$ denotes $\mathbf{x}$ with coordinates permuted by $\alpha$ permutations on $\mathcal{G}$. An instance $\max\{s(\mathcal{P}):\mathcal{P}\in\mathcal{F}\}$ on a ground set $P_{all}$ is strongly \textit{symmetric} with respect to a group of permutations $\mathcal{G}$ on $P_{all}$ if $s(\mathcal{P})=s(\alpha(\mathcal{P}))$ for all $\mathcal{P}\subseteq P_{all}$ and $\alpha\in\mathcal{G}$, and $\mathcal{P}\in\mathcal{F}\Leftrightarrow\mathcal{P}^{\prime}\in\mathcal{F}$ whenever $\mathbb{E}_{\alpha\in\mathcal{G}}[\mathbf{1}_{\alpha}(\mathcal{P})]=\mathbb{E}_{\alpha\in\mathcal{G}}[\mathbf{1}_{\alpha(\mathcal{P}^{\prime})}]$.
		A cardinality constraint $\mathcal{F}=\{\mathcal{P}\subseteq P_{all}:|\mathcal{P}|\leq \gamma\}$ is strongly symmetric with respect to all permutations because the condition $\mathcal{P}\in\mathcal{F}$ depends only on the symmetrized vector $\overline{\textbf{1}_{\mathcal{P}}}=\frac{|\mathcal{P}|}{|P_{all}|}\textbf{1}$~\cite{vondrak2013}. Let $\max\{s(\mathcal{P}):\mathcal{P}\in\mathcal{F}\}$ be an instance on a ground set $P_{all}$ which is strongly symmetric with respect to $\mathcal{G}\subset \mathcal{S}(P_{all})$ where $\mathcal{S}(P_{all})$ is the symmetric group (of all permutations) on the ground set $P_{all}$. The \textit{symmetry gap} of $\max\{s(\mathcal{P}):\mathcal{P}\in\mathcal{F}\}$ is $\gamma=\overline{OPT}/{OPT}$ where $OPT=\max\{F(\textbf{x}):\textbf{x}\in P(\mathcal{F})\}$ and $\overline{OPT}=\max\{F(\overline{\textbf{x}}):\textbf{x}\in P(\mathcal{F})\}$ and $\overline{\textbf{x}}=\mathbb{E}_{\alpha\in\mathcal{G}}[\alpha(\textbf{x})]$. The symmetry gap for any strongly symmetric instance translates automatically into hardness of approximation for refined instances~\cite{vondrak2013}.
		
		Now, we let $P_{all}=\{P_1,P_2\}$ and for any $\mathcal{P}\subseteq P_{all}$, $s(\mathcal{P})=1$ if $|\mathcal{P}|=1$ and 0 otherwise. Consider the instance $\max\{s(\mathcal{P}):\mathcal{P}\subseteq P_{all}\}$ (\ie Max Cut problem on $K_2$). Note that this instance exhibits a simple symmetry, the group of all (two) permutations on $\{P_1,P_2\}$. $OPT=F(1,0)=F(0,1)=1$ while $\overline{OPT}=F(\frac{1}{2},\frac{1}{2})=\frac{1}{2}$. Hence, the symmetry gap is $\frac{1}{2}$. Since $s(\mathcal{P})$ is nonnegative submodular and there is no constraint on $\mathcal{P}\subseteq P_{all}$, this will be the case for any refinement of the instance as well. The same symmetry gap holds if simple constraints (\eg, $\max\{s(\mathcal{P}):|\mathcal{P}|\leq 1\}$) are imposed. Hence, non-monotone submodular maximization under cardinality constraints of the type $|\mathcal{P}|\leq \frac{|P_{all}|}{2}$ also has hardness of $\frac{1}{2}$.
\end{proof}}

\vspace{0ex}
\subsection{Coverage, Cognitive Load, and Similarity} \label{sec:quan}
Next, we quantify the coverage, cognitive load, and similarity measures used in the pattern score $s(P^\prime)$.

\textbf{Coverage.} Recall from Section~\ref{sec:problem}, we can compute the coverage of a pattern $p$ as $cov_p=|\bigcup_{i\in|S(p)|}E_i|$. Since the edge sets of $G_T=(V_T,E_T)$ and $G_O=(V_O,E_O)$ are mutually exclusive, we further modify $cov_p$ to include a weight factor to account for effects exerted by the sizes of $G_T$ and $G_O$. Specifically, $cov_p=|\bigcup_{i\in|S(p)|}E_i|\frac{|G_x|}{|E|}$ where $G_x \in \{G_T, G_O\}$ for patterns obtained from $G_x$. However, exact computation of coverage for each candidate pattern is prohibitively expensive. Hence, we approximate $cov_p$ as follows: $cov_{ub(p)}=|E_p|\times freq(p)\times\frac{|G_x|}{|E|}$. Observe that $cov_{ub(p)}$ is in fact the upper bound of $cov_p$ when no isomorphic instances of $p$ in $G$ overlap\eat{\footnote{\scriptsize Any superior upper bound that can be computed efficiently can be incorporated as the canned pattern selection is orthogonal to the choice of upper bound.}}. Any superior upper bound that can be computed efficiently can be incorporated. Unlike $cov_p$, computation of $cov_{ub(p)}$ requires only $freq(p)$, which is significantly more efficient.

The order of pattern extraction in $G_O$ (\textit{e.g.}, extracting stars and asterisms before small patterns) may affect the frequency of the extracted patterns. Hence, \textit{normalization} of $cov_{ub}$ is performed for each class of patterns (\eat{$k$-\textsc{cp} and \textsc{ccp}, star and asterism, and small pattern}$k$-\textsc{cp}, \textsc{ccp}, star, asterism, and small pattern) as follows:
\vspace{-1ex}
\begin{equation} \small
	cov_{ub(p)}= \frac{cov^{\prime}_{ub(p)}-Min(cov^{\prime}_{ub}(P_{t}))+1}{Max(cov^{\prime}_{ub}(P_{t}))-Min(cov^{\prime}_{ub}(P_{t}))+1}
\end{equation}
where $t \in \eat{\{truss, star, path, cycle, small\}}\{k-CP, CCP, star, asterism, small\}$ represents a class of pattern. Specifically, we compute $k$-\textsc{cp}s and \textsc{ccp}s in $G_T$.\eat{; stars and asterisms; paths; cycles; and small patterns, respectively} Stars, asterisms and small patterns are computed in $G_O$. The normalized $cov_{ub}$ is in [0-1].

\textbf{Cognitive Load.}  \cite{catapult,midas} measure cognitive load based on size and density only, ignoring edge crossings. Since it is designed for a collection of small- or medium-sized data graphs, it is a reasonable measure as in many applications such data graphs have very few edge crossings (\eg chemical compounds), if any. In contrast, edge crossings occur frequently in large networks and hence cannot be ignored in our context. In fact, Huang and colleagues examined the effect of edge crossings on mental load of users and found that cognitive load displays a relationship with edge crossings that resembles the logistic curve~\cite{huang2010} $f(x)=\frac{L}{1+e^{-k(x-x_0)}}$ where $L$ is the curve's maximum value, $x_0$ is the $x$ value of sigmoid's midpoint and $k$ is the logistic growth rate \cite{zeide1993}.

\vspace{-1ex}\begin{lemma}\label{lem:crossing}
	\textit{The crossing number (\ie number of edge crossings) of any simple graph $G=(V,E)$ with at least 3 vertices satisfies $cr\geq|E|-3|V|+6$.}
\end{lemma}

\eat{\begin{proof} (Sketch)
	Consider a graph $G=(V,E)$ with $cr$ crossings. Since each crossing can be removed by removing an edge from $G$, a graph with $|E|-cr$ edges and $|V|$ vertices contains no crossings (\ie planar graph). Since $|E|\leq 3|V|-6$ for the  planar graph (\ie Euler's formula), hence, $|E|-cr\leq 3|V|-6$ for $|V|\geq 3$. Rewriting the inequality, we have $cr\geq|E|-3|V|+6$.
\end{proof}}

\eat{\begin{lemma}\label{lem:crossing}
		\emph{The crossing number of $G=(V,E)$ satisfies $cr\geq\frac{1}{33.75}\frac{E^3}{V^2}$ for $E\geq 7.5V$~\cite{pach1997}}.
	\end{lemma}
	
	\begin{proof} (Sketch)
		Let $G$ be a simple graph drawn in the plane with $cr$ crossings and suppose that $E\geq 7.5V$. Consider construction of a random subgraph $G^{\prime}\subseteq G$ by selecting each vertex of $G$ independently with probability $p=\frac{7.5V}{E}\leq 1$, and let $G^{\prime}=(V^{\prime},E^{\prime})$ be the subgraph induced by the selected vertices. The expected number of vertices of $G^{\prime}$, $\mathbb{E}[V^{\prime}]=pV$. Similarly, $\mathbb{E}[E^{\prime}]=p^2E$. Further, the expected number of crossings in the drawing of $G^{\prime}$ inherited from $G$ (denoted $cr_{G^{\prime}}$) is $p^4cr$ and the expected value of the crossing number of $G^{\prime}$ is even smaller~\cite{pach1997}. By Corollary~\ref{cor:crossing}, $cr_{G^{\prime}}\geq 5E^{\prime}-25V^{\prime}$ for every $G^{\prime}$. Taking expectations, $p^4cr\geq\mathbb{E}[cr_{G^{\prime}}]\geq 5\mathbb{E}[E^{\prime}]-25\mathbb{E}[V^{\prime}]=5p^2E-25pV$. Hence, $cr\geq\frac{1}{33.75}\frac{E^3}{V^2}$ for $E\geq 7.5V$.
	\end{proof}
}

\eat{In \textsc{Tattoo}, cognitive load of a pattern $p$ is determined by size ($sz_p=|E_p|$), density ($d_p=2\frac{|E_p|}{|V_p|(|V_p|-1)}$) and edge crossing ($cr_p$). Note that $cr_p=0$ if $p$ is planar and $cr_p=\frac{1}{33.75}\frac{|E_p|^3}{|V_p|^2}$ (Lemma~\ref{lem:crossing}) \cite{pach1997} if otherwise. Hence, the normalized cognitive load function in \textsc{Tattoo} is defined as $cog_p=\frac{1}{3}\sum_{x\in\{sz_p,d_p,cr_p\}}(1-e^{-x})$. Note that $cog_p\in[0-1]$. In Section~\ref{sec:expt} (Exp 1), we evaluate this cognitive load function against others.}

Hence cognitive load of a pattern $p$ is computed based on the size ($sz_p=|E_p|$), density ($d_p=2\frac{|E_p|}{|V_p|(|V_p|-1)}$) and edge crossing ($cr_p$). $cr_p=0$ if $p$ is planar. Otherwise, it is $cr_p=|E_p|-3|V_p|+6$. We modelled the normalized cognitive load function in \textsc{Tattoo} according to the logistic curve:
\vspace{0ex}
\eat{\begin{equation} \small
		cog_p=\frac{1}{3}\sum_{x\in\{sz_p,d_p,cr_p\}}(1-e^{-x})
	\end{equation}
	\begin{equation} \small
		cog_p=1-e^{-(sz_p+d_p+cr_p)}
\end{equation}}
\begin{equation} \small
	cog_p=1/(1+e^{-0.5\times(sz_p+d_p+cr_p-10)})
\end{equation}
Parameters of $cog_p$ are set empirically to ensure even distribution within the range of [0 1].

\textbf{Similarity.} Given a partial pattern set $\mathcal{P}^\prime$ and two candidate patterns $p_1$ and $p_2$, \textsc{Tattoo} selects $p_1$ preferentially to add to $\mathcal{P}^\prime$ if $\max_{p\in\mathcal{P}^\prime} sim(p_1,p)<\max_{p\in\mathcal{P}^\prime} sim(p_2,p)$. To this end, we utilize \textit{NetSimile}, a size-independent graph similarity approach based on distance between feature vectors \cite{berlingerio2013}\eat{, for the following reasons. First, \textit{NetSimile} leverages local-level (\eg node-level and egonet-level features) network similarity which are more interpretable. Second, it is computationally less expensive to compute local features as compared to global features~\cite{berlingerio2013}. Hence,}. It is scalable with runtime complexity linear to the number of edges.\eat{ Briefly, for every node, \textit{NetSimile} extracts 7 features, namely, degree, clustering coefficient, average degree of neighbours, average clustering coefficient of neighbours, number of edges in ego-network, number of outgoing edges of ego-network and number of neighbours of ego-network. Aggregator functions (\ie median, mean, standard deviation, skewness and kurtosis) are then applied on each local feature to generate the ``signature'' vector for a graph. Similarity score of two given graphs is normalized Canberra distance (range in [0-1]) between their ``signature'' vectors. Nevertheless, any superior and efficient network similarity technique can be adopted.}

\eat{Given two candidate patterns $p_1$ and $p_2$ of same size, $p_1$ is considered more diverse than $p_2$ if more steps are required to reconstruct $p_1$ compared to $p_2$ using existing patterns in the (partial) canned pattern set $\mathcal{P}$. The problem of reconstructing a pattern $p$ using $\mathcal{P}$ is equivalent to the set cover problem which is NP-hard \cite{korte2012}. Hence, \textsc{Tattoo} computes $div$ by leveraging a greedy algorithm \cite{chvatal1979}.}

\eat{\begin{lemma}\label{lem:greedyAlgo}
		The greedy algorithm is an $H(n)$ factor approximation algorithm for the minimum set cover problem, where $H(n)=\sum^{n}_{j=1}\frac{1}{j}\approx log n$.
\end{lemma}}

\eat{Briefly, for each candidate pattern $p$, \textsc{Tattoo} iterates through the current $\mathcal{P}$ to look for existing patterns that are subgraph isomorphic to $p$. Then, it greedily selects patterns from it to reconstruct $p$. \textit{Diversity} of $p$ is simply the ratio of the minimum number of patterns used ($step$) to the size of $p$ (\textit{i.e.}, $div_p=\frac{step}{|E_p|}$). Note that the algorithm is guaranteed to terminate as $\mathcal{P}$ contains a $1$-path default pattern, which is essentially the atomic building block of any canned pattern.}

\eat{\begin{theorem}\label{lem:complexityOfPatternMining} \color{red}**UPDATE AFTER COG FUNCTION AND ALGO FINALIZED**\color{black}
		\emph{The worst-case time and space complexities of Algorithm~\ref{alg:patternMining} are $O(\gamma|P_{all}|(\gamma|V_{p(max)}|!|V_{p(max)}|+|E_{p(max)}|))$ and $O((|P_{all}|+\gamma)(|V_{p(max)}|+|E_{p(max)}|))$, respectively, where $|V_{p(max)}|$ and $|E_{p(max)}|$ are the number of vertices and edges in the largest candidate pattern.
	}\end{theorem}
	
	\vspace{-0.5ex}\textbf{Remark.} Note that in Theorem~\ref{lem:complexityOfPatternMining} the exponential cost is due to the subgraph isomorphism check during diversity computation. In practice, it can be performed fast as candidate patterns are small in size.\eat{ Also, our data-driven approach for canned pattern selection enables a user to personalize the \textsc{gui} by choosing the pattern budget and \textsc{Tattoo} can select the patterns accordingly.}
}

\eat{\begin{algorithm}[!t]
		\algsetup{
			linenosize=\scriptsize
		}
		\begin{algorithmic}[1]
			\scriptsize
			\REQUIRE Candidate pattern set $P_{all}$ and its frequency $freq(P_{all})$, pattern budget $b=(\eta_{min}, \eta_{max}, \gamma)$;
			\ENSURE Canned pattern set $\mathcal{P}$;
			
			\STATE $s_{best}\leftarrow 0\;$
			
			\WHILE{$\gamma>0$}\label{line:patternMiningSelectStart}
			
			\STATE $p_{best}\leftarrow\phi\;$
			
			\STATE $s_{map}\leftarrow\phi\;$
			
			\STATE $C\leftarrow\phi\;/* \textrm{list of good candidates} */$
			
			\FOR{$p\in P_{all}$}\label{line:selectGoodCandidatesStart}
			
			\STATE $f_{covub(p\bigcup\mathcal{P})}\leftarrow \textsc{GetCoverage}(p,\mathcal{P}, freq(\mathcal{P}\bigcup\{p\}))\;$\label{line:patternMiningCovCogStart}
			
			\STATE $f_{cog(p\bigcup\mathcal{P})}\leftarrow \textsc{GetCognitiveLoad}(p,\mathcal{P})\;$\label{line:patternMiningCog}
			
			\STATE $f_{sim(p\bigcup\mathcal{P})}\leftarrow \textsc{GetSimilarity}(p,\mathcal{P})\;$\label{line:patternMiningDiversity}
			
			\STATE $s\leftarrow\frac{1}{3}(f_{covub(p\bigcup\mathcal{P})}-f_{sim(p\bigcup\mathcal{P})}-f_{cog(p\bigcup\mathcal{P})}+2)$
			
			\IF{$s>s_{best}$ \textrm{ and } $|\mathcal{P}|=0$}\label{line:marginalImprovement}
			
			\STATE $s_{best}\leftarrow s\;$
			
			\STATE $p_{best}\leftarrow p\;$
			
			\ELSIF{$s>s_{best}$}\label{line:goodCandidates}
			
			\STATE $s_{map}\leftarrow\textsc{UpdateScore}(s_{map},s,p)$
			
			\STATE $C\leftarrow C\bigcup\{p\}\;$
			
			\ENDIF\label{line:marginalImprovementEnd}
			
			\ENDFOR\label{line:selectGoodCandidatesEnd}

			\IF{$|C|>0$}\label{line:randomStart}
			
			\STATE $p_{best}\leftarrow\textsc{RandomChoose}(C)$\;
			
			\STATE $s_{best}\leftarrow\textsc{GetScore}(s_{map},p_{best})$\;
			
			\ENDIF\label{line:randomEnd}
			
			\IF{$p_{best}\neq\phi$}
			
			\STATE $\mathcal{P}\leftarrow\mathcal{P}\bigcup\{p_{best}\}\;$
			
			\STATE $P_{all}\leftarrow P_{all}\setminus\{p_{best}\}\;$
			
			\STATE $\gamma\leftarrow\gamma-1\;$
			
			\ELSE
			
			\STATE \textrm{break}\;\label{line:noGoodCandidates}
			
			\ENDIF
			
			\ENDWHILE\label{line:patternMiningSelectEnd}
			
		\end{algorithmic}
		\caption{$\textsc{Greedy}$.} \label{alg:greedy}
	\end{algorithm}
}
\vspace{0ex}
\subsection{CPS-Randomized Greedy Algorithm} \label{sec:cps-algo}
The canned pattern selection algorithm is as follows. First, it retrieves the default pattern set ($1$-path, $2$-path, $3$-cycle and $4$-cycle). Next, it prunes candidate patterns whose sizes do not satisfy the plug specification or are ``nearly-unique'' (\ie $freq(p)<\delta$ where $\delta$ is a pre-defined threshold). Note that the latter patterns have very low occurrences in $G$ and are unlikely to be as useful for query construction in their entirety\footnote{\scriptsize In the case, a user is interested in patterns with low coverage, $\delta$ can be set to 0 along with the reduction in $\alpha_{f_{cov}}(P^\prime)$ in $s(P^\prime)$ (Defn.~\ref{def:patternScore}).}. Then, it selects $\mathcal{P}$ from the remaining candidates.

Recall from Section~\ref{sec:soln}, the \textsc{cps} problem can be cast as a maximization of submodular function problem subject to cardinality constraint. \eat{Recently, the algorithm community has proposed two techniques with quality guarantees in~\cite{buchbinder2014} to address it. We exploit these approaches in our \textsc{cps} problem. To the best of our knowledge, these approaches have not been utilized for graph querying.}Recently, the algorithm community has proposed a technique with quality guarantee in~\cite{buchbinder2014} to address it. We exploit this approach, referred to as \textit{CPS-Randomized Greedy} (\textsf{CPS-R-Greedy}, Algorithm~\ref{alg:greedy}), in our \textsc{cps} problem. To the best of our knowledge, this approach has not been utilized for graph querying.

In particular, \textsf{CPS-R-Greedy} extends the discrete greedy algorithm~\cite{nemhauser1978} using a randomized approach. At every step, a random candidate pattern is chosen from a set of ``reasonably good'' candidates (Lines~\ref{line:randomStart}-\ref{line:randomEnd}). Intuitively, these candidates should have very few edge crossings, good coverage and are different from patterns already in $\mathcal{P}$. These candidates are identified as follows. For every candidate pattern $p$, we compute the pattern set score (Definition~\ref{def:patternScore}) assuming $p$ is added to the canned pattern set. A ``good'' candidate $p$ improves on the score of the set when it is added (Definition~\ref{def:goodPattern}). Note that $cov_{ub}$, $cog$, and $sim$ changes as $\mathcal{P}$ changes. Hence, we recompute them at every iteration. Then, we randomly select a ``good'' candidate and assign it to $\mathcal{P}$. The algorithm terminates either when the set contains the desired number of patterns or when there exists no more good candidates. The following quality guarantee can be derived from~\cite{buchbinder2014}.

\begin{theorem}\label{thm:greedyApprox}
	\textit{\textsf{CPS-R-Greedy} achieves $\frac{1}{e}$-approximation of \textsc{cps}.}
\end{theorem}

\eat{The proof \eat{ and~\ref{thm:greedy} follow} follows from~\cite{buchbinder2014} and is summarized in~\cite{tech}.}
\eat{\begin{proof}
		Let $A_i$ be an event fixing all the random decisions of \textit{Greedy} for every iteration $i$ and $\mathcal{A}_i$ be the set of all possible $A_i$ events. We denote $s(\mathcal{P}_{i-1}\bigcup p_i)-s(\mathcal{P}_{i-1})$ as $s_{p_i}(\mathcal{P}_{i-1})$. Further, let the desired size of $\mathcal{P}$ be $\gamma$, $1\leq i\leq\gamma$ and $A_{i-1}\in\mathcal{A}_{i-1}$. Unless otherwise stated, all the probabilities, expectations and random quantities are implicitly conditioned on $A_{i-1}$. Consider a set $M^{\prime}_i$ containing the patterns of $OPT\setminus\mathcal{P}_{i-1}$ plus enough dummy patterns to make the size of $M^{\prime}_i$ exactly $\gamma$.
		
		Note that $\mathbb{E}[s_{p_i}(\mathcal{P}_{i-1})]=\gamma^{-1}\cdot\sum_{p\in M_i}
		s_{p}(\mathcal{P}_{i-1})\geq\gamma^{-1}\cdot\sum_{p\in M^{\prime}_i}s_{p}
		(\mathcal{P}_{i-1})=\gamma^{-1}\cdot\sum_{p\in OPT\setminus\mathcal{P}_{i-1}}s_{p}(\mathcal{P}_{i-1})\geq\linebreak\frac{s(OPT\bigcup\mathcal{P}_{i-1})-s(\mathcal{P}_{i-1})}{\gamma}$ \cite{buchbinder2014}, where the first inequality follows from the definition of $M_i$ (\ie set of ``good'' candidate patterns) and the second from the submodularity of $s(.)$ Unfixing the event $A_{i-1}$ and taking an expectation over all possible such events, $\mathbb{E}[s_{p_i}(\mathcal{P}_{i-1})]\geq\frac{\mathbb{E}[s(OPT\bigcup\mathcal{P}_{i-1})]-\mathbb{E}[s(\mathcal{P}_{i-1})]}{\gamma}
		\geq\frac{(1-\frac{1}{\gamma})^{i-1}\cdot s(OPT)-\mathbb{E}[s(\mathcal{P}_{i-1})]}{\gamma}$, where the second inequality is due to observation that for every $0\geq i\geq\gamma$, $\mathbb{E}[s(OPT\bigcup\mathcal{P}_i)]\geq(1-\frac{1}{\gamma})^i\cdot s(OPT)$\cite{buchbinder2014}.
		
		We now prove by induction that $\mathbb{E}[s(\mathcal{P}_i)]\geq\frac{i}{\gamma}\cdot(1-\frac{1}{\gamma})^{i-1}\cdot s(OPT)$. Note that this is true for $i=0$ since $s(\mathcal{P}_0)\geq 0=\frac{0}{\gamma}\cdot(1-\frac{1}{\gamma})^{-1}\cdot s(OPT)$. Further, we assume that the claim holds for every $i^{\prime}<i$. Now, we prove it for $i>0$. $\mathbb{E}[s(\mathcal{P}_i)]=\mathbb{E}[s(\mathcal{P}_{i-1})]+\mathbb{E}[s_{p_i}(\mathcal{P}_{i-1})]
		\geq\mathbb{E}[s(\mathcal{P}_{i-1})]+\frac{(1-\frac{1}{\gamma})^{i-1}\cdot s(OPT)-\mathbb{E}[s(\mathcal{P}_{i-1})]}{\gamma}\linebreak=(1-\frac{1}{\gamma})\cdot\mathbb{E}[s(\mathcal{P}_{i-1})]
		+\gamma^{-1}(1-\frac{1}{\gamma})^{i-1}\cdot s(OPT)\geq(1-\frac{1}{\gamma})\cdot[\frac{i-1}{\gamma}\cdot(1-\frac{1}{\gamma})^{i-2}\cdot s(OPT)]+\gamma^{-1}(1-\frac{1}{\gamma})^{i-1}\cdot s(OPT)=[\frac{i}{\gamma}]\cdot(1-\frac{1}{\gamma})^{i-1}\cdot s(OPT)$. Hence, $\mathbb{E}[s(\mathcal{P}_k)]\geq\frac{\gamma}{\gamma}\cdot(1-\frac{1}{\gamma})^{\gamma-1}\cdot s(OPT)\geq e^{-1}\cdot s(OPT)$. That is, Alg.~\ref{alg:greedy} achieves $\frac{1}{e}$-approximation of \textsc{cps}.
\end{proof}}

\begin{theorem}\label{thm:greedyComplexity}
	\textit{\textsf{CPS-R-Greedy} has worst-case time and space complexity of $O(|P_{cand}|\gamma|V_{max}||V_{max}|!)$ and $O(|P_{cand}|(|V_{max}|+|E_{max}|))$, respectively, where $|V_{max}|$ and $|E_{max}|$ are the number of vertices and edges in the largest candidate pattern.}
	\vspace{-1ex}\end{theorem}

\begin{figure}[!t]
	\centering
	\includegraphics[width=\linewidth]{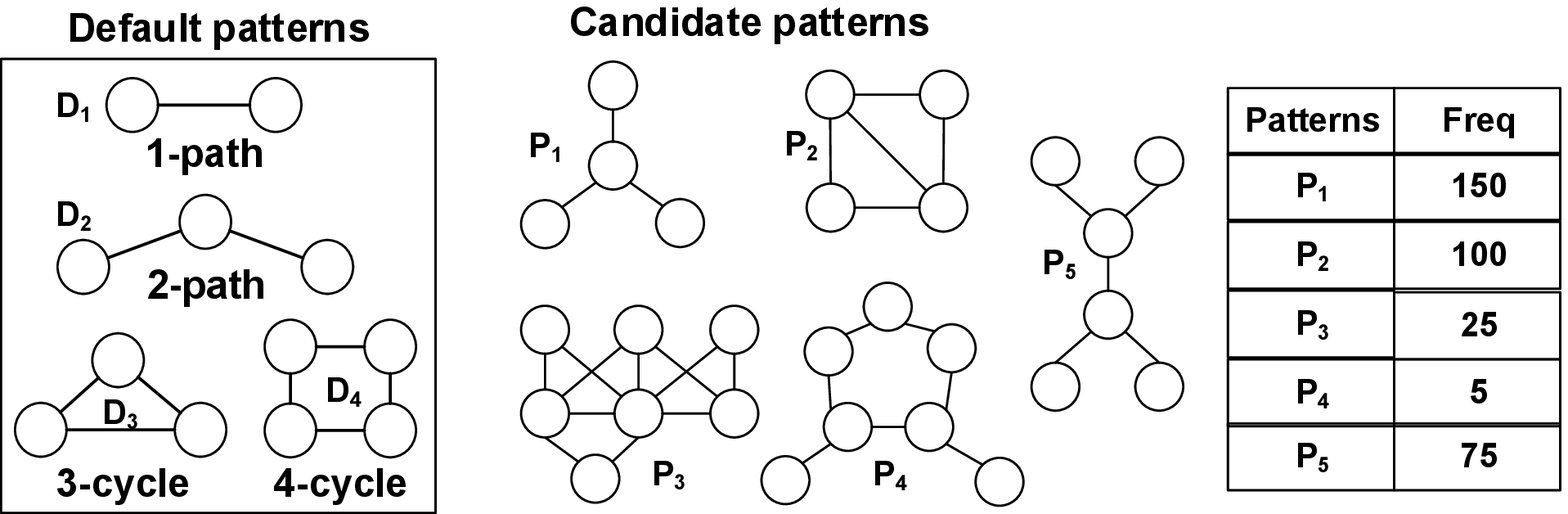}
	\vspace{-3ex}\caption{Default patterns and candidate patterns.}\label{fig:patternMining}
	\vspace{-2ex}
\end{figure}

\eat{\begin{proof} (Sketch).
		Let $G_{max}=(V_{max},E_{max})$ be the largest candidate pattern in $P_{all}$. In the worst-case, time complexity of Algorithm~\ref{alg:greedy} is $O(|P_{all}|\gamma|V_{max}|!|V_{max}|)$ since there are $|P_{all}|$ candidate patterns and the while-loop in Algorithm~\ref{alg:greedy} iterates at most $\gamma$ times. For each iteration, the score function requires computation of coverage, cognitive load and redundancy which requires $O(|V_{max}|!|V_{max}|)$, $O(|V_{max}|+|E_{max}|)$ and $O(|V_{max}|+|V_{max}|log(|V_{max}|)$~\cite{berlingerio2013}, respectively. Note that $|V_{max}|log(|V_{max}|)\approx|E_{max}|$ in real-world graphs~\cite{berlingerio2013}. The space complexity is due to storage of all candidate patterns. Hence, Algorithm~\ref{alg:greedy} has space complexity of $O(|P_{all}|(|V_{max}|+|E_{max}|))$.
\end{proof}}

\eat{\vspace{-0.5ex}\textbf{Remark.} Note that in Theorem~\ref{thm:greedyComplexity} the exponential time complexity is due to the subgraph isomorphism check during diversity computation.\eat{ In practice, it can be performed fast as candidate patterns are small in size.}}

\vspace{-0.5ex}\begin{example} \label{eg:greedyExample}
	Consider a \textsc{gui} $\mathbb{I}$ and a plug $b=(3,11,6\eat{, 1})$. Suppose there are four default patterns and five candidate patterns (\ie $P_{cand}$) as depicted in Figure~\ref{fig:patternMining}.\eat{ Observe that $p_1$ and $p_5$ belong to star-like patterns whereas $p_2$ and $p_3$ are $k$-\textsc{cp} and \textsc{ccp}.} Let $\delta=10$. The algorithm first removes $p_4$ since $freq(p_4)<\delta$. Then, for the remaining patterns in $P_{cand}$, each is considered in turn to be added to $\mathcal{P}$ by exploiting \textsf{CPS-R-Greedy} technique. It first considers adding $p_1$ to $\mathcal{P}$ and computes the resulting coverage ($f_{covub(\mathcal{P}\bigcup p_1)}$), cognitive load ($f_{cog(\mathcal{P}\bigcup p_1)}$) and similarity ($f_{sim(\mathcal{P}\bigcup p_1)}$). The pattern set score of $\mathcal{P}\bigcup p_1$ is then computed using Defn.~\ref{def:patternScore}. The scores of the other candidate patterns are computed similarly. Suppose the scores are 0.72, 0.63, 0.54, 0.68 for $p_1$, $p_2$, $p_3$, $p_5$, respectively. Then, in the first iteration, $p_1$ is selected (and removed from subsequent iterations) and the current best score $s_{best}$ is updated to 0.72. In the next (\ie final) iteration, the candidates are again considered in turn to be added to $\mathcal{P}$ and corresponding pattern set scores are computed. However, unlike the first iteration, only those candidates whose scores are greater than $s_{best}$ are considered. Let the scores of $p_2$, $p_3$ and $p_5$ be 0.81, 0.7 and 0.77, respectively. Then, a candidate will be randomly selected from $p_2$ or $p_5$. Suppose $p_2$ is chosen, then the final pattern set is $\{d_1, d_2, d_3, d_4, p_1, p_2\}$.
	\EndOfProof\end{example}

\begin{algorithm}[!t]
	\algsetup{
		linenosize=\scriptsize
	}
	\begin{algorithmic}[1]
		\scriptsize
		\REQUIRE Candidate pattern set $P_{all}$ and its frequency $freq(P_{all})$, pattern budget $b=(\eta_{min}, \eta_{max}, \gamma)$;
		\ENSURE Canned pattern set $\mathcal{P}$;
		
		\STATE $s_{best}\leftarrow 0\;$
		
		\WHILE{$\gamma>0$}\label{line:patternMiningSelectStart}
		
		\STATE $p_{best}\leftarrow\phi\;$
		
		\STATE $s_{map}\leftarrow\phi\;$
		
		\STATE $C\leftarrow\phi\;/* \textrm{list of good candidates} */$
		
		\FOR{$p\in P_{all}$}\label{line:selectGoodCandidatesStart}
		
		\STATE $f_{covub(p\bigcup\mathcal{P})}\leftarrow \textsc{GetCoverage}(p,\mathcal{P}, freq(\mathcal{P}\bigcup\{p\}))\;$\label{line:patternMiningCovCogStart}
		
		\STATE $f_{cog(p\bigcup\mathcal{P})}\leftarrow \textsc{GetCognitiveLoad}(p,\mathcal{P})\;$\label{line:patternMiningCog}
		
		\STATE $f_{sim(p\bigcup\mathcal{P})}\leftarrow \textsc{GetSimilarity}(p,\mathcal{P})\;$\label{line:patternMiningDiversity}
		
		\STATE $s\leftarrow\frac{1}{3}(f_{covub(p\bigcup\mathcal{P})}-f_{sim(p\bigcup\mathcal{P})}-f_{cog(p\bigcup\mathcal{P})}+2)$
		
		\IF{$s>s_{best}$ \textrm{ and } $|\mathcal{P}|=0$}\label{line:marginalImprovement}
		
		\STATE $s_{best}\leftarrow s\;$
		
		\STATE $p_{best}\leftarrow p\;$
		
		\ELSIF{$s>s_{best}$}\label{line:goodCandidates}
		
		\STATE $s_{map}\leftarrow\textsc{UpdateScore}(s_{map},s,p)$
		
		\STATE $C\leftarrow C\bigcup\{p\}\;$
		
		\ENDIF\label{line:marginalImprovementEnd}
		
		\ENDFOR\label{line:selectGoodCandidatesEnd}

		\IF{$|C|>0$}\label{line:randomStart}
		
		\STATE $p_{best}\leftarrow\textsc{RandomChoose}(C)$\;
		
		\STATE $s_{best}\leftarrow\textsc{GetScore}(s_{map},p_{best})$\;
		
		\ENDIF\label{line:randomEnd}
		
		\IF{$p_{best}\neq\phi$}
		
		\STATE $\mathcal{P}\leftarrow\mathcal{P}\bigcup\{p_{best}\}\;$
		
		\STATE $P_{all}\leftarrow P_{all}\setminus\{p_{best}\}\;$
		
		\STATE $\gamma\leftarrow\gamma-1\;$
		
		\ELSE
		
		\STATE \textrm{break}\;\label{line:noGoodCandidates}
		
		\ENDIF
		
		\ENDWHILE\label{line:patternMiningSelectEnd}
		
	\end{algorithmic}
	\caption{$\textsc{CPS-R-Greedy}$.} \label{alg:greedy}
\end{algorithm}

\eat{\begin{algorithm}[!t]
		\algsetup{
			linenosize=\scriptsize
		}
		\begin{algorithmic}[1]
			\scriptsize
			\REQUIRE Candidate pattern set $P_{all}$ and its frequency $freq(P_{all})$, pattern budget $b=(\eta_{min}, \eta_{max}, \gamma)$;
			\ENSURE Canned pattern set $\mathcal{P}$;
			
			\STATE $\mathcal{P}_A\leftarrow\phi\;$
			
			\STATE $\mathcal{P}_B\leftarrow P_{all}\;$
			
			\WHILE{$|\mathcal{P}_A|\neq|\mathcal{P}_B|$}\label{line:dGreedyStart}
			
			\FOR{$p\in\mathcal{P}_B\setminus\mathcal{P}_A$}
			
			\STATE $\alpha_p\leftarrow\max\{s(p\bigcup\mathcal{P}_A)-s(\mathcal{P}_A),0\}$\label{line:dGreedyAlpha}
			
			\STATE $\beta_p\leftarrow\max\{s(\mathcal{P}_B\setminus p)-s(\mathcal{P}_B),0\}$\label{line:dGreedyBeta}
			
			\ENDFOR
			
			\STATE $\alpha\leftarrow\textsc{Bound}(\alpha,\gamma,|\mathcal{P}_A|)$\label{line:dGreedyBound}
			
			\FOR{$p\in\mathcal{P}_B\setminus\mathcal{P}_A$}
			
			\STATE $\mathcal{P}_A\leftarrow\textsc{updateSetA}(\mathcal{P}_A,\frac{\alpha_p}{\alpha_p+\beta_p},p)$\label{line:dGreedyUpdateSetA}
			
			\STATE $\mathcal{P}_B\leftarrow\textsc{updateSetB}(\mathcal{P}_B,\frac{\beta_p}{\alpha_p+\beta_p},p)$\label{line:dGreedyUpdateSetB}
			
			\ENDFOR
			
			\ENDWHILE\label{line:dGreedyEnd}
			
			\STATE $\mathcal{P}\leftarrow\mathcal{P}\bigcup\mathcal{P}_A$
			
		\end{algorithmic}
		\caption{$\textsc{D-Greedy}$.} \label{alg:dGreedy}
	\end{algorithm}
}

\eat{\begin{algorithm}[!t]
		\algsetup{
			linenosize=\scriptsize
		}
		\begin{algorithmic}[1]
			\small
			\REQUIRE Candidate pattern set $P_{all}$, plug $b=(\eta_{min}, \eta_{max}, \gamma, \mathbb{N})$;
			\ENSURE Canned pattern set $\mathcal{P}$;
			
			\STATE $x_0\leftarrow\phi\;$
			
			\STATE $y_0\leftarrow P_{cand}\;$
			
			\FOR{$i\in[1,|P_{cand}|]$}
			
			\STATE $\alpha_i^{\prime}\leftarrow F(x_{i-1}+ p_i)-F(x_{i-1})$\label{line:dGreedyMultilinear1}
			
			\STATE $\beta_i^{\prime}\leftarrow F(y_{i-1}- p_i)-F(y_{i-1})$\label{line:dGreedyMultilinear2}
			
			\STATE $\alpha_i^{\prime}\leftarrow\max\{\alpha_i,0\}$
			
			\STATE $\beta_i^{\prime}\leftarrow\max\{\beta_i,0\}$
			
			\STATE $x_i\leftarrow x_{i-1}+\frac{\alpha_i^{\prime}}{\alpha_i^{\prime}+\beta_i^{\prime}}\cdot\{p_i\}$\label{line:dGreedyUpdateX}
			
			\STATE $y_i\leftarrow y_{i-1}-\frac{\beta_i^{\prime}}{\alpha_i^{\prime}+\beta_i^{\prime}}\cdot\{p_i\}$\label{line:dGreedyUpdateY}
			
			\ENDFOR
			
			\STATE $\mathcal{P}\leftarrow GetIntSol(x_{|P_{cand}|})$/* Equiv. to $GetIntSol(y_{|P_{cand}|})$ */\label{line:dGreedyRounding}
			
		\end{algorithmic}
		\caption{\textsf{CPS-CD}-Greedy Algorithm.} \label{alg:dGreedy}
\end{algorithm}}

\eat{\textbf{CPS-Continuous Double Greedy Algorithm (CPS-CD-Greedy).} The \textit{Continuous Double Greedy} algorithm in ~\cite{buchbinder2014} is an extension of the \textit{Discrete Double Greedy} (\textsf{DD-Greedy}) approach~\cite{buchbinder2015}. We briefly describe \textsf{DD-Greedy} to ease exposition of \textsf{CPS-CD-Greedy}.\eat{ The \textsf{DD-Greedy} was originally developed to address \textit{unconstrained submodular maximization} (\textsc{usm}) problem.} Given a non-negative submodular function $f$, let the complement of $f$ be $\overline{f}$ defined as $\overline{f}(X)\triangleq f(\mathcal{N}\setminus X)$ for any $X\subseteq\mathcal{N}$. \eat{Consider the greedy algorithm. For $f$, the algorithm starts with an empty solution and in each iteration, greedily adds an element to it. Conversely, in the case of $\overline{f}$, the algorithm starts with solution $\mathcal{N}$ and iteratively removes elements from it. It is known that both approaches fail when applied separately for \textsc{usm}~\cite{buchbinder2015}.} The \textsf{DD-Greedy} algorithm is as follows: it starts with two solution sets $X$ and $Y$ which are initialized to an empty set and $\mathcal{N}$, respectively. Then, it considers one random element (\eg $p$) at a time. Let $\alpha_p$ and $\beta_p$ be the marginal gains of adding $p$ to $X$ and removing $p$ from $Y$, respectively. Note that $p$ will be added to $X$ (resp. removed from $Y$) if $\alpha_p\geq\beta_p$ (resp. $\alpha_p<\beta_p$). After a single pass of the ground set $\mathcal{N}$, the two solutions coincide and is returned as the output of the algorithm. Such correlated execution on both $f$ and $\overline{f}$ yields an approximation with guarantee of $\frac{1}{2}$~\cite{buchbinder2015}.
	
	The \textit{Continuous Double Greedy} algorithm extends \textsf{DD-Greedy} to a continuous setting to allow greater control over the sizes of $X$ and $Y$. Algorithm~\ref{alg:dGreedy} outlines it in the context of \textsc{cps} problem. In contrast to \textsf{DD-Greedy}, it assigns a fractional value to each $p$ and then the resulting fractional solution is transformed into an integral solution with the same expected cost (Line~\ref{line:dGreedyRounding}) using known rounding techniques~\cite{buchbinder2014} (we use \textit{pipage rounding}~\cite{ageev2004}, which iteratively transforms a given fractional solution $x$ into a new solution $x^{\prime}$ with fewer number of non-integral components). This is achieved by (1) replacing the original submodular pattern set score function $s(.)$ (Definition~\ref{def:patternScore}) with its multilinear extension $F$, and (2) replacing the solution sets $X$ and $Y$ with the vectors $x_i,y_i\in[0,1]^{|P_{cand}|}$. In particular, the value of $F$ at point $x\in[0,1]^{|P_{cand}|}$ is the expected value of $s(.)$ over a random subset $\mathcal{P}_x\subseteq P_{cand}$. The subset $\mathcal{P}_x$ contains each pattern $p\in P_{cand}$ independently with probability $x_p$. For every $x\in[0,1]^{|P_{cand}|}$, $F(x)\triangleq \mathbb{E}[\mathcal{P}_x]=\sum_{M\subseteq P_{cand}}s(M)\prod_{p\in M}x_p\prod_{p\notin M}(1-x_p)$. Similar to~\cite{buchbinder2015}, we abuse notations and unify a set with its characteristic vector. In particular, for $x,y\in[0,1]^{|P_{cand}|}$, we use $x\vee y$ and $x\wedge y$ to denote the coordinate-wise maximum and minimum, respectively, of $x$ and $y$. That is, $(x\vee y)_p=\max\{x_p,y_p\}$ and $(x\wedge y)_p=\min\{x_p,y_p\}$. Further, for Lines~\ref{line:dGreedyUpdateX} and~\ref{line:dGreedyUpdateY}, if $\alpha_i^{\prime}$ or $\beta_i^{\prime}$ is 0, then $\frac{\alpha_i^{\prime}}{\alpha_i^{\prime}+\beta_i^{\prime}}$ and $\frac{\beta_i^{\prime}}{\alpha_i^{\prime}+\beta_i^{\prime}}$ are set to 1 and 0, respectively.
	
	\eat{In \textsc{Tattoo}, rounding (Line~\ref{line:dGreedyRounding}) is achieved using pipage rounding~\cite{ageev2004}. Briefly, pipage rounding iteratively transform a given fractional solution $x$ into a new solution $x^{\prime}$ with fewer number of non-integral components. We begin by formulating a general rounding problem. Given a bipartite graph $G=(U,W;E)$, a function $F(x)$ defined on rational points $x=(x_e:e\in E)$ of the $|E|$-dimensional cube $[0,1]^{|E|}$ computable in polynomial time, and $p:U\bigcup W\rightarrow \mathbb{Z}_+$. Consider the binary program $\max F(x)$ s.t. three constraints are satisfied, namely, (1) $\sum_{e\in\delta(v)}x_e\leq p_v$, $v\in U\bigcup W$, (2) $0\leq x_e\leq 1$, $e\in E$, and (3) $x_e\in\{0,1\}$, $e\in E$. The general step of pipage rounding is as follows: If $x$ is integral, rounding terminates and $x$ is returned. Otherwise, a subgraph $H_x$ of $G$ with the same vertex set and edge set $E_x$ is constructed. Note that $E_x$ is defined by the condition that $e\in E_x$ iff $x_e$ is non-integral. If $H_x$ contains cycles, then set $R$ to be such a cycle. If $H_x$ is a forest, $R$ is set as a path of $H_x$ whose endpoints have degree 1. Observe that in both cases, $R$ can be uniquely represented as the union of two matchings (\ie $M_1$ and $M_2$) since $H_x$ is bipartite. Then, a new solution $x(\epsilon,R)$ is defined with the following rule: if $e\in E\setminus R$, then $x_e(\epsilon,R)$ coincides with $x_e$, otherwise, $x_e(\epsilon,R)=x_e+\epsilon$ when $e\in M_1$, and $x_e(\epsilon,R)=x_e-\epsilon$ when $e\in M_2$. Further, $\epsilon_1=\min\{\min_{e\in M_1}x_e,\min_{e\in M_2}(1-x_e)\}$, $\epsilon_2=\min\{\min_{e\in M_1}(1-x_e),\min_{e\in M_2}x_e\}$, $x_1=x(-\epsilon_1,R)$ and $x_2=x(\epsilon_2,R)$. \textit{Rounding of $x$ over $R$} is done by setting $x^{\prime}=x_1$ if $F(x_2)<F(x_1)$ and $x^{\prime}=x_2$ otherwise. Note that $F(\overline{x})\geq F(x)$ where $\overline{x}$ is the integral vector converted from the fractional solution $x$~\cite{ageev2004}.}
	
	\eat{\begin{lemma}\label{lem:dGreedyMarginalGain}
			In \textit{D-Greedy}, for every $1\geq i\geq |P_{all}|, F(OPT_{i-1})-F(OPT_i)\leq\frac{1}{2}\cdot[F(x_i)-F(x_{i-1})+F(y_i)-F(y_{i-1})].$
	\end{lemma}}
	
	\eat{\begin{proof}
			Similar to \cite{buchbinder2015}, we begin by proving that for every $1\geq i\geq |P_{all}|, s(OPT_{i-1})-s(OPT_i)\leq\frac{1}{2}\cdot[s(X_i)-s(X_{i-1})+F(Y_i)-F(Y_{i-1})]$. Assume without loss of generality that $\alpha_i\geq\beta_i$. That it, $X_i\leftarrow X_{i-1}\bigcup\{p_i\}$ and $Y_i\leftarrow Y_{i-1}$ (the other case is analogous). Note that in this case, $OPT_i=(OPT\bigcup X_i)\bigcap Y_i=OPT_{i-1}\bigcup\{p_i\}$ and $Y_i=Y_{i-1}$. Thus, the inequality can be written as $s(OPT_{i-1})-s(OPT_{i-1}\bigcup\{p_i\})\leq s(X_i)-s(X_{i-1})=\alpha_i$. Let's consider the case where $p_i\in OPT$. Note that LHS of the inequality is 0. Further, since $(X_{i_1}\bigcup\{p_i\})\bigcup(Y_i\setminus\{p_i\})=Y_{i-1}$ and $(X_{i-1}\bigcup\{p_i\})\bigcap(Y_i\setminus\{p_i\})=X_{i-1}$ and by combining these observations with submodularity, $\alpha_i+\beta_i=[s(X_{i-1}\bigcup\{p_i\})-s(X_{i-1})]+[s(Y_{i-1}\setminus\{p_i\})-s(Y_{i_1})]=[s(X_{i-1}\bigcup\{p_i\})+s(Y_{i-1}\setminus\{p_i\})]-[s(X_{i-1})+s(Y_{i-1})]\geq 0$ (\ie for every $1\leq i\leq |P_{all}|$, $\alpha_i+\beta_i\geq 0$ (Lemma A)). Coupled with our assumption that $\alpha_i\geq\beta_i$ (by our assumption), then $\alpha_i$ is non-negative. Now, considering the case where $p_i\notin OPT$, then, $p_i\notin OPT_{i-1}$. Thus, $s(OPT_{i-1})-s(OPT_{i-1}\bigcup\{p_i\})\leq s(Y_{i-1}\setminus\{p_i\})-s(Y_{i-1})=\beta_i\leq\alpha_i$. Note that the first inequality follows by submodularity since $OPT_{i-1}=((OPT\bigcup X_{i-1})\bigcap Y_{i-1})\subseteq Y_{i-1}\setminus\{p_i\}$, $p_i\in Y_{i-1}$ and $p_i\notin OPT_{i_1}$. The second inequality follows from the assumption that $\alpha_i\geq\beta_i$.
			
			By Lemma A and $\alpha_i+\beta_i\geq 0$, $\alpha_i$ and $\beta_i$ cannot be both strictly less than 0. Thus, there are 3 cases left to consider, namely, $\alpha_i\geq 0$ and $\beta_i\leq 0$ (case 1); $\alpha_i<0$ and $\beta_i\geq 0$ (case 2) and $\alpha_i\geq 0$ and $\beta_i>0$ (case 3). For case 1: $\frac{\alpha_i^{\prime}}{\alpha_i^{\prime}+\beta_i^{\prime}}=1$ and so $y_i=y_{i-1}$, $x_i\leftarrow x_{i-1}\vee\{p_i\}$. Hence, $F(y_i)-F(y_{i-1})=0$. By our definition, $OPT_i=(OPT\vee x_i)\wedge y_i=OPT_{i-1}\vee\{p_i\}$ and it is left to prove that $F(OPT_{i-1})-F(OPT_{i-1}\vee\{p_i\})\leq\frac{1}{2}\cdot[F(x_i)-F(x_{i-1})]=\frac{\alpha_i}{2}$. Note that if $p_i\in OPT$, then LHS of inequality is 0 (\ie $\leq\frac{\alpha_i}{2}$). When $p_i\notin OPT$, then $F(OPT_{i-1})-F(OPT_{i-1}\vee\{p_i\})\leq F(y_{i-1}-\{p_i\})-F(y_{i-1})=\beta_i\leq 0\leq\frac{\alpha_i}{2}$. Observe that the first inequality follows from submodularity since $OPT_{i-1}=((OPT\wedge x_{i-1})\vee y_{i-1})\leq y_{i-1}-\{p_i\}$. Note that $(y_{i-1})_{p_i}=1$ and $(OPT_{i-1})_{p_i}=0$. Case 2 is analogous to case 1 and we omit its proof. Case 3: $\alpha_i^{\prime}=\alpha_i$, $\beta_i^{\prime}=\beta_i$, and so, $x_i\leftarrow x_{i+1}+\frac{\alpha_i}{\alpha_i+\beta_i}\cdot\{p_i\}$ and $y_i\leftarrow y_{i-1}-\frac{\beta_i}{\alpha_i+\beta_i}\cdot\{p_i\}$. Therefore, $F(x_i)-F(x_{i-1})=\frac{\alpha_i}{\alpha_i+\beta_i}\cdot[F(x_{i-1}\vee\{p_i\})-F(x_{i-1})]=\frac{\alpha_i^2}{\alpha_i+\beta_i}$ (Equation A). Similarly, $F(y_i)-F(y_{i-1})=\frac{\beta_i^2}{\alpha_i+\beta_i}$ (Equation B). Now, we upper bound $F(OPT_{i-1})-F(OPT_i)$. Let's assume that $p_i\notin OPT$ (proof for $p_i\in OPT$ is similar). Note that $OPT_i=(OPT\vee x_i)\wedge y_i$. $F(OPT_{i-1})-F(OPT_i)=\frac{\alpha_i}{\alpha_i+\beta_i}\cdot[F(OPT_{i-1})-F(OPT_{i-1}\vee\{p_i\})]\leq
			\frac{\alpha_i}{\alpha_i+\beta_i}\cdot[F(y_{i-1}-\{p_i\})-F(y_{i-1})]=\frac{\alpha_i\beta_i}{\alpha_i+\beta_i}$ (Equation C). The above inequality follows from the submodularity of $s$ since $OPT_{i-1}=((OPT\vee x_{i-1})\wedge y_{i-1})\leq y_{i-1}-\{p_i\}$. Note that $(y_{i-1})_{p_i}=1$ and $(OPT_{i-1})_{p_i}=0$. We now plug Equations A, B and C into the inequality that we need to prove and obtain $\frac{\alpha_i\beta_i}{\alpha_i+\beta_i}\leq\frac{1}{2}\cdot\frac{\alpha_i^2+\beta_i^2}{\alpha_i+\beta_i}$.
	\end{proof}}
	
	\begin{theorem}\label{thm:greedy}
		Algorithm~\ref{alg:dGreedy} achieves $\frac{1}{2}$-approximation of \textsc{cps} if there is oracle access to $F$. Otherwise, it achieves an approximation ratio of $\frac{1}{2}-o(1)$ with sampling.
	\end{theorem}

	\eat{\begin{proof}
			Let $OPT_i\triangleq(OPT\vee x_i)\wedge y_i$. For the sequence $F(OPT_0)$,$\cdots$,$F(OPT_{|P_{all}|})$, $OPT_0=OPT$ and $OPT_{|P_{all}|}=x_{|P_{all}|}=y_{|P_{all}|}$ since the sequence starts with the value of an optimal solution and the sequence ends at a fractional point whose value is the expected value of the algorithm's output. Further, summing up Lemma~\ref{lem:dGreedyMarginalGain} for every $1\leq i\leq |P_{all}|$ gives $\sum^{|P_{all}|}_{i=1}[F(OPT_{i-1})-F(OPT_i)]\leq\frac{1}{2}\cdot\sum^{|P_{all}|}_{i=1}[F(x_i)-F(x_{i-1})]+
			\frac{1}{2}\cdot\sum^{|P_{all}|}_{i=1}[F(y_i)-F(y_{i-1})]$. Collapsing the above telescopic sum, $F(OPT_0)-F(OPT_{|P_{all}|})\leq\frac{1}{2}\cdot[F(x_{|P_{all}|})-F(x_0)]+\frac{1}{2}\cdot[F(y_{|P_{all}|})-F(y_0)]\leq\frac{F(x_{|P_{all}|})+F(y_{|P_{all}|})}{2}$. Given the definitions of $OPT_0$ and $OPT_{|P_{all}|}$, then $F(x_{|P_{all}|})=F(y_{|P_{all}|})\geq\frac{s(OPT)}{2}$.
	\end{proof}}

	\begin{theorem}\label{thm:dGreedyComplexity}
		Algorithm~\ref{alg:dGreedy} has worst-case time and space complexity of $O(|P_{cand}||V_{max}|!|V_{max}|)$ and $O(|P_{cand}|(|V_{max}|+|E_{max}|))$, respectively.
	\end{theorem}
	
	\eat{\begin{proof}
			The proof is similar to that Theorem~\ref{thm:greedyComplexity}. In the case of Algorithm~\ref{alg:dGreedy}, the for-loop iterates for $|P_{all}|$ times and for each iteration, the score function requires computation of coverage, cognitive load and redundancy which requires $O(|V_{max}|!|V_{max}|)$, $O(|V_{max}|+|E_{max}|)$ and $O(|V_{max}|+|V_{max}|log(|V_{max}|)$~\cite{berlingerio2013}, respectively. Hence, worst-case time complexity is $O(|P_{all}||V_{max}|!|V_{max}|)$. The space complexity is due to storage of all candidate patterns. Hence, the space complexity is $O(|P_{all}|(|V_{max}|+|E_{max}|))$.
	\end{proof}}

	\vspace{-0.5ex}\begin{example} \label{eg:dGreedyExample}
		Consider a \textsc{gui} $\mathbb{I}$ and a plug $b=(3,11,6\eat{, 1})$. Suppose there are four default patterns and five candidate patterns ($P_{cand}$) as depicted in Figure~\ref{fig:patternMining}. The \textsf{CPS-CD-Greedy} algorithm starts with $x_0$ and $y_0$ being set to $\phi$ and $P_{cand}$, respectively. In the first iteration (\ie $i=1$), it computes the marginal gains of adding candidate $p_1$ to $x$ and removing $p_1$ from $y$, denoted as $\alpha_1^{\prime}$ and $\beta_1^{\prime}$, respectively. In particular, $\alpha_1^{\prime}=s(p_1)-s(\phi)$. Similarly, $\beta_1^{\prime}=s(P_{cand}\setminus p_1)-s(P_{cand})$.
		The partial solutions $x_1$ and $y_1$ are updated as $x_0+\frac{\alpha_1^{\prime}}{\alpha_1^{\prime}+\beta_1^{\prime}}\cdot\{p_1\}$ and $y_0+\frac{\beta_1^{\prime}}{\alpha_1^{\prime}+\beta_1^{\prime}}\cdot\{p_1\}$, respectively. As \textsf{CPS-CD-Greedy} iterates through $P_{cand}$, it assigns a fractional value to each candidate $p_1$, $p_2$, $p_3$, $p_5$. In the $|P_{cand}|^{th}$ iteration, the fractional value of the candidates in $x_{|P_{cand}|}$ converges with that of $y_{|P_{cand}|}$. Let $x_{|P_{cand}|}=(x_{p_1}, x_{p_2}, x_{p_3}, x_{p_5})=(0.8, 0.7, 0.3, 0.2)$. The fractional solution is then transformed into an integral solution using pipage rounding.
		\eat{ First, candidates whose $x$ coordinate values are zero are pruned (\ie no pruning is required in this case). Then, a pair of $x$ coordinates is selected (\ie $(x_{P_1},x_{P_5})=(0.8,0.2)$). There are two possible transformations, namely, (1) (0.8+0.2,0.2-0.2) where $\epsilon_1=0.2$ and (2) (0.8-0.8,0.2+0.8) where $\epsilon_2=0.8$. The probabilities of transformations 1 and 2 are $\frac{\epsilon_2}{\epsilon_1+\epsilon_2}=\frac{8}{10}$ and $\frac{\epsilon_1}{\epsilon_1+\epsilon_2}=\frac{2}{10}$, respectively. Suppose transformation 1 is selected, then $x$ is updated to $(1, 0.7, 0.3, 0)$. In the next iteration, the non-integral coordinates remaining are those corresponding to $P_2$ and $P_3$ (\ie $(0.7,0.3)$). Similarly, the possible transformations are (0.7+0.3,0.3-0.3) and (0.7-0.7,0.3+0.7) with probabilities $\frac{7}{10}$ and $\frac{3}{10}$, respectively.} Suppose $x$ is updated to $(1, 1, 0, 0)$. Then, the final canned pattern set is $\{d_1, d_2, d_3, d_4, p_1, p_2\}$.
		\EndOfProof\end{example}}

\vspace{-2ex}
\section{Performance Study}\label{sec:expt}
\textsc{Tattoo} is implemented in C++ with GCC 4.2.1 compiler. \eat{ In particular, the \textit{Boost} library is used to perform subgraph checks, isomorphism checks and for multithreading.} We now report the key performance results of \textsc{Tattoo}.\eat{ Additional results and a case study are discussed in~\cite{tech}.} All experiments are performed on a 64-bit Windows 10 desktop with Intel(R) Core(TM) i7-4770K CPU (3.50GHz) and 16GB RAM.\eat{ We tested \textsc{=} on both Windows 10 and Ubuntu 16.04 platforms. Note that we report only results on the Windows platform as the time taken for all experiments are similar on both platforms and the remaining performance measures recorded are identical.}
\vspace{0ex}
\subsection{Experimental Setup}\label{sec:exptSetup}
\textbf{Datasets.} We evaluate \textsc{Tattoo}'s performance using 10 large networks (Table~\ref{tab:graphDecomposition}) from \textsc{snap} (\url{http://snap.stanford.edu/data/index.html}) containing up to 34.7 million edges.

\textbf{Algorithms.}\eat{ To the best of our knowledge, there is no existing data-driven technique that selects canned patterns automatically from \emph{large} networks.}\eat{ Note that the recent work in~\cite{catapult} focuses on a large collection of small- or medium-sized graphs and not large networks.} State-of-the-art \textsc{gui}s for large networks~\cite{PH+17,PH+18} do not support canned patterns. Hence, we compare \textsc{Tattoo} with the following baselines: (a) \textit{\textsc{Catapult}}~\cite{catapult}: We assign same labels to all nodes of a network and partition it into a collection of small- or medium-sized data graphs using \textsc{Metis}~\cite{karypis1997}. Then the algorithm in~\cite{catapult} is used to select canned patterns.\eat{ Note that except \textit{Amazon}, other datasets either cannot be segmented by \textsc{Metis} or fail to generate patterns\footnote{\scriptsize Pattern generation using \textsc{Catapult} did not complete within 12 hrs likely due to too many possible matches due to unlabelled graphs and consequently exponential time required for processing \textsc{ged} of every possible match.}.} (b) \textit{Use \textit{graphlets}, \textit{frequent subgraphs}, \textit{random patterns}, default patterns, and edge-at-a-time (\ie pattern oblivious):} \textit{$x$-node graphlets} where $x\in[2-5]$ are generated using the approach in \cite{chen2016}. \textit{Random patterns} are generated by randomly selecting subgraphs of specific sizes from a network. The number of candidates per size follows a uniform distribution. \textit{Frequent subgraphs} are generated using \textit{Peregrine}\cite{jamshidi2020} (downloaded from~\cite{peregrine-c}). These subgraphs are considered as candidates from which the canned patterns are selected using our algorithm in Section~\ref{sec:cps-algo}.

\textbf{Query sets and GUI.} We use different query sets for the user study and automated performance study. We shall elaborate on them in respective sections. The \textsc{gui} used for user study is viewable at \url{https://youtu.be/sL0yHV1eEPw}.
\eat{We use the \textsc{gui} framework of~\cite{catapult,kai2020} for formulating all queries in the user study. The \textsc{gui} is implemented in HTML and Linkurious.js \eat{(a graph library) }and is launched using Nodes.js (9.9.0).}

\textbf{Parameter settings.} Unless specified otherwise, we set $\eta_{min}=3$, $\eta_{max}=15$, $\gamma=30$, $\delta=3$, and $\epsilon=5$.

\textbf{Performance measures.} We measure the performance of \textsc{Tattoo} using the followings: (1) \textit{Run time}: Execution time of \textsc{Tattoo}. (2) \textit{Memory requirement} (\textsc{mr}): Peak memory usage when executing \textsc{Tattoo}. (3) \textit{Reduction ratio} (denoted as $\mu$): Given a subgraph query $Q$, $\mu=\frac{step_{total}-step_P}{step_{total}}$ where $step_P$ is the \textit{minimum} number of steps required to construct $Q$ when $\mathcal{P}$ is used and $step_{total}$ is the total number of steps needed when \textit{edge-at-a-time} approach is used. Note that the number of steps excludes vertex label assignments which is a constant for a given $Q$ regardless of the approach. For simplicity in automated performance study, we follow the same assumptions in \cite{catapult}: (1) a canned pattern $p\in\mathcal{P}$ can be used in $Q$ iff $p \subseteq Q$; (2) when multiple patterns are used to construct $Q$, their corresponding isomorphic subgraphs in $Q$ do not overlap. In the user study, we shall jettison these assumptions by allowing users to modify the canned patterns and no restrictions are imposed (\ie $step_P$ does not need to be minimum). Smaller values of $div$ imply better pattern diversity. For ease of comparison, the diversity plots are based on the inverse of $div$.

\vspace{-1ex}
\subsection{User Study}\label{sec:userStudy}
We undertake a user study to demonstrate the benefits of using our framework from a user's perspective. \eat{ Lastly, we compare \textsc{Tattoo} against four baseline approaches (\ie edge-at-a-time, graphlets, random and default patterns) by conducting a user study. 27 unpaid volunteers (ages from 20 to 35) took part in accordance to \textsc{hci} research that recommends at least 10 participants\cite{Faulkner03,LFH10}. A pre-study survey showed that our subjects had a range of familiarity and expertise with subgraph queries.} 27 unpaid volunteers (ages from 20 to 35), who were students of, or, researchers within different majors took part in the user study\eat{ in accordance to \textsc{hci} research that recommends at least 10 participants\cite{Faulkner03,LFH10}}. None of them has used our \textsc{gui} prior to the study. First, we presented a 10-min scripted tutorial of our \textsc{gui} describing how to visually formulate queries. Then, we allowed the subjects to play with the tool for 15 min\eat{. and attempt to formulate at least two queries of their choice}.\eat{ During this time we answered any questions they had about the \textsc{gui}. A video of the query formulation process using the \textsc{gui} is available at \url{https://youtu.be/sL0yHV1eEPw}.}

For each dataset, 5 subgraph queries with size in the range [10-28] are selected. These queries mimic topology of real-world queries containing various structures described in Section~\ref{sec:pattop}. To describe the queries to the participants, we provided printed visual subgraph queries.\eat{ For every query, participants were allocated some time to determine the steps needed for formulate the query visually.} A subject then draws the given query using a mouse in our \textsc{gui}. The users are asked to make maximum use of the patterns to this end. Each query was formulated 5 times by different participants. We ensure the same query set is constructed in a random order (the order of the query and the approach are randomized) to counterbalance learning effects\eat{ (see \cite{tech} for details)}.

\begin{figure}[!t]
	\centering
	\includegraphics[width=3.3in]{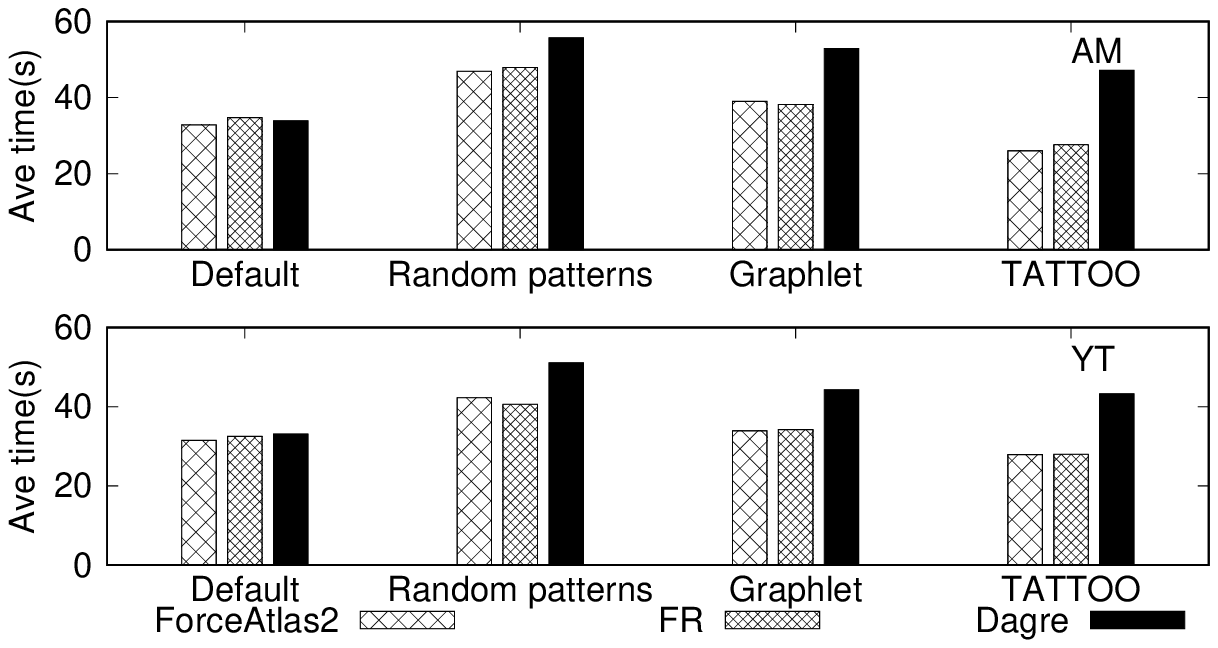}
	\vspace{-2ex}\caption{Effect of graph layout algorithm on query formulation time.}\label{fig:userStudy_orientation_time}
\end{figure}

\begin{figure}[!t]
	\centering
	\includegraphics[width=3.3in]{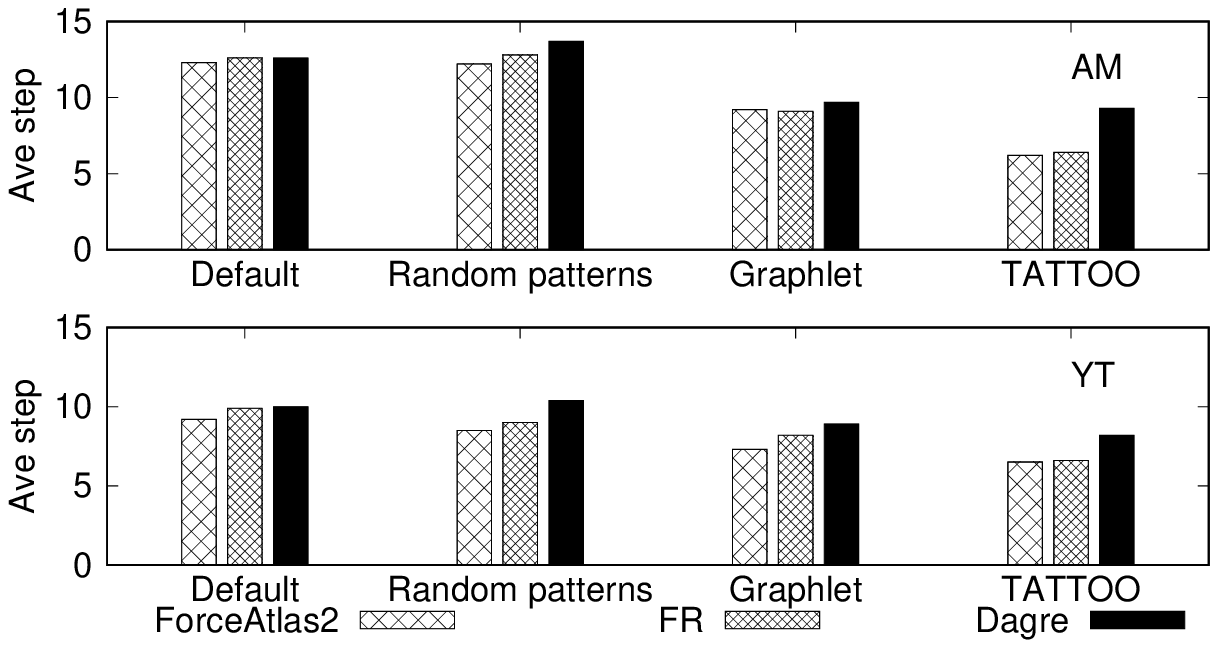}
	\vspace{-2ex}\caption{Effect of graph layout algorithm on the number of query formulation steps}\label{fig:userStudy_orientation_step}
\end{figure}

\begin{figure}[!t]
	\centering
	\includegraphics[width=3in, height=3cm]{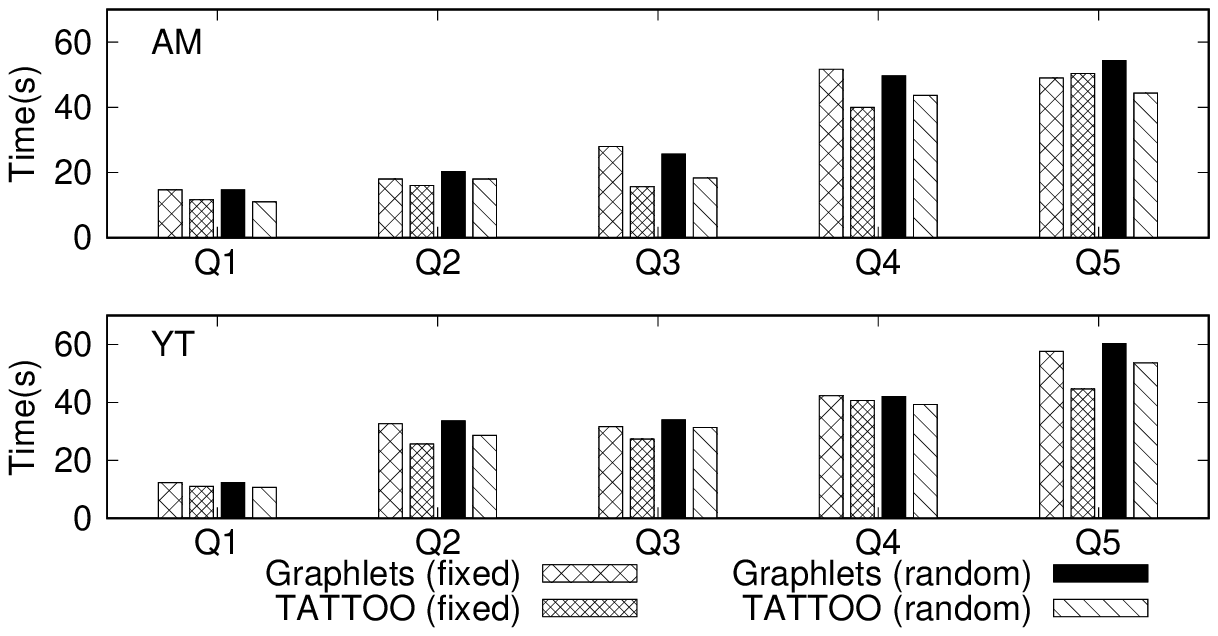}
	\vspace{0ex}\caption{Learning effect on query formulation time.}\label{fig:userStudy_random}
\end{figure}

The canned patterns on the \textsc{gui} are grouped by size and displayed using \textit{ForceAtlas2} layout~\cite{jacomy2014} in different pages according to their sizes. This multi-page-based organization yields faster average query formulation time and fewer steps compared to other alternatives\eat{ (see \cite{tech} for details)}.

\eat{Since the same query set is used repeatedly for each approach, there may be a \textit{learning effect} where volunteers start to commit the query set to memory if the study is conducted in a fixed order.\eat{ Particularly, approaches that are tested latter in the study may gain an unfair advantage over earlier approaches.} Hence, we ensure the same query set is constructed in a random order (the order of the query and the approach are randomized). This counterbalancing minimized learning effects (see \cite{tech} for details).}

\underline{\textit{Display layout of canned patterns.}} We first explore 3 different graph layout algorithms, namely, \textit{ForceAtlas2} \cite{jacomy2014}, \textit{Fruchterman Reingold} (\textsc{fr}) \cite{fruchterman1991} and \textit{Dagre} \cite{pettitt2014} to determine the most suitable layout for displaying canned patterns our \textsc{gui}. Ten participants were asked to construct a set of 5 queries each for the \textit{Amazon} and \textit{YouTube} datasets\eat{where the queries and dataset were given in a fixed order}. Each participant repeated this experiment 3 times where a different graph layout was used on the canned pattern set each time. Figures~\ref{fig:userStudy_orientation_time} and~\ref{fig:userStudy_orientation_step} plot the average query formulation time and average number of steps.\eat{ The query formulation time (\textsc{qft}) is measured from the time the user clicks the start button until the entire query formulation is completed and the stop button is clicked.} In general, the \textit{ForceAtlas2} layout yielded faster average \textsc{qft} compared to \textsc{fr} and \textit{Dagre}. Participants also took fewer steps using \textit{ForceAtlas2} compared to the rest. Hence, in subsequent experiments, \textsc{Tattoo} leverages \textit{ForceAtlas2} layout to display the canned patterns.

\begin{figure}[!t]
	\centering
	\includegraphics[width=3.2in, height=1.8cm]{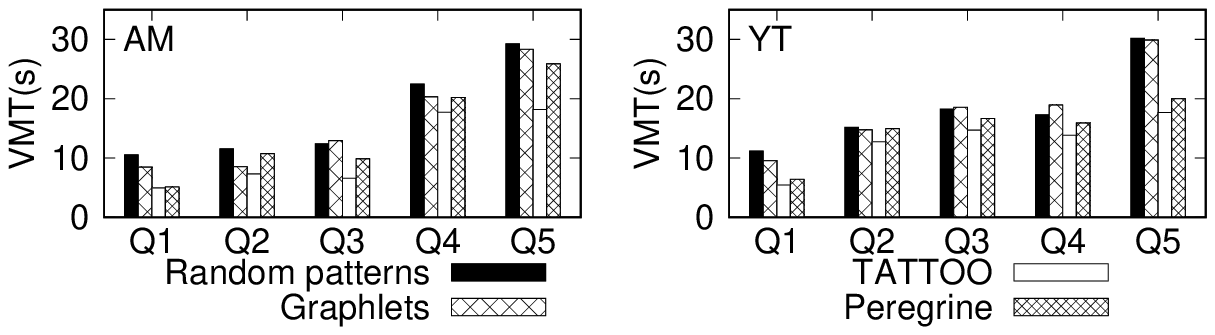}
	\vspace{-3ex}\caption{Visual mapping time of canned patterns.}\label{fig:mappingTime_userStudy}
	\vspace{-2ex}
\end{figure}

\begin{figure}[!t]
	\centering
	\includegraphics[width=3.2in, height=1.7cm]{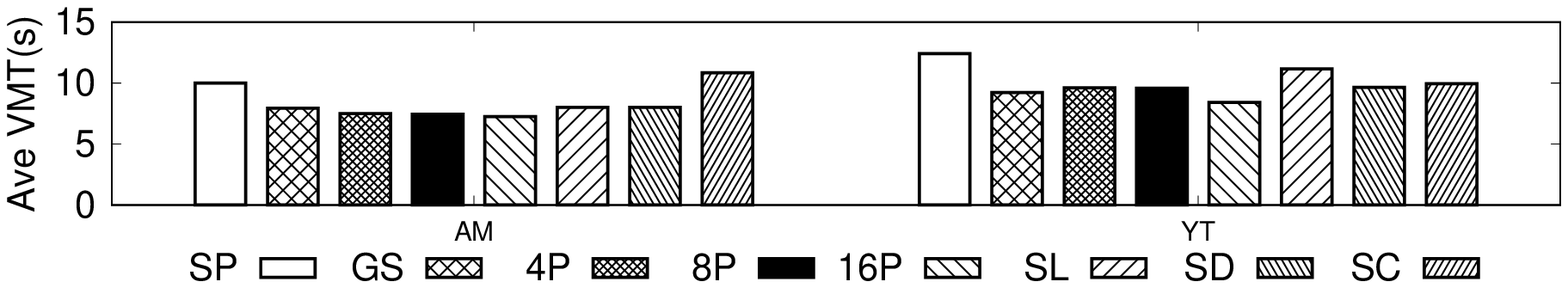}
	\vspace{-3ex}\caption{Effect of canned pattern layout on VMT.}\label{fig:vmt}
	\vspace{-3ex}
\end{figure}

\underline{\textit{Learning effect.}} Since the same query set is used repeatedly for each approach, there may be a learning effect where volunteers start to commit the query set to memory if the study is conducted in a fixed order. Particularly, approaches that are tested latter in the study may gain an unfair advantage over earlier approaches. We investigate it further with an experiment. Ten participants ($U_f$) were asked to construct 5 queries on the \textit{Amazon} and \textit{YouTube} datasets in a fixed order while another ten participants ($U_r$ where $U_f \cap U_r = \emptyset$) were asked to construct the same query set in a random order (the order of the query and the approach are randomized) to minimize learning effects. Figure~\ref{fig:userStudy_random} reports the average time taken for query formulation. Interestingly, formulation time is generally faster using \textsc{Tattoo}'s canned patterns compared to graphlets regardless of the order on the query and approach. We further examined the difference in average formulation time (\textit{i.e.}, $t_{Graphlet}-t_{\textsc{Tattoo}}$) across queries and datasets for these two orders. In particular, the average time difference for \textit{Amazon} (resp. \textit{YouTube}) is 5.5s (resp. 5.5s) and 5.9s (resp. 3.7s) for fixed order and random order, respectively. This is possibly due to the learning effect. Hence, in subsequent experiments we follow the randomized order to minimize learning effect.

\underline{\textit{Visual mapping time.}} In order to use canned patterns for query formulation, a user needs to browse the pattern set and visually map them to her query. We refer to this as \textit{visual mapping time} (\textsc{vmt}). For each pattern used, we record the \textit{pattern mapping time} (\textsc{pmt}) as the duration when the mouse cursor is in the \textit{Pattern Panel} to the time a user selects and drags it to the \textit{Query Canvas}. The \textsc{vmt} of a query is its average \textsc{pmt}. Intuitively, a longer \textsc{vmt} implies greater cognitive load on a user. Figure~\ref{fig:mappingTime_userStudy} shows the \textsc{vmt} of \textsc{tattoo} patterns, graphlets, frequent subgraphs, and random patterns on \textsc{am} and \textsc{yt} datasets. On average, \textsc{Tattoo} patterns consume the least \textsc{vmt}.

We investigate the effect of various \textsc{gui} canned pattern layout options (\ie single page ($SP$); group by size ($GS$); 4 per page ($4P$); 8 per page ($8P$); 16 per page ($16P$); sort by cognitive load ($SL$), diversity ($SD$) and sort by coverage ($SC$)) on \textsc{vmt} for \textit{AM} and \textit{YT} datasets. Note that $SP$ arranges the patterns in the order that they are identified. The plug was set to $b=(4,15,30,\ceil{\frac{30}{12}})$.\eat{ The average \textsc{vmt} is found to be in the range of [7.2-12.4] seconds.} Figure~\ref{fig:vmt} shows that the average \textsc{vmt} for layout options with multiple pages tends to be shorter (up to 33.1\%) than those in a single page (\ie $SP$, $SL$, $SD$ and $SC$). When the patterns are organized in pages, an increased number of pages reduces the need for a user to scroll and browse the patterns on a particular page and increases the need to toggle between various pages to identify useful patterns. Compared to single page options, the multi-page options ($GS$, $4P$, $8P$, $16P$) achieve superior performance primarily due to the former. Hence, the multi-page-based organization is used for our user study.

\begin{figure}[!t]
	\centering
	\includegraphics[width=3.1in, height=2.7cm]{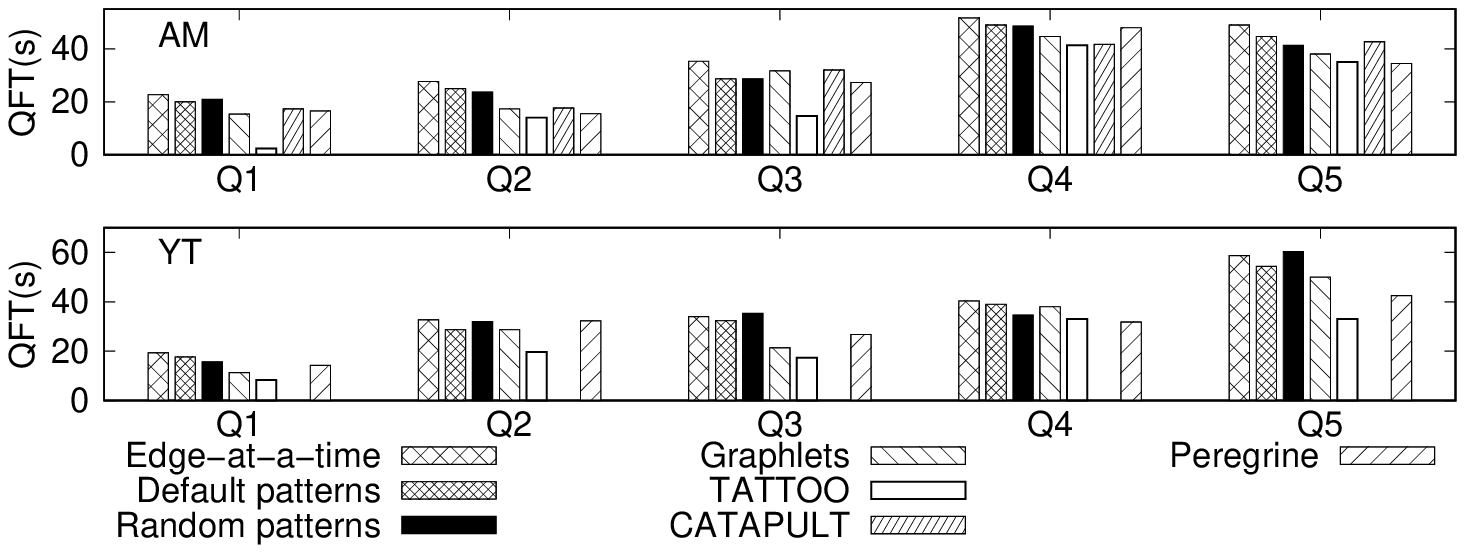}
	\vspace{-2ex}\caption{Query formulation time in user study.}\label{fig:userStudy_time}
	\vspace{-3ex}
\end{figure}

\begin{figure}[!t]
	\centering
	\includegraphics[width=\linewidth, height=3cm]{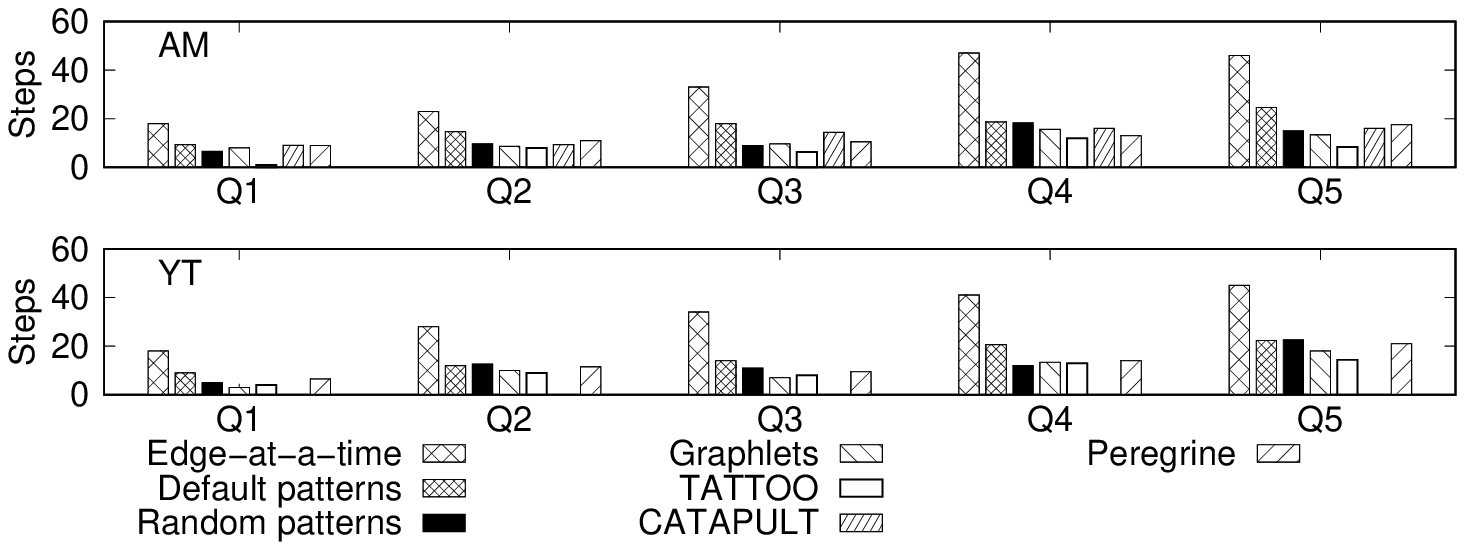}
	\vspace{-5ex}\caption{Query construction steps in user study.}\label{fig:userStudy_step}
	\vspace{-2ex}
\end{figure}

\eat{\begin{figure}[!t]
		\centering
		\includegraphics[width=3.3in, height=1.5cm]{userStudy_varyNumPatterns_time_all.eps}
		\vspace{-2ex}\caption{Effect of varying $|P|$ on QFT and steps. Query size is indicated in round brackets.}\label{fig:numGUI_time}
		\vspace{-2ex}
\end{figure}}

\begin{figure}[!t]
	\centering
	\includegraphics[width=3.1in, height=3.2cm]{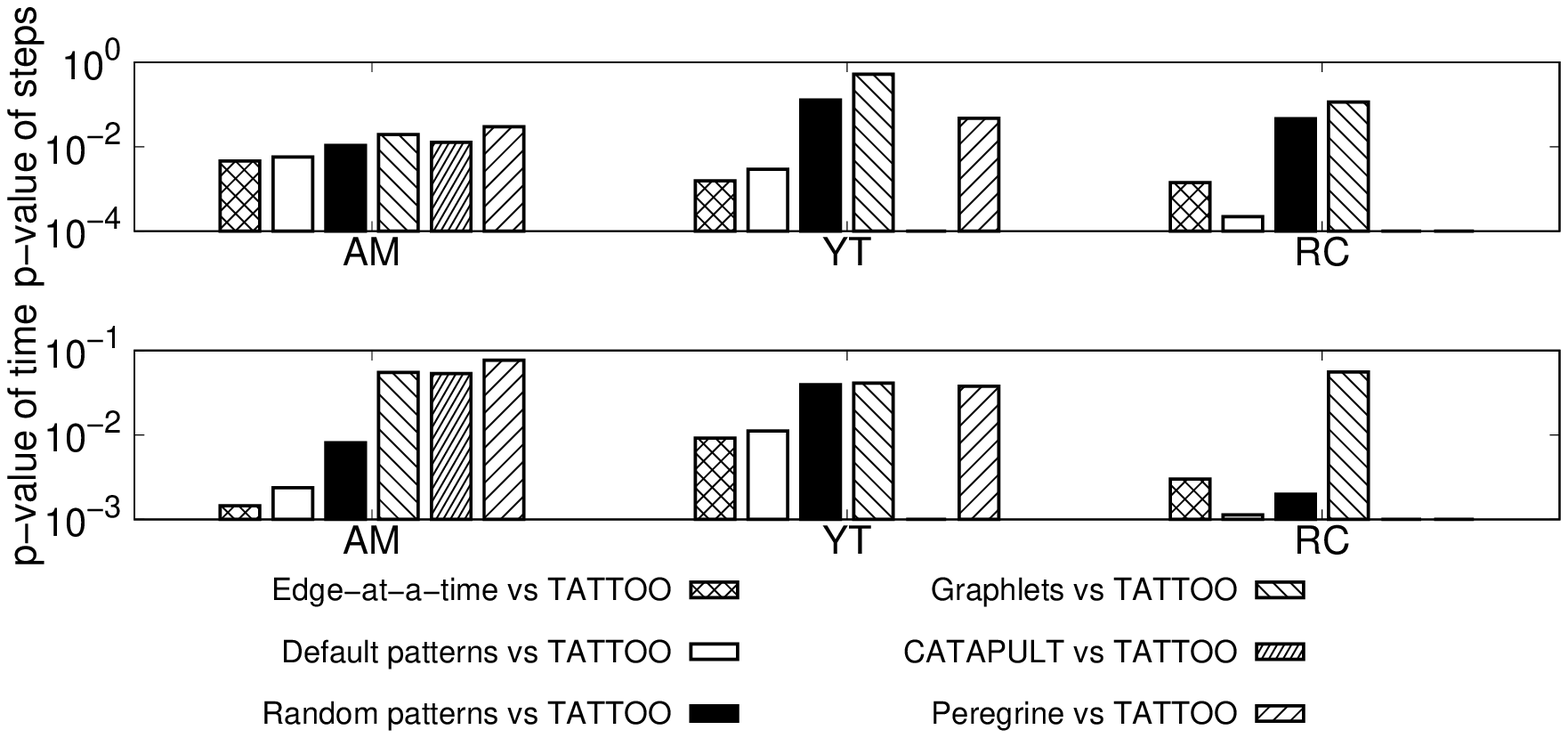}
	\vspace{-2ex}\caption{P-values of user study.}\label{fig:userStudy_pValues}
\end{figure}

\begin{figure}[!t]
	\centering
	\includegraphics[width=\linewidth]{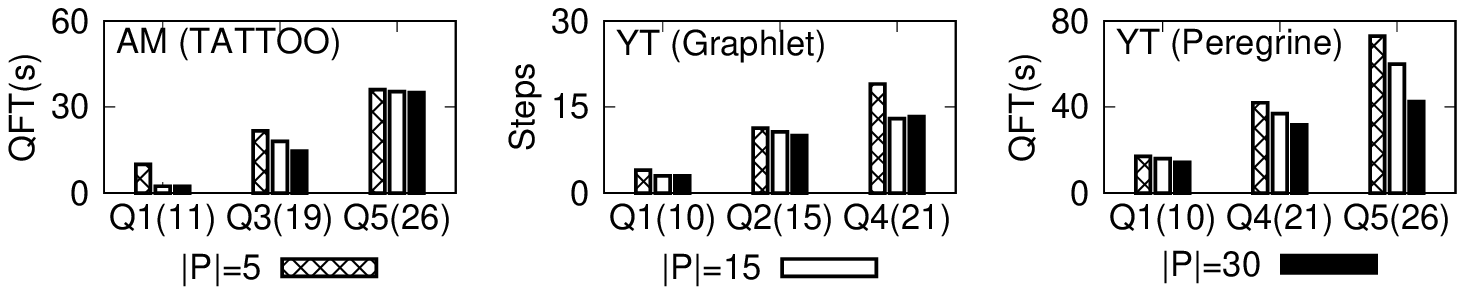}
	\vspace{-5ex}\caption{Effect of varying $|P|$ on QFT and steps. Query size is indicated in round brackets.}\label{fig:numGUI_time_step}
	\vspace{-2ex}
\end{figure}

\underline{\textit{Query formulation time (\textsc{qft}) and number of steps.}} Figures~\ref{fig:userStudy_time} and~\ref{fig:userStudy_step} plot the average \textsc{qft} and the average number of steps taken, respectively, for \textsc{am} and \textsc{yt}. Note that a \textsc{qft} includes the \textsc{vmt} and the steps include addition/deletion of nodes and edges and merger of nodes. \eat{We observe that random patterns demand the longest \textsc{qft}. This is likely due to more time needed to visually map these patterns, which tend to be dense. }As expected, the edge-at-a-time approach took the most steps. Paired t-test shows that the superior performance of \textsc{Tattoo} is statistically significant ($p<0.05$) for 79.4\% of the comparisons \eat{(see~\cite{tech} for details)}(Figure~\ref{fig:userStudy_pValues}). In particular, it takes up to 18X, 9.3X, 6.7X, 8X, 9X, and 9X fewer steps \eat{and is up to 9.8X faster }compared to edge-at-a-time, default pattern, random patterns, graphlet, frequent patterns, and \textsc{Catapult}-generated patterns, respectively. For \textsc{qft}, \textsc{Tattoo} is up to 9.7X, 8.6X, 9X, 6.6X, 7.1X, and 7.4X faster, respectively. The results are qualitatively similar in other datasets. Note that we can run \textsc{Catapult} only on \textsc{am} for reasons discussed later.

\eat{\underline{\textit{Effect of $|\mathcal{P}|$.}} The number of patterns on a \textsc{gui} may also impact a user cognitively as larger $|\mathcal{P}|$ means a user needs to browse more patterns to select relevant ones. Hence, we investigate the effect of varying $|\mathcal{P}|$ on average \textsc{qft} (Figure~\ref{fig:numGUI_time})\eat{ and the average number of steps taken (Figure~\ref{fig:numGUI_step})}. Note that there is no significant change in the number of steps taken as $|\mathcal{P}|$ varies. Interestingly, the average \textsc{qft} for larger queries (\ie $Q4$ and $Q5$) is reduced by up to 32.8\% when $|\mathcal{P}|$ is increased from 5 to 30. Increased in $|\mathcal{P}|$ results in two opposing effects: (1) longer time needed to browse and select appropriate canned patterns and longer \textsc{vmt} and (2) potentially more and larger patterns available for query construction resulting in fewer construction steps and shorter \textsc{qft}. In the case of the larger queries, the latter effect dominates. Results for other approaches are similar qualitatively.}

\underline{\textit{Effect of $|\mathcal{P}|$.}} The number of patterns on a \textsc{gui} may also impact a user cognitively as larger $|\mathcal{P}|$ means a user needs to browse more patterns to select relevant ones. Hence, we investigate the effect of $|\mathcal{P}|$ on \textsc{qft} and the number of steps (Figure~\ref{fig:numGUI_time_step}). Interestingly, \textsc{qft} and steps are reduced by average of 12\% and 22\% (maximum reduction of 77\% and 80\%), respectively, when $|\mathcal{P}|$ is increased from 5 to 30. Increase in $|\mathcal{P}|$ exposes more patterns that could be leveraged for query formulation, reducing query formulation steps. Further, it results in two opposing effects: (1) longer time needed to browse and select appropriate patterns (longer \textsc{vmt}) and (2) potentially more and larger patterns available for query construction resulting in fewer construction steps and shorter \textsc{qft}. The latter effect dominates.

\eat{\begin{figure}[!t]
		\centering
		\includegraphics[width=\linewidth, height=3.7cm]{userStudy_varyNumPatterns_step.eps}
		\vspace{-2ex}\caption{Effect of varying $|P|$ on query construction step. Query size is indicated in round brackets.}\label{fig:numGUI_step}
\end{figure}}

\eat{\textit{\underline{Comparison with graphlets, \textsc{Catapult}-generated patterns, and}} \linebreak \textit{\underline{random patterns.}} Next, we compare \textsc{Tattoo}'s canned patterns with those of graphlets (30 patterns derived from 2, 3, 4 and 5 nodes) and \textsc{Catapult}-based (Figure~\ref{fig:graphlets}) We observe that \textsc{Tattoo}'s patterns result in higher $\mu$ (up to $\sim 5\%$) and improved diversity (up to $\sim 12\%$), and are comparable in cognitive load when compared to graphlets. Hence, \textit{\textsc{Tattoo}'s canned patterns achieves greater $\mu$ and are more diversified than graphlets of comparable cognitive load}. The results are qualitatively similar for other datasets. Note that coverage is not examined since it is 100\% in all cases as all queries can be constructed using 2-node graphlet. The comparison with random patterns are reported in~\cite{tech}.}

\eat{\begin{figure}[!t]
		\centering
		\includegraphics[width=3.2in, height=1.7cm]{catapultTiming.eps}
		\vspace{-3ex}\caption{Run time of \textsc{Catapult} vs \textsc{Tattoo}.}\label{fig:catapultTiming}
\end{figure}}

\begin{figure}[!t]
	\centering
	\includegraphics[width=3in, height=3cm]{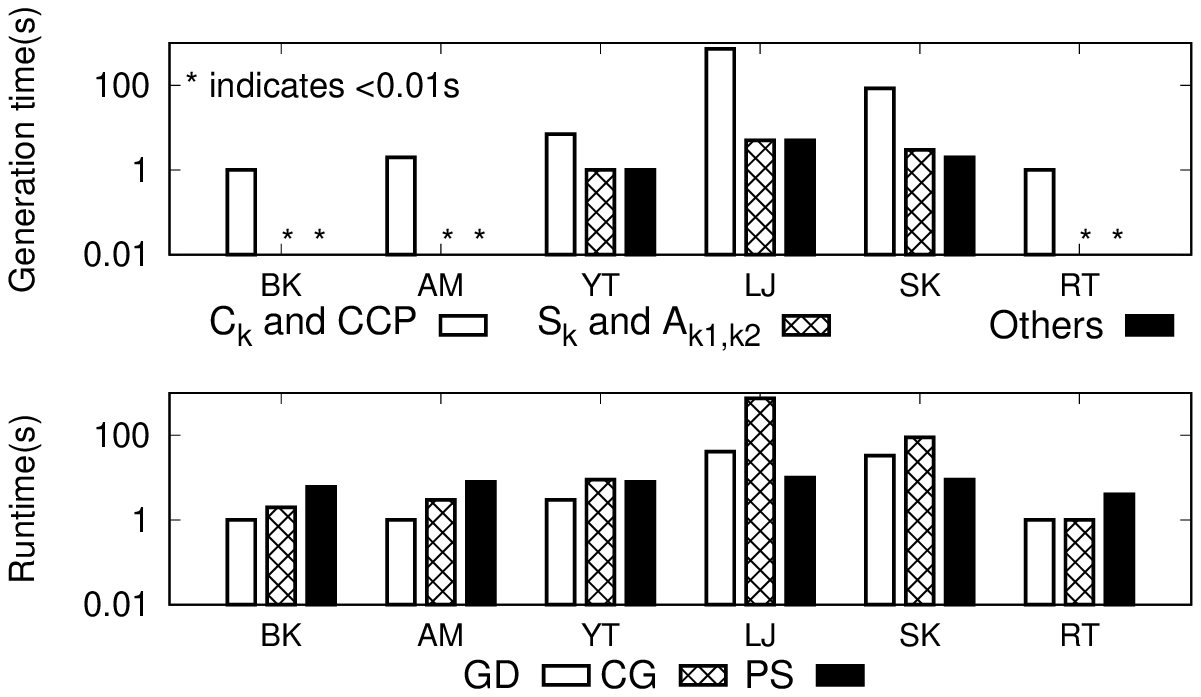}
	\vspace{-3ex}\caption{Run time. GD, CG and PS represent truss-based graph decomposition, candidate generation and pattern selection, respectively.}\label{fig:patternType_time}
	\vspace{-2ex}
\end{figure}

\begin{figure}[!t]
	\centering
	\includegraphics[width=3.3in, height=2cm]{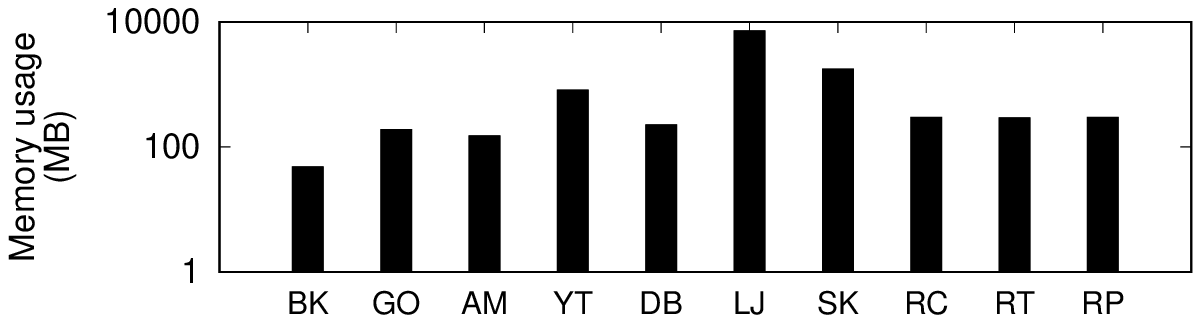}
	\vspace{-2ex}\caption{Memory requirement.}\label{fig:patternType_memory}
\end{figure}

\begin{figure}[!t]
	\centering
	\includegraphics[width=\linewidth, height=1.7cm]{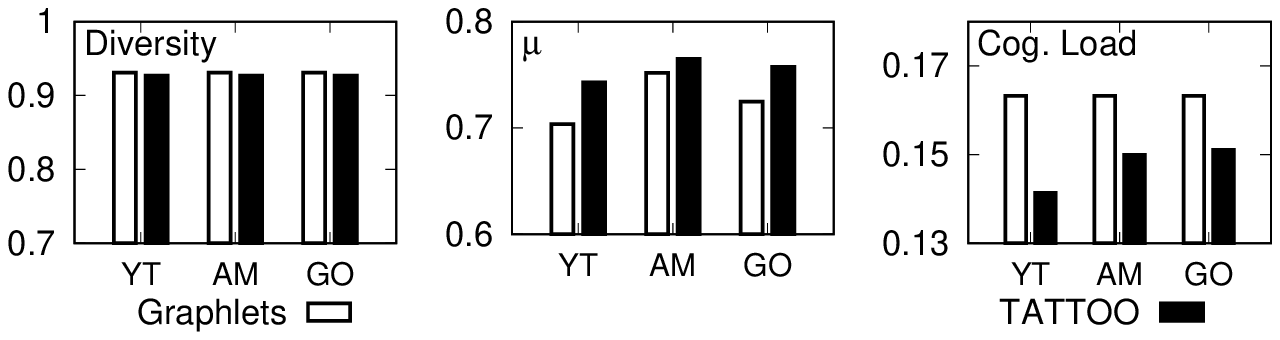}
	\vspace{-0.8cm}\caption{\textsc{TATTOO} vs graphlet  patterns.}\label{fig:graphlets}
	\vspace{-2ex}
\end{figure}

\eat{\textbf{Comparison with \textsc{Catapult}-generated patterns.} Then, we compare \textsc{Tattoo}'s canned patterns with those generated using \textsc{Catapult} on \textit{Amazon} dataset. For the latter, the single large graph is partitioned into a collection of smaller subgraphs (> 1M subgraphs) using \textsc{Metis}. Lazy sampling of \textsc{Catapult} is used to sample these subgraphs and 30 canned patterns (size in range of [4-15]) are then generated. Figure~\ref{fig:catapultTiming} shows the timing requirement of \textsc{Catapult} and \textsc{Tattoo} whereas Figure~\ref{} illustrates the canned pattern quality. Observe that \textsc{Tattoo} is 751X faster compared to \textsc{Catapult}. In terms of canned pattern quality, XXXXX}

\eat{\underline{\textit{Qualitative evaluation.}} We also conducted a post-study questionnaire to gain some qualitative feedback on \textsc{Tattoo}. Overall, \emph{all} participants prefer to use pattern-at-a-time approach compared to edge-at-a-time approach and \textsc{Tattoo}'s canned patterns are rated the most useful with average rating of 4.5 on a 5-point Likert scale. Graphlets and random patterns have average ratings of 3.5 and 1.75, respectively. More details are given in~\cite{tech}.}

\underline{\textit{Qualitative evaluation.}} We also conducted a post-study questionnaire to gain some qualitative feedback on \textsc{Tattoo}. Overall, \emph{all} participants prefer to use pattern-at-a-time approach compared to edge-at-a-time approach and \textsc{Tattoo}'s canned patterns are rated the most useful with average rating of 4.5 on a 5-point Likert scale. Graphlets, frequent and random patterns have average ratings of 3.5, 4 and 1.75, respectively.

In addition, participants were queried using an adapted \textsc{nasa-tlx}~\cite{hart1988} questionnaire (on a 5-point Likert scale) regarding the mental demand, performance and frustration level of query formulation using canned patterns generated from various approaches. Briefly, \textsc{nasa-tlx} is a commonly used tool for assessing perceived workload based on user inputs in the form of a questionnaire and consists of ratings for 6 categories\footnote{\scriptsize Mental demand can be interpreted as the cognitive load on the user; Physical demand assesses the amount of physical activity required for the task; Temporal demand is related to the time pressure experienced based on the pace of the task; Performance measures a user's satisfaction with performance of the task; Effort can be interpreted as the overall mental and physical demand needed to perform the task; Frustration level sets out to measure user's feelings (\textit{i.e.}, irritated, stressed and annoyed versus content, relaxed and complacent) during performance of the task.}, namely, mental demand, physical demand, temporal demand, performance, effort and frustration. Since we do not impose a time limit on the query formulation task, there is no temporal demand. Physical demand is also negligible in our problem setting. In particular, effort and mental demand are equivalent in this case and can be associated directly with cognitive load. The mental demand is 2.5, 3, 3.5 and 4.25 for \textsc{Tattoo}, frequent patterns, graphlets and random patterns, respectively, where larger values imply greater mental demand. In terms of performance, it is 4.25, 4, 4 and 3.75, respectively, where larger values are associated with better performance. The frustration level is 2.75, 2.75, 3 and 4, respectively, where larger values relates to more frustration. \eat{ More details are given in~\cite{tech}.} This highlights the benefits of canned patterns generated by \textsc{\textsc{Tattoo}}.

Some users elaborated on their preferences. Table~\ref{tab:feedback} lists the key comments by these users. Several users highlighted that they found edge-at-a-time approach tedious to use compared to pattern-at-a-time approach due to the repetitive task of drawing vertices and edges. This is consistent with \textsc{hci} research as remarked in Section~\ref{sec:intro}. They also felt that canned patterns of \textsc{Tattoo} are more diverse and easy to map to the query graphs. Hence, using them during query construction do not require much effort, and the patterns are useful in speeding up query formulation. Lastly, several users highlighted the usefulness of default patterns in extending other larger canned patterns during query formulation.

\begin{table}[!t]\caption{Examples of user comments.}
	\small \centering
	\vspace{0ex}\begin{tabular}{| l | p{70mm} |}
		\hline
		\textbf{Index} & \textbf{Comment} \\
		\hline
		1 & I like to use canned patterns to draw a query graph because it is faster and less tedious. Constructing the query graph each vertex and each edge at a time is just too tedious and boring!\\
		\hline
		2 & The random and graphlet patterns appear very cluttered and confusing!. It takes me considerable time to figure out if I can use them or not for my query. \\
		\hline
		3 & The default patterns are simple and easy for me to figure out. But I need to choose several of them repeatedly for a large query. The \textsc{Tattoo} patterns are a great complement to the default patterns to make the drawing faster.\\
		\hline
		4 & \textsc{Tattoo} patterns are more varied and I can usually find some patterns to use for drawing subgraph queries.\\
		\hline
		5 & I find it easy to map the \textsc{Tattoo} patterns to a query compared to other patterns like the random ones. It is actually much faster to construct a query using them than if I were to draw the vertices and edges one at a time. \\
		\hline
		6 & I like the ease of use of the \textsc{gui}. I can just use one interface to query different datasets. I don't need to switch to different interfaces for different sources.\\
		\hline
		
		7& The default patterns are very useful for extending other bigger patterns like those from \textsc{Tattoo} and graphlets when I draw a subgraph query.\\
		\hline
	\end{tabular}\label{tab:feedback}
	\vspace{0ex} \end{table}

\vspace{0ex}
\subsection{Automated Performance Study}
\eat{The preceding section reports user experience with \textsc{Tattoo}.} In this section, we evaluate \textsc{Tattoo} from the following perspectives. First, we compare the runtime and quality of patterns of \textsc{Tattoo} with the baseline approaches (\textit{Exp 1, 2}). Second, we present results that support our design decisions (\textit{Exp 3, 4, 5}). To this end, we generate \textbf{1000} queries (size [4-30]) for each dataset where 500 are randomly generated and remaining (evenly distributed) are path-like, tree, star-like, cycle-like and flower-like queries.

\textbf{Exp 1: Run time.} First, we evaluate the generation time of different patterns types in canned pattern sets. Figure~\ref{fig:patternType_time} (top) shows the results. In particular, generation of chord-like patterns requires significantly more time (up to 146\% more for \textsc{lj}) than other pattern types. This is primarily due to checks for different types of edge merger required for \textsc{ccp}s. Figure~\ref{fig:patternType_time} (bottom) reports the time taken by various phases of \textsc{Tattoo} as well as runtime of \textsc{Catapult}. \textit{ \textsc{Tattoo} selects canned patterns efficiently within a few minutes}. Observe that the time cost for the small pattern extraction phase is small in practice. \eat{As expected, candidate generation is generally the most expensive phase and pattern selection requires 10s or less. }In general, pattern selection is the most expensive phase and requires a couple of minutes or less. Results are qualitatively similar for other datasets.\eat{ Observe that our design choice of ensuring candidate pattern generation is independent from pattern selection is beneficial here as for a given network an end user can generate different \textsc{gui}s (\ie pattern selection) quickly based on her needs by using different plugs.} Figure~\ref{fig:patternType_memory} plots the memory requirement for \textsc{Tattoo}. It is largely dependent on the size of the dataset where the largest dataset $LJ$ has the greatest memory cost. \eat{Memory usage is reported in~\cite{tech}.}

Lastly, observe that \textsc{Tattoo} is 735X faster than \textsc{Catapult}, which is not designed for large networks. Except \textsc{am}, other datasets either cannot be processed by \textsc{Metis} or fail to generate patterns in a reasonable time (within 12 hrs) due to too many possible matches of unlabelled graphs that require expensive graph edit distance computation. In the sequel, we shall omit discussions on \textsc{Catapult}.

\eat{Figure~\ref{fig:patternType_memory} plots the memory requirement for \textsc{Tattoo}. It is largely dependent on the size of the dataset where the largest dataset $LJ$ has the greatest memory cost.}

\eat{\textbf{Exp 2: Comparison with graphlets and frequent subgraphs.} Next, we compare \textsc{Tattoo}'s patterns with those of graphlets (30 patterns derived from graphlets). Figure~\ref{fig:graphlets} reports the results. Observe that \textsc{Tattoo}'s patterns are superior to graphlets in all aspects.\eat{ Compared to \textsc{Catapult} (\textit{Amazon} dataset), \textsc{Tattoo} generated patterns have XXX $\mu$ (up to $\sim XXX\%$) and have slightly lower diversity ($\sim 3\%$) and cognitive load ($\sim 1\%$).}  The results are qualitatively similar for other datasets. Note that coverage is not examined since it is 100\% in all cases as all queries can be constructed using a 2-node graphlet.

We compare the canned pattern set derived from frequent subgraphs generated by\textit{ Peregrine}  (denoted as $\mathcal{P}_{P}$) to those generated by \textsc{Tattoo}.  We observe that \textit{ Peregrine} failed to extract larger size patterns (\ie $|V|\geq 8$) within 12 hrs for all networks.  Specifically, for \textsc{rp}, \textsc{rc}, and \textsc{rt} (resp. \textsc{am}), it was able to extract frequent patterns of size $|V|\leq 7$ (resp.  $|V|\leq 6$) within 2.5 hrs. For \textsc{bk} and \textsc{db} (resp. \textsc{yt}, \textsc{lj}, \textsc{sk}, and \textsc{go}) it can extract up to size $|V|\leq 5$ (resp. $|V|\leq 4$) within 2.5hrs. However, it took around 39 hrs on  \textsc{am} to yield a meaningful number of candidate patterns (994 patterns with $|V|\leq 7$ and $|E|\leq 21$)  when the minimum threshold is set to 100.  \textit{Hence,  \textsc{Tattoo}  is orders of magnitude faster than frequent pattern-based solution.}  Consequently, we restrict the canned pattern sets of both \textsc{Tattoo} and $\mathcal{P}_{P}$ to 30 patterns with $|V|\leq 7$ and $|E|\leq 21$ for \textsc{am} in our experiments for fair comparison.\eat{  Note that \textsc{Tattoo} took only 3.9 sec to select these canned patterns.} Consistent with our user study, \textsc{Tattoo}'s pattern set is superior to $\mathcal{P}_{P}$ in most aspects. The average coverage, cognitive load, diversity and $\mu$ for \textsc{Tattoo} (resp. \textit{Peregrine}) are 0.3 (resp. 0.27), 0.15 (resp. 0.14), 0.64 (resp. 0.59) and 0.23 (resp. 0.24), respectively. In summary, \textsc{Tattoo} generates better quality canned patterns.}

\textbf{Exp 2: Comparison with graphlets, frequent subgraphs, and random patterns.} Next, we compare \textsc{Tattoo}'s patterns with those of graphlets (30 patterns derived from graphlets). Figure~\ref{fig:graphlets} reports the results. Observe that \textsc{Tattoo}'s patterns are superior to graphlets in all aspects.\eat{ Compared to \textsc{Catapult} (\textit{Amazon} dataset), \textsc{Tattoo} generated patterns have XXX $\mu$ (up to $\sim XXX\%$) and have slightly lower diversity ($\sim 3\%$) and cognitive load ($\sim 1\%$).}  The results are qualitatively similar for other datasets. Note that coverage is not examined since it is 100\% in all cases as all queries can be constructed using 2-node graphlet.

We compare the canned pattern set derived from frequent subgraphs generated by\textit{ Peregrine} (denoted as $\mathcal{P}_{P}$) to those generated by \textsc{Tattoo}.\eat{ Table~\ref{tab:peregrine} shows a summary of the frequent subgraph mining using \textit{Peregrine}.} We observe that \textit{ Peregrine} failed to extract larger size patterns (\ie $|V|\geq 8$) within 12 hrs for all networks. Specifically, for \textsc{rp}, \textsc{rc}, and \textsc{rt} (resp. \textsc{am}), it was able to extract frequent patterns of size $|V|\leq 7$ (resp. $|V|\leq 6$) within 2.5 hrs. For \textsc{bk} and \textsc{db} (resp. \textsc{yt}, \textsc{lj}, \textsc{sk}, and \textsc{go}) it can extract upto size $|V|\leq 5$ (resp. $|V|\leq 4$) within 2.5hrs. However, it took around 39 hrs on \textsc{am} to yield a meaningful number of candidate patterns (994 patterns with size $|V|\leq 7$ and $|E|\leq 21$) when the minimum threshold is set to 100. \textit{Hence, \textsc{Tattoo} is orders of magnitude faster than frequent pattern-based solution.} Consequently, we restrict the canned pattern sets of both \textsc{Tattoo} and $\mathcal{P}_{P}$ to 30 patterns with $|V|\leq 7$ and $|E|\leq 21$ for \textsc{am} in our experiments for fair comparison.\eat{  Note that \textsc{Tattoo} took only 3.9 sec to select these canned patterns.} Consistent with our user study, \textsc{Tattoo}'s pattern set is superior to $\mathcal{P}_{P}$ in most aspects. The average coverage, cognitive load, diversity and $\mu$ for \textsc{Tattoo} (resp. \textit{Peregrine}) are 0.3 (resp. 0.27), 0.15 (resp. 0.14), 0.64 (resp. 0.59) and 0.23 (resp. 0.24), respectively.

The comparison with random patterns are reported in Figure~\ref{fig:random}. We observe that \textsc{Tattoo}'s patterns result in higher $\mu$, \eat{greater coverage, }and are significantly lower in cognitive load (up to 3.2X) when compared to random patterns. \textsc{Tattoo}'s patterns is up to $9\%$ less diverse compared to random patterns. The greater diversity of the random pattern set is likely due to the unrestricted way in generating the random patterns as compared to pattern generation of \textsc{Tattoo} which are partially derived from defined structures such as trusses, paths, cycles and stars. Despite the greater diversity, random patterns are more difficult to use in practice (Section~\ref{sec:userStudy}) likely due to the greater cognitive load necessary to interpret the patterns.

\begin{figure}[!t]
	\centering
	\includegraphics[width=3.1in]{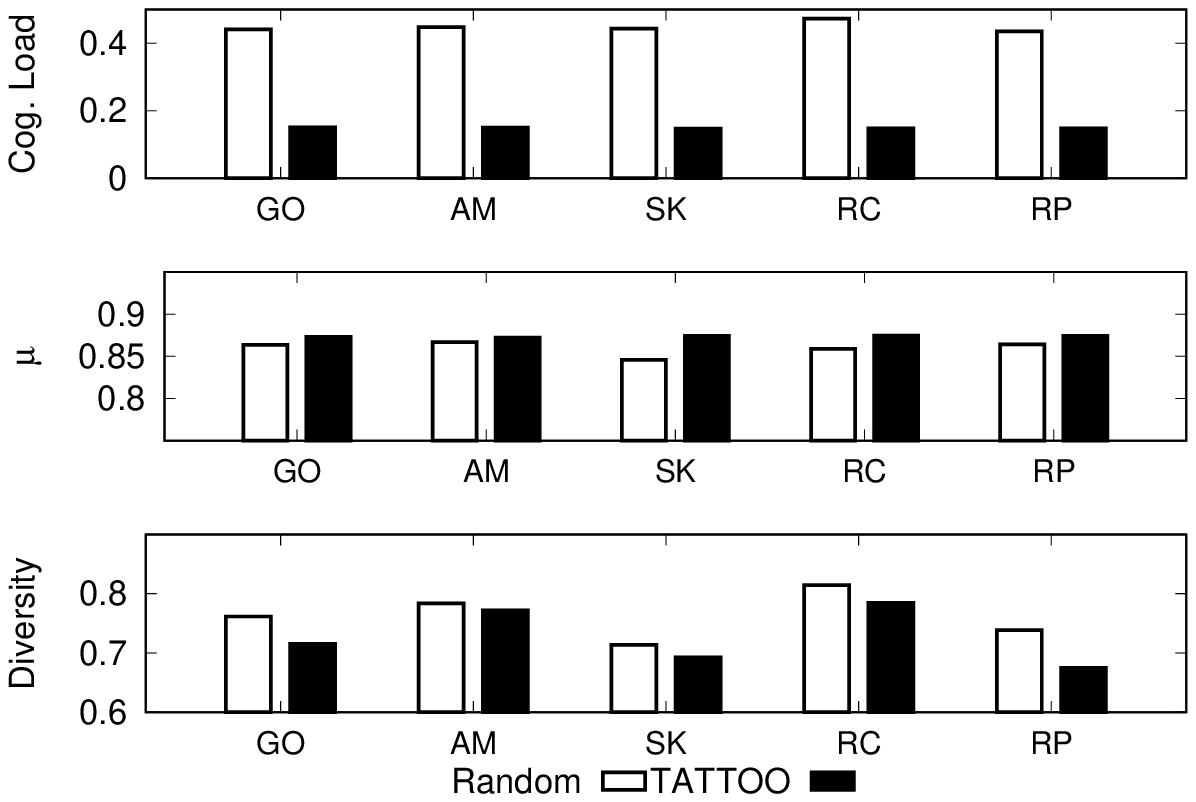}
	\vspace{0ex}\caption{TATTOO vs random patterns.}\label{fig:random}
\end{figure}

\begin{figure}[!t]
	\centering
	\includegraphics[width=0.9\linewidth, height=4cm]{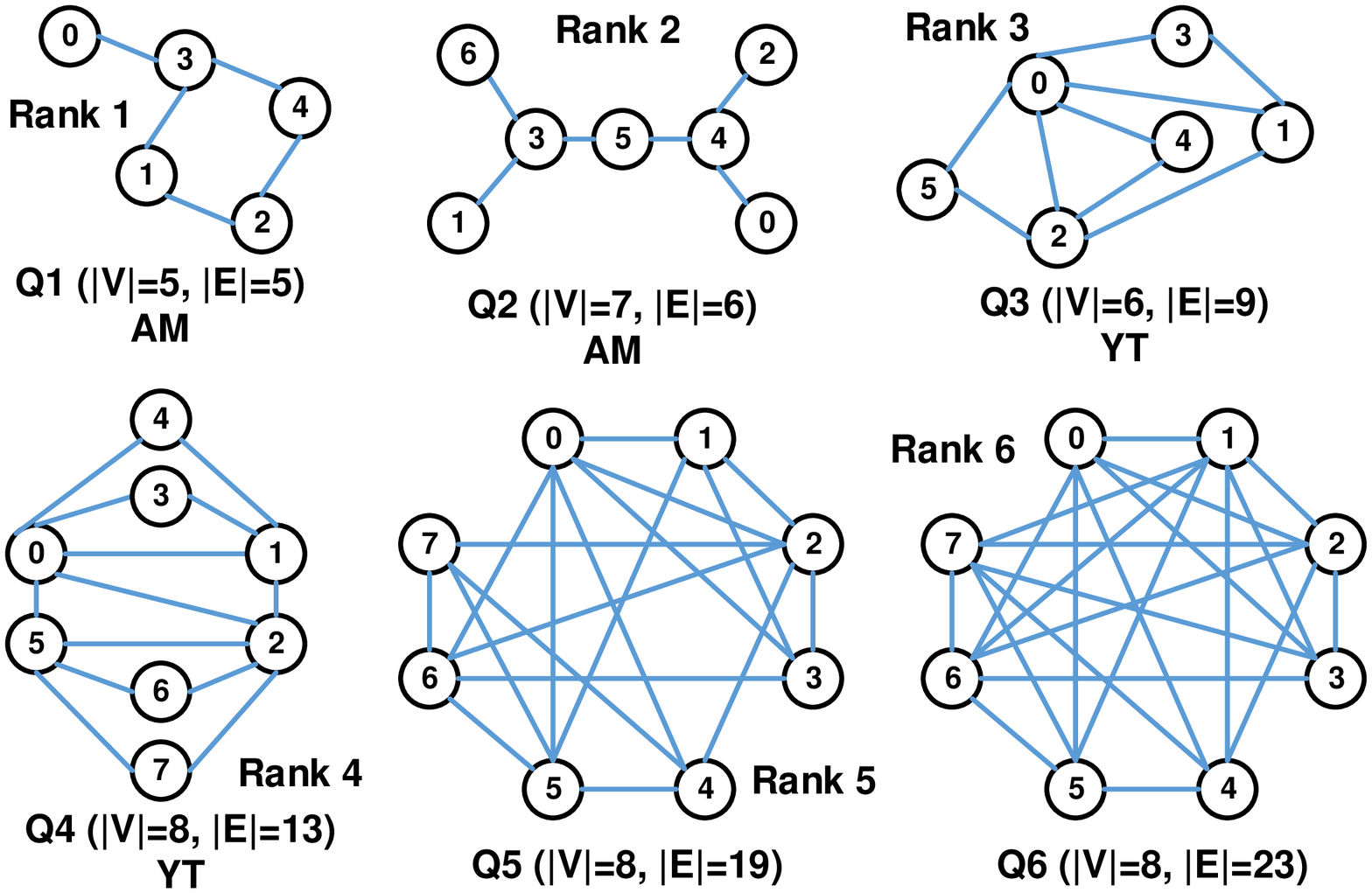}
	\vspace{-2ex}\caption{Graphs used for assessing cognitive load.}\label{fig:cogGraphs}
	\vspace{-2ex}
\end{figure}

\eat{\begin{figure}[!t]
		\centering
		\includegraphics[width=0.8\linewidth, height=1.3cm]{cog.eps}
		\vspace{-3ex}\caption{Comparison of cognitive load functions.}\label{fig:cogLoad_tau}
		\vspace{-2ex}
\end{figure}}

\begin{figure}[!t]
	\centering
	\includegraphics[width=0.9\linewidth, height=1.7cm]{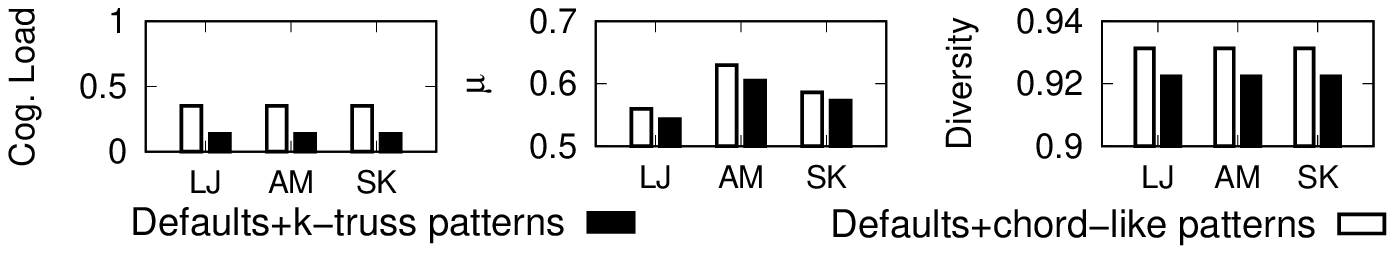}
	\vspace{-5ex}\caption{Chord-like patterns vs $k$-trusses.}\label{fig:kTruss}
	\vspace{-2ex}
\end{figure}

\begin{figure}[!t]
	\centering
	\includegraphics[width=3.1in, height=3.3cm]{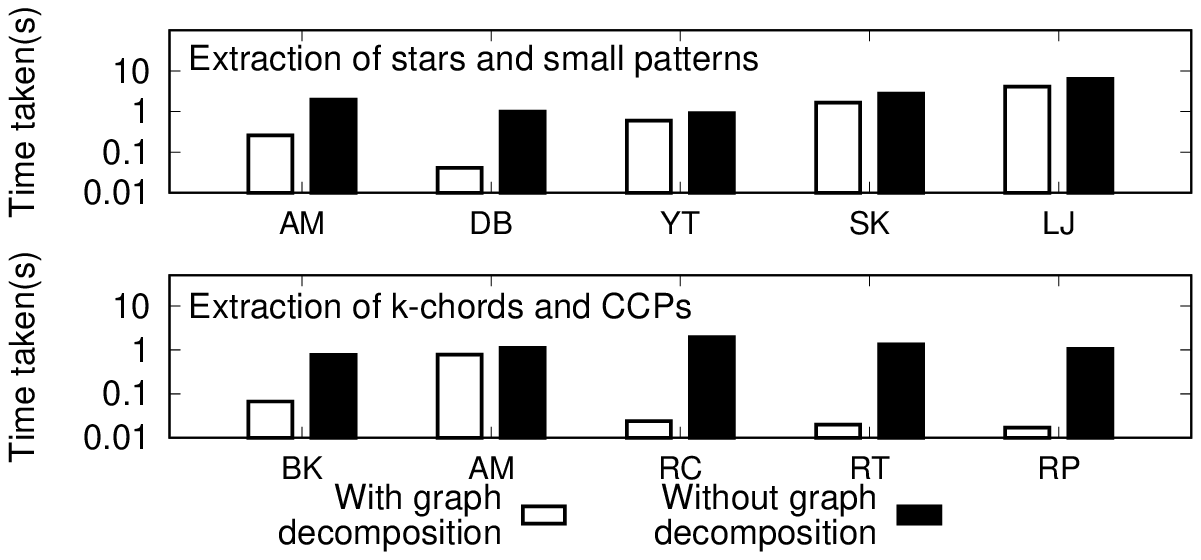}
	\vspace{-2ex}\caption{Effect of graph decomposition.}\label{fig:effectOfGraphDecomposition}
	\vspace{0ex}
\end{figure}

\textbf{Exp 3: Measuring cognitive load.} We now justify the choice of our proposed cognitive load measure. Specifically, we compare several ways of measuring cognitive load of a pattern $p$, namely, $f_{cog1}=\frac{1}{3}\sum_{x\in\{sz_p,d_p,cr_p\}}(1-e^{-x})$; $f_{cog2}=1/(1+e^{-0.5\times(sz_p+d_p+cr_p-10)})$; $f_{cog3}=sz_p+d_p+cr_p$; $f_{cog4}=sz_p\times d_p$ (used in~\cite{catapult}); and $f_{cog5}=cr_p$ (recall $sz_p$, $d_p$, $cr_p$ from Sec.~\ref{sec:quan}). 20 volunteers were asked to rank the visual representations of six graphs (Figure~\ref{fig:cogGraphs}) of varying sizes and topology, in terms of cognitive effort required to interpret these graphs.\eat{ Observe that this task simulates the effort users spend to select a pattern from the canned pattern set by visually inspecting them. Note that $Q1$ to $Q4$ are random graphs obtained from \textsc{am} and \textsc{yt}. We also include two graphs ($Q5$, $Q6$) with 8 vertices and random number of edge crossings to examine the effect of edge crossings.} A ``ground truth'' ranking for these graphs is obtained based on the average ranks assigned by the volunteers. Then, the graphs are ranked according to the five cognitive load measures and compared against the ground truth using Kendall's $\tau$~\eat{ $\tau=\frac{n_c-n_d}{0.5\times n(n-1)}$} \cite{kendall1948}\eat{ where $n$ is the number of observations; $n_c$ and $n_d$ are the number of concordant and discordant pairs, respectively}. $f_{cog2}$ and $f_{cog3}$ achieve the highest $\tau = 1$. We select $f_{cog2}$ as the cognitive load measure since it is in the range of $[0,1]$ and facilitates easy formulation of a non-negative and non-monotone submodular pattern score function (Theorem~\ref{thm:scoreNonNegativeNonMonotone}).

\textbf{Exp 4: Chord patterns vs $k$-trusses.} Next, we show the benefits of using $k$-\textsc{cp}/\textsc{ccp}s (\ie $k$-truss-like structures) compared to simply utilizing $k$-trusses as topology for canned patterns (recall from Section~\ref{sec:pattop}). We generate 100 random queries of size [4-30] from $G_T$ and these yielded 11 $k$-\textsc{cp}/\textsc{ccp}s and 3 $k$-trusses.
Observe that $k$-\textsc{cp}/\textsc{ccp}s improve both $\mu$ and diversity but have poorer cognitive load (Figure~\ref{fig:kTruss}). Here the cognitive load and diversity of a pattern set is the average value for respective measures. Importantly, more $k$-\textsc{cp}/\textsc{ccp}s than $k$-trusses satisfying the plug are generated due to relaxed structure of the former. For instance, the \textsc{rp} dataset produces $266.67\%$ more $k$-\textsc{cp}/\textsc{ccp}s due to the small size of $G_T$ (see Table~\ref{tab:graphDecomposition}). That is, $k$-trusses may not result in sufficient number of canned patterns on a \textsc{gui}. Hence, chord patterns improve the quality of canned patterns in terms of $\mu$ and diversity compared to $k$-trusses and yielded more candidate patterns.

\begin{figure}[!t]
	\centering
	\includegraphics[width=\linewidth]{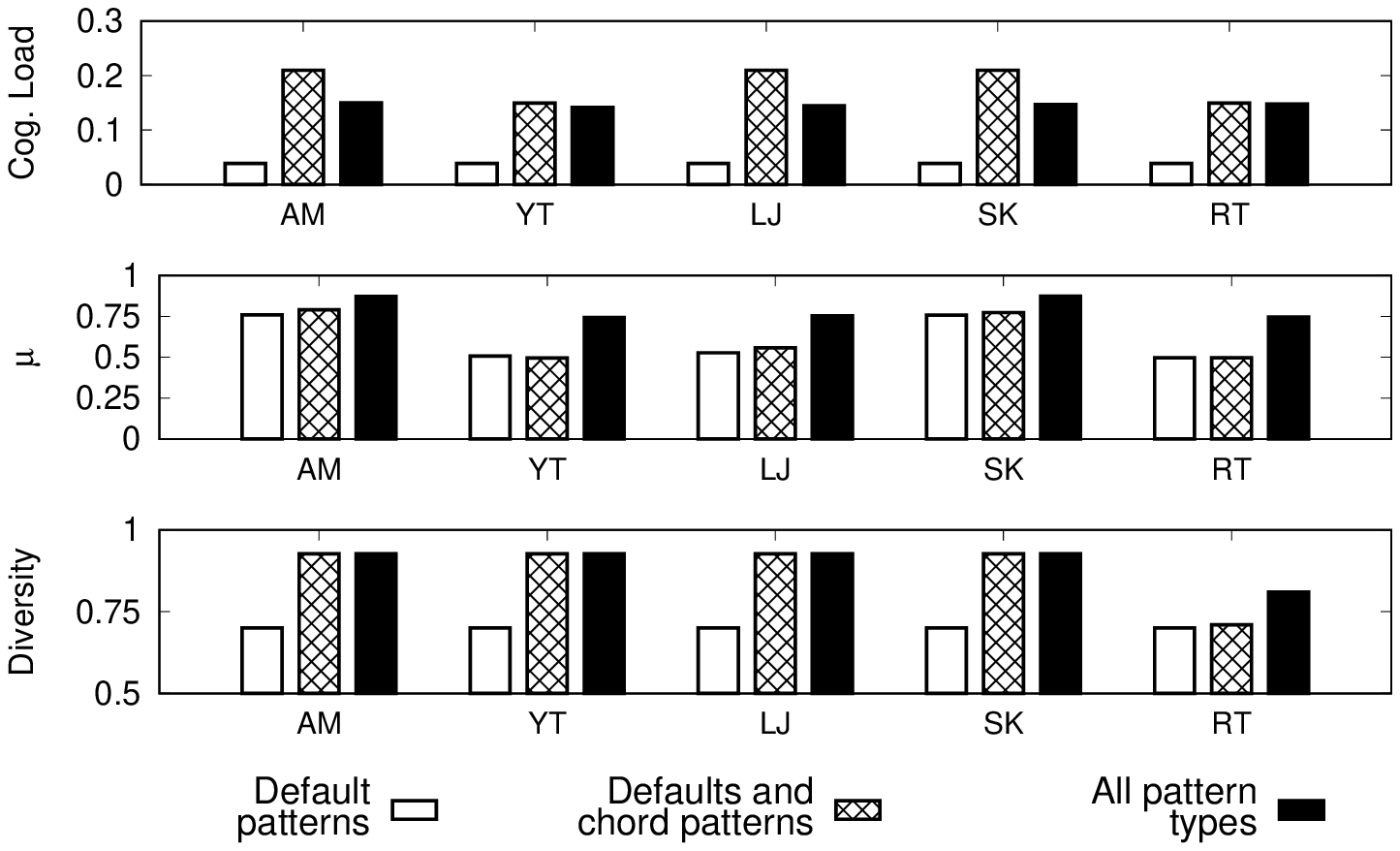}
	\vspace{-2ex}\caption{Patterns from $G_T$ and $G_O$.}\label{fig:patternType_metrics}
	\vspace{0ex}
\end{figure}

\eat{\begin{figure}[!t]
	\centering
	\includegraphics[width=0.95\linewidth, height=1.7cm]{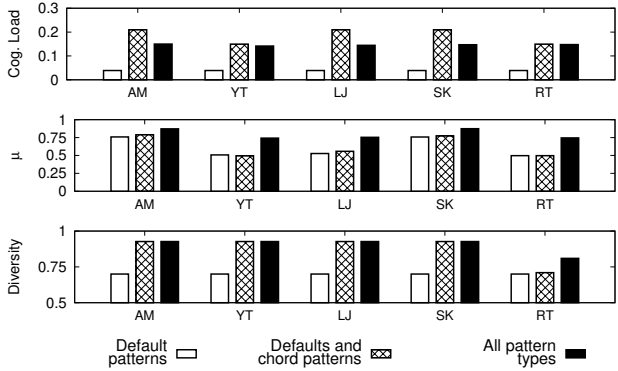}
	\vspace{-3ex}\caption{Patterns from $G_T$ and $G_O$.}\label{fig:patternType_metrics}
	\vspace{-3ex}
\end{figure}}

\textbf{Exp 5: Generating patterns from $G_T$ and $G_O$.} Lastly, we examine the (a) benefits brought by graph decomposition on pattern extraction and (b) the characteristics (\ie average cognitive load, $\mu$, and average diversity) of canned pattern sets generated from $G_T$ and $G_O$ and compare them with the default patterns (\ie $G_T$- and $G_O$-oblivious). Graph decomposition reduces the pattern extraction time for $k$-truss-like structures, as well as, stars and small patterns across \textit{all} datasets (Figure~\ref{fig:effectOfGraphDecomposition}). The effect is most prominent for extraction of $k$-truss-like structures where it is up to 81.5X faster when graph decomposition is applied. For extraction of stars and small patterns, it is up to 24.3X faster with graph decomposition. Hence, this justifies our decision to perform truss-based graph decomposition.  

Figure~\ref{fig:patternType_metrics} reports the characteristics of the canned pattern sets. Note that coverage is 100\% for the three cases as all queries can be constructed using the default pattern $D_1$ (Figure~\ref{fig:patternMining}). On the other hand, \textsc{Tattoo}'s canned patterns achieve $49\%$ and $78\%$ average coverage of $G_T$ and $G_O$, respectively\eat{ (see \cite{tech} for detailed results)}. Observe that patterns obtained from $G_O$ and $G_T$ contribute to higher $\mu$ and greater diversity over the default patterns, respectively. That is, despite 100\% coverage of the latter, it is less efficient (\ie more number of steps for query formulation) than the former. Also notice the increase in cognitive load due to the inclusion of chord-like patterns as they are likely to be denser than others. Hence, \textit{patterns from $G_T$ and $G_O$ complement the default patterns by improving $\mu$ and diversity}.\eat{ That is, choosing patterns from $G_T$ and $G_O$ enable us to generate superior patterns.}\eat{ The results are qualitatively similar for other datasets.}

\eat{\begin{figure}[!t]
		\centering
		\includegraphics[width=3.1in, height=3.5cm]{cogRRDiv_patternTypes_expt1.eps}
		\vspace{-3ex}\caption{Patterns from $G_T$ and $G_O$.}\label{fig:patternType_metrics}
\end{figure}}

\eat{\begin{table}[!t]\caption{Summary of motif generation.}
		\scriptsize
		\centering
		\vspace{0ex}\begin{tabular}{| l | l | l | l | l |}
			\hline
			\textbf{Data} & \textbf{Motif} & \textbf{No. of} & \textbf{Total time} & \textbf{Support range}\\
			&  \textbf{size} & \textbf{motifs} & \textbf{taken (min)} & \\
			\hline
			$AM$ & 3 & 2 & 0.04 & [0.7M - 7M] \\
			\cline{2-5}
			& 4 & 6 & 0.44 & [0.3M - 0.1B] \\
			\cline{2-5}
			& 5 & 21 & 6.86 & [0.1M - 6.3B] \\
			\cline{2-5}
			& 6 & 112 & 46.27 & [3K - 0.5T] \\
			\cline{2-5}
			& 7 & 853 & 2315.4 & [32 - 40T] \\
			\hline
			$YT$ & 3 & 2 & 0.03 & [3M - 1.5B] \\
			\cline{2-5}
			& 4 & 6 & 8.02 & [5M - 5.7T] \\
			\hline
		\end{tabular}\label{tab:peregrine}
		\vspace{-2ex}
\end{table}}

\eat{\begin{figure}[!t]
		\centering
		\includegraphics[width=3.1in, height=1.5cm]{peregrine.eps}
		\vspace{-3ex}\caption{TATTOO vs frequent motifs.}\label{fig:peregrine}
		\vspace{-2ex}
\end{figure}}

\begin{figure*}[!t]
	\centering
	\includegraphics[width=7in, height=7cm]{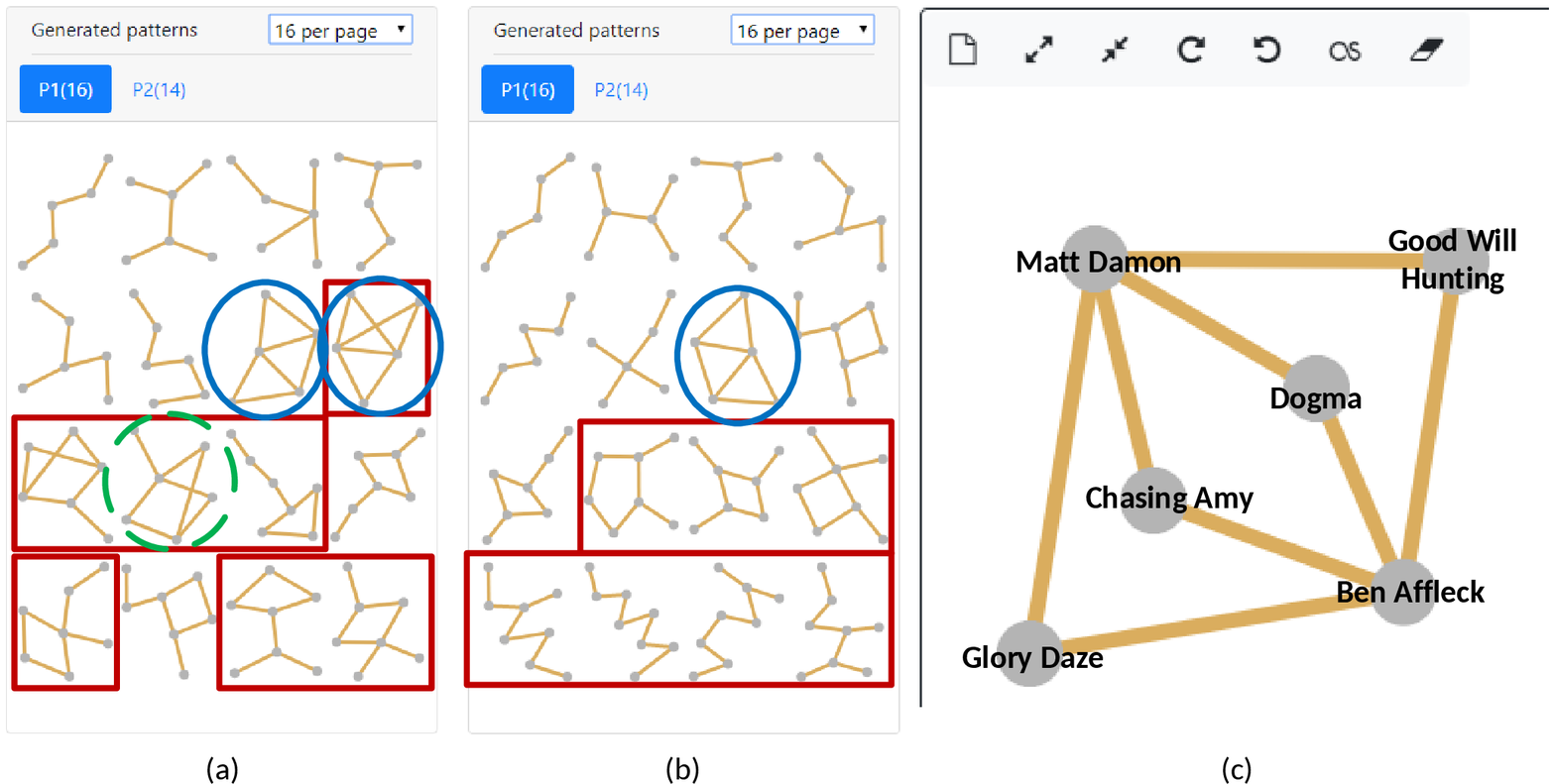}
	\vspace{-2ex}\caption{Case study.}\label{fig:gui}
\end{figure*}
\vspace{0ex}\section{Case Study}\label{sec:case}
In this section, we describe the application of \textsc{Tattoo} for visual query formulation on the \textit{Amazon} and \textit{RoadNet-TX} datasets. We generate 30 canned patterns of sizes between 4 and 15 from each of these datasets. Figures~\ref{fig:gui}(a) and~\ref{fig:gui}(b) depict some of the patterns selected from \textit{Amazon} and \textit{RoadNet-TX}, respectively. Observe that the patterns are different for different datasets. Specifically, patterns in red rectangle boxes in Figures~\ref{fig:gui}(a) are not found in Figures~\ref{fig:gui}(b) and vice versa. This emphasizes the fact that different datasets may expose different collection of canned patterns to aid efficient query formulation. Second, observe that the patterns have low cognitive load as one can easily recognize their topology with a quick glance. In particular, patterns encapsulated by blue ellipses (solid lines) in Figures~\ref{fig:gui}(a) and (b) are examples of some patterns derived from $G_T$.

We now illustrate efficient query construction using canned patterns. Suppose one is interested in making a new movie involving \textit{Ben Affleck} and \textit{Matt Damon}. She would like to identify other actors that have prior working experience with these two actors. She may construct a subgraph query such as the one in Figure~\ref{fig:gui}(c) to query the \textit{Amazon} dataset containing movie titles and associated actors. Query formulation takes 8 steps ($\sim20s$) by utilizing one canned pattern (highlighted by green ellipse with broken line in Figure~\ref{fig:gui}(a)) and adding an edge. Note that the steps taken include vertex label assignment. In comparison, edge-at-a-time  requires a total of 20 steps ($\sim39s$).

\vspace{-1ex}
\section{Conclusions \& Future Work}\label{sec:conclusion}
Canned patterns play a pivotal role in supporting efficient visual subgraph query formulation using direct-manipulation interfaces. We present \textsc{Tattoo}, which takes a data-driven approach to selecting them from the underlying network by exploiting real-world query characteristics and optimizing coverage, diversity, and cognitive load of the patterns. Our experimental study demonstrates superiority of our framework to several baselines. As part of future work, we plan to explore the problem in a distributed settings.\eat{ We believe our decomposition-based strategy is amenable towards this goal.}

\vspace{1ex}\textbf{Acknowledgements.} The first four authors are supported by the AcRF Tier-2 Grant MOE2015-T2-1-040. Wook-Shin Han was supported by Institute of Information \& communications Technology Planning \& Evaluation (IITP) grant funded by the Korea government (MSIT) (No. 2018-0-01398). Byron Choi is supported by HKBU12201518.

\bibliographystyle{abbrv}
\balance

\appendix

\section{Proofs} \label{app:proof}

\vspace{1ex}\noindent\textbf{Proof of Theorem~\ref{thm:cannedPattern} (Sketch).} The \textsc{cps} is a multi-objective optimization problem which can be reformulated as a constrained single-objective optimization problem where the objective function is $\max f_{cov}$ and the constraints are $\min (f_{sim},f_{cog})$. This reformulated problem (\ie $\max f_{cov}$) can be reduced from the maximum coverage problem, which is a classical NP-hard optimization problem \cite{karp1972}. In particular, given a number $k$ and a collection of sets $S$, the maximum coverage problem aims to find a set $S^{\prime}\subset S$ such that $|S^{\prime}|\leq k$ and the number of covered elements is maximized. In \textsc{cps}, the collection of sets $S$ is the set that consists of all possible subgraphs of the graph dataset $D$. The subset $S^{\prime}$ is the canned pattern set and $k$ is the size of the canned pattern set. The number of covered elements corresponds to the number of covered subgraphs in $D$. Note that the reformulated optimization problem is at least as hard as the maximum coverage problem since optimizing the objective may result in solutions that are sub-optimal with regards to additional imposed constraints.

\vspace{1ex}\noindent\textbf{Proof of Lemma~\ref{lem:complexityFindSimpleTrussPatterns} (Sketch).} In Algorithm~\ref{alg:findSimpleTrussPatterns}, the worst-case time complexity is due to Lines~\ref{line:simpleTrussPatternComputeTrussnessPerEdgeStart} to~\ref{line:simpleTrussPatternComputeTrussnessPerEdgeEnd} which computes the trussness of each edge $e\in E_T$ ($O(|\sqrt{E_T}|)$ \cite{wang2012}), updates $freq(C_k)$ and stores $C_k$ in the candidate pattern set. Hence, the worst-case time complexity is $O(|E_T|^{1.5}+|E_T|k_{max})$ since upper bound of $k$ is $k_{max}$. Algorithm~\ref{alg:findSimpleTrussPatterns} uses $O(|E_T|+|V_T|)$ and $O(|E_T|)$ space to hold $G_T$ and $t(e)$, respectively. Further, all possible $k$-chord patterns ($3\geq k\geq k_{max}$) and its frequency have to be stored in the worst-case ($O(k_{max})$). Hence, worst-case space complexity is $O(|E_T|+|V_T|)$ since $|E_T|+|V_T|\gg k_{max}$ for large graph in practice.

\vspace{1ex}\noindent\textbf{Proof of Lemma~\ref{lem:combinedSimple3TrussPattern} (Sketch).} The simple 3-truss pattern $C_3=(V_{c3},E_{c3})$ is simply a triangle. Hence, $\forall e=(u,v)\in E_{c3}$, there is a vertex $w$ that is adjacent to both $u$ and $v$. That is, all edges in $C_3$ have similar structure. Hence, all different types of single edge merger between two $C_3$ produces a pattern with a merged edge $e_m=(x,y)$ and vertices $x$ and $y$ have two common adjacent vertices $w_1$ and $w_2$ which is essentially $C_4$ where its truss edge correspond to the merged edge of the two $C_3$(Figure~\ref{fig:combine2Simple3Truss}).

\begin{figure}[!t]
	\centering
	\includegraphics[width=3.3in]{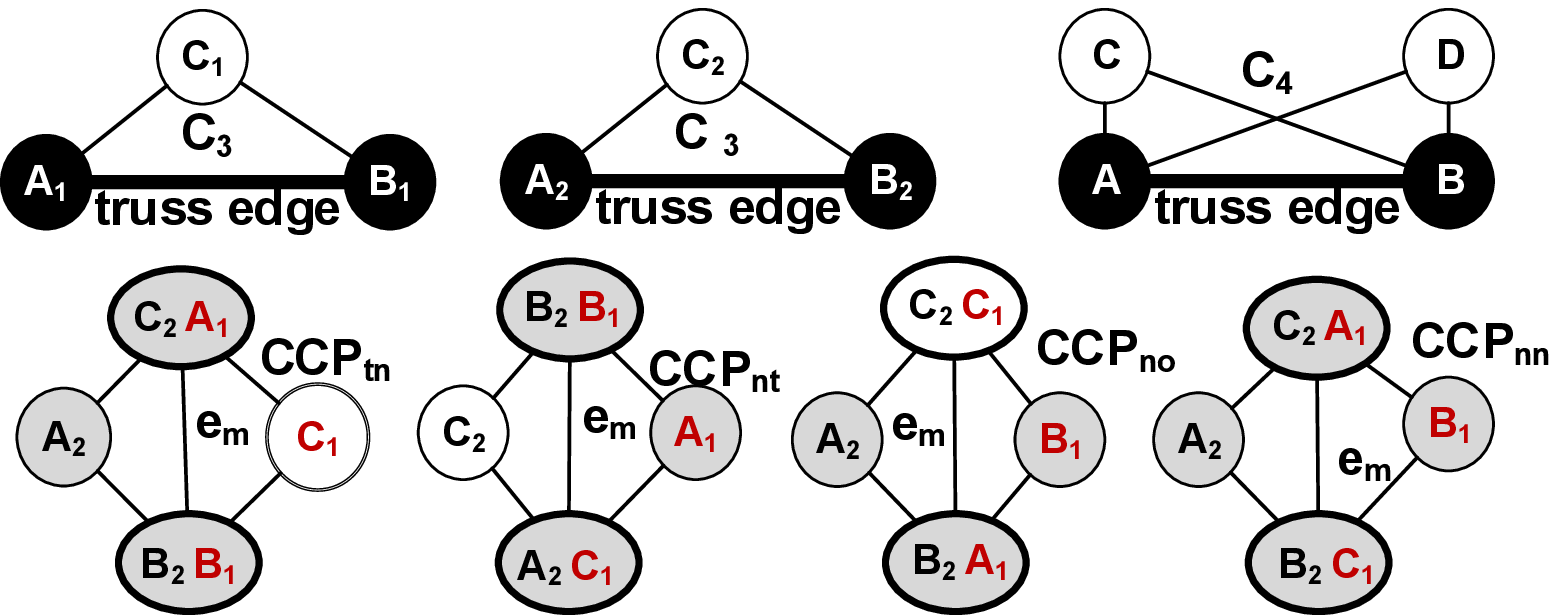}
	\vspace{-2ex}\caption{Combination of two 3-chord patterns.}\label{fig:combine2Simple3Truss}
	\vspace{2ex}
\end{figure}

\vspace{1ex}\noindent\textbf{Proof of Lemma~\ref{lem:cnn} (Sketch).} Observe that $k$-chord pattern on an edge $e=(u,v)$ implies that $k$-2 triangles in the graph contains $e$. Since $\textsc{nb}_{cc}(k,e)$ is the set of nodes $W$ adjacent to $u$ and $v$ such that $\forall w\in W$, $t((u,w))\geq k$ and $t((w,v))\geq k$, $|\textsc{nb}_{cc}(k,e)|$ is equivalent to the number of triangles around $e$. Hence, when $|\textsc{nb}_{cc}(k,e)|\geq (k-2)$, a $k$-chord pattern must exist on $e$.

\vspace{1ex}\noindent\textbf{Proof of Theorem~\ref{lem:complexityFindCombinedTrussPatterns} (Sketch).} In Algorithm~\ref{alg:findCombinedTrussPatterns}, for each edge $e\in E_T$, there are $k_1\times|\textsc{eb}_{cc}(k_1,e_1)|$ iterations that computes the procedures $GetTN$ ($O(k_{max})$), $GetNN$ ($O(k_{max}|\textsc{eb}_{max}|)$) where $\textsc{eb}_{max}$ is the $k$-CCP edge neighbourhood with the largest size. The worst-case time complexity is $O(k_{max}^2|E_T|\times|\textsc{eb}_{max}|^2)$ since $k_{max}$ is the upper bound of $k_1$. Algorithm~\ref{alg:findCombinedTrussPatterns} requires $O(|V_T|+|E_T|)$ and $O(k_{max}|E_T|)$ to store $G_T$ and \textsc{nb}, respectively. In the worst-case, all possible combinations of $CCP_{tn(k_1,k_2)}$, $CCP_{no(k_1,k_2)}$ and $CCP_{nn(k_1,k_2)}$, and their respective frequency are stored\linebreak($O(k_{max}\frac{k_{max}-3}{2})$). The worst-case space complexity is\linebreak $O(k_{max}|E_T|+|V_T|)$ since $k_{max}|E_T|+|V_T|\gg k_{max}\frac{k_{max}-3}{2}$ for large graph in practice.

\vspace{1ex}\noindent\textbf{Proof of Lemma~\ref{lem:complexityFindStarLikePatterns} (Sketch).} In the worst-case, finding the stars and asterism patterns requires performing \textsc{bfs} for each vertex in $V_O$. In the worst case, the graph is strongly connected and every other vertex in $V_O$ is visited during the \textsc{bfs}. Hence, the worst-case time complexity is $O(|V_O|^2)$. Algorithm~\ref{alg:findStarPatterns} requires $O(|V_O|+|E_O|)$ space for storing $G_O$. In the worst-case, there are $deg_{max}-\epsilon+1$ and $\frac{deg_{max}-\epsilon+1}{2}(1+(deg_{max}-\epsilon+1))$ possible $S_k$ and $A_S$, respectively. Since $deg_{max}$ occurs when every node $v\in V_O$ is connected to every other nodes in $V_O$, $deg_{max}$ has worst-case complexity $O(|V_O|)$. Hence, storage of $S_k$ and $A_S$ requires $O(|V_O|)$ and $O(|V_O|^2)$, respectively and Algorithm~\ref{alg:findStarPatterns} requires $O(|V_O|+|E_O|)$ space in the worst-case.

\vspace{1ex}\noindent\textbf{Proof of Lemma~\ref{lem:complexityFindSmallPatterns} (Sketch).} In Algorithm~\ref{alg:findSmallPatterns}, the worst-case time complexity is due to the graph isomorphism check ($O(\eta_{max}!\eta_{max})$ \cite{cordella2004}) on Line~\ref{line:findSmallPatternsFindUncommonIsomorphicStart} which is within a for-loop with maximum of $|V_R|$ iterations. Hence, the worst-case time complexity is $O(\eta_{max}|V_R|\eta_{max}!)$. Algorithm~\ref{alg:findSmallPatterns} requires $O(|V_R|+|E_R|)$ space for storing $G_R$. Since every $k$-path ($P_k$), $k$-cycle ($Y_k$) and subgraphs with unique topology ($U$) consists of multiple nodes, the number of possible $P_k$, $Y_k$ and $U$ is less than $|V_R|$ and the storage required will be $O(|V_R|)$. Hence, the worst-case space complexity is $O(|V_R|+|E_R|)$.

\vspace{1ex}\noindent\textbf{Proof of Lemma~\ref{lem:coverageSubmodular}.} Submodular functions satisfies the property of diminishing marginal returns. That is given a set of $n$ elements ($N$), a function $f(.)$ is submodular if for every $A\subseteq B\subseteq N$ and $j\notin B$, $f(A\bigcup \{j\})-f(A)\geq f(B\bigcup \{j\})-f(B)$. Given a graph $G$ and canned pattern sets $\mathcal{P}_A$ and $\mathcal{P}_B$ where $\mathcal{P}_A\subseteq\mathcal{P}_B$, let the coverage of $\mathcal{P}_A$ and $\mathcal{P}_B$ be $f_{cov}(\mathcal{P}_A)$ and $f_{cov}(\mathcal{P}_B)$, respectively. Observe that $\mathcal{P}_B$ consists of $\mathcal{P}_A$ and additional patterns (\ie $\mathcal{P}^{\prime}=\mathcal{P}_B\setminus\mathcal{P}_A$). For each canned pattern $p\in\mathcal{P}^{\prime}$, we let $s=\min(|f_{cov}(p)|,|f_{cov}(\mathcal{P}_A)|)$ and $K$ denotes the overlapping set $f_{cov}(p)\bigcap f_{cov}(\mathcal{P}_A)$. The coverage of $p$ falls under one of four possible scenarios, namely, (1) $K=f_{cov}(p)$ if $s=|f_{cov}(p)|$, (2) $K=f_{cov}(\mathcal{P}_A)$ if $s=|f_{cov}(\mathcal{P}_A)|$, (3) $K$ is an empty set and (4) otherwise (\ie $0<|f_{cov}(p)\bigcap f_{cov}(\mathcal{P}_A)|<s$).

In the case where coverage of every $p$ falls under scenario 1, then $f_{cov}(\mathcal{P}_A)=f_{cov}(\mathcal{P}_B)$. Should any $p$ falls under scenario 2, 3 or 4, then $f_{cov}(\mathcal{P}_A)\subset f_{cov}(\mathcal{P}_B)$. Hence, $f_{cov}(\mathcal{P}_A)\subseteq f_{cov}(\mathcal{P}_B)$. Consider a canned pattern $p^{\prime}\notin\mathcal{P}_B$, let $t=\min(|f_{cov}(p^{\prime})|,|f_{cov}(\mathcal{P}_A)|)$. Suppose $f_{cov}(p^{\prime})\bigcap f_{cov}(\mathcal{P}_A)=f_{cov}(p^{\prime})$ where $|f_{cov}(p^{\prime})|<|f_{cov}(\mathcal{P}_A)|$ (Scenario 1), then $f_{cov}(\mathcal{P}_A\bigcup \{p^{\prime}\})-f_{cov}(\mathcal{P}_A)$ is an empty set. Note that we use the minus and set minus operator interchangeably in this proof. Since $f_{cov}(\mathcal{P}_A)\subseteq f_{cov}(\mathcal{P}_B)$, $f_{cov}(\mathcal{P}_B\bigcup \{p^{\prime}\})=f_{cov}(\mathcal{P}_B)$. Hence, $f_{cov}(\mathcal{P}_A\bigcup \{p^{\prime}\})-f_{cov}(\mathcal{P}_A)= f_{cov}(\mathcal{P}_B\bigcup \{p^{\prime}\})\linebreak-f_{cov}(\mathcal{P}_B)$.

Now, consider $f_{cov}(p^{\prime})\bigcap f_{cov}(\mathcal{P}_A)=f_{cov}(\mathcal{P}_A)$ where\linebreak$|f_{cov}(p^{\prime})|>|f_{cov}(\mathcal{P}_A)|$ (Scenario 2). $f_{cov}(\mathcal{P}_A\bigcup \{p^{\prime}\})-f_{cov}(\mathcal{P}_A)\linebreak=f_{cov}(p^{\prime})-f_{cov}(\mathcal{P}_A)$ where $f_{cov}(\mathcal{P}_A)\subset f_{cov}(p^{\prime})$. Let $L$ and $M$ be $f_{cov}(p^{\prime})\setminus f_{cov}(\mathcal{P}_A)$ and $f_{cov}(\mathcal{P}_B)\setminus f_{cov}(\mathcal{P}_A)$, respectively. Observe that, similar to previous observation, it is possible for (1) $L$ to be fully contained in $M$ if $|L|<|M|$, (2) $M$ to be fully contained in $L$ if $|M|<|L|$, (3) $L\bigcap M$ to be empty or (4) otherwise (\ie $0<|L\bigcap M|<t$ where $t=\min(|L|,|M|)$). Hence, $|L\bigcap M|\in[0,t]$. When $|L\bigcap M|=0$, $f_{cov}(\mathcal{P}_A\bigcup \{p^{\prime}\})-f_{cov}(\mathcal{P}_A)= f_{cov}(\mathcal{P}_B\bigcup \{p^{\prime}\})-f_{cov}(\mathcal{P}_B)$. Otherwise, there are some common graphs covered by $L$ and $M$, resulting in $f_{cov}(\mathcal{P}_B\bigcup \{p^{\prime}\})-f_{cov}(\mathcal{P}_B)=L\setminus (L\bigcap M)$. Hence, $|f_{cov}(\mathcal{P}_A\bigcup \{p^{\prime}\})\linebreak-f_{cov}(\mathcal{P}_A)|>|f_{cov}(\mathcal{P}_B\bigcup \{p^{\prime}\})-f_{cov}(\mathcal{P}_B)|$. Taken together, for scenario 2, $|f_{cov}(\mathcal{P}_A\bigcup \{p^{\prime}\})-f_{cov}(\mathcal{P}_A)|\geq |f_{cov}(\mathcal{P}_B\bigcup \{p^{\prime}\})-f_{cov}(\mathcal{P}_B)|$.

For scenario 3, it is similar to scenario 2 where $L$ is $f_{cov}(p^{\prime})$ instead of $f_{cov}(p^{\prime})\setminus f_{cov}(\mathcal{P}_A)$. $f_{cov}(\mathcal{P}_A\bigcup \{p^{\prime}\})-f_{cov}(\mathcal{P}_A)=L$ and $f_{cov}(\mathcal{P}_B\bigcup \{p^{\prime}\})-f_{cov}(\mathcal{P}_B)=L\setminus (L\bigcap M)$. Since $|L\bigcap M|\in[0,t]$, $|f_{cov}(\mathcal{P}_A\bigcup \{p^{\prime}\})-f_{cov}(\mathcal{P}_A)|\geq |f_{cov}(\mathcal{P}_B\bigcup \{p^{\prime}\})-f_{cov}(\mathcal{P}_B)|$.

For scenario 4, it is the same as scenario 3 except that $L=f_{cov}(p^{\prime})\setminus (f_{cov}(\mathcal{P}_A)\bigcap f_{cov}(p^{\prime}))$. Observe that $|f_{cov}(\mathcal{P}_A\bigcup \{p^{\prime}\})-f_{cov}(\mathcal{P}_A)|\geq |f_{cov}(\mathcal{P}_B\bigcup \{p^{\prime}\})-f_{cov}(\mathcal{P}_B)|$ due to $|L\bigcap M|\in[0,t]$.

Hence, in all cases, $|f_{cov}(\mathcal{P}_A\bigcup \{p^{\prime}\})-f_{cov}(\mathcal{P}_A)|\geq\linebreak|f_{cov}(\mathcal{P}_B\bigcup \{p^{\prime}\})-f_{cov}(\mathcal{P}_B)|$ applies and $f_{cov}(.)$ is submodular.

\vspace{1ex}\noindent\textbf{Proof of Lemma~\ref{lem:similaritySubmodular}.} We begin by stating the \textit{first order difference}. Given a submodular function $f(.)$, for every $\mathcal{P}_A\subseteq\mathcal{P}_B\subseteq D$ and every $p\subset D$ such that $p\notin\mathcal{P}_A,\mathcal{P}_B$, the first order difference states that $f(\mathcal{P}_A\bigcup \{p\})-f(\mathcal{P}_A)\geq f(\mathcal{P}_B\bigcup \{p\})-f(\mathcal{P}_B)$.

Given a graph $G$, a canned pattern $p\notin\mathcal{P}_B$ and canned pattern sets $\mathcal{P}_A$ and $\mathcal{P}_B$ where $\mathcal{P}_A\subseteq\mathcal{P}_B$, let the similarity of $\mathcal{P}_A$ and $\mathcal{P}_B$ be $f_{sim}(\mathcal{P}_A)$ and $f_{sim}(\mathcal{P}_B)$, respectively. $f_{sim}(\mathcal{P}_B\bigcup \{p\})-f_{sim}(\mathcal{P}_B)=\sum_{p_i\in\mathcal{P}_B}sim(p,p_i)$ and $f_{sim}(\mathcal{P}_A\bigcup \{p\})-f_{sim}(\mathcal{P}_A)=\sum_{p_i\in\mathcal{P}_A}sim(p,p_i)$. Since $sim(p_i,p_j)\geq 0$ $\forall p_i,p_j\subset G$, $\mathcal{P}_A\subseteq\mathcal{P}_B$ and by definition of the first order difference, $f_{sim}(.)$ is supermodular. The proof is similar for $f_{cog}(.)$.

\vspace{1ex}\noindent\textbf{Proof of Theorem~\ref{thm:scoreNonNegativeNonMonotone} (Sketch).} Consider a partial pattern set $\mathcal{P}^{\prime}$ and a candidate pattern $p$. Suppose $p$ does not improve the set coverage of $\mathcal{P}^{\prime}$ and adds a high cost in terms of cognitive load and diversity. Then, $s(\mathcal{P}^{\prime})>s(\mathcal{P}^{\prime}\bigcup \{p\})$. Hence, the score function $s(.)$ is \textit{non-monotone}. Since $f_{cov}(\mathcal{P}^{\prime}),f_{sim}(\mathcal{P}^{\prime}),f_{cog}(\mathcal{P}^{\prime})\in[0,|\mathcal{P}^{\prime}|]$, $f_{cov}(\mathcal{P}^{\prime})-f_{sim}(\mathcal{P}^{\prime})-f_{cog}(\mathcal{P}^{\prime})$ is in the range [-2$|\mathcal{P}^{\prime}|$,$|\mathcal{P}^{\prime}|$]. Hence, $\frac{1}{3|\mathcal{P}^{\prime}|}(f_{cov}(\mathcal{P}^{\prime})-f_{sim}(\mathcal{P}^{\prime})-f_{cog}(\mathcal{P}^{\prime})+2|\mathcal{P}^{\prime}|)$ (Definition~\ref{def:patternScore}) is in the range $[0,1]$ and is \textit{non-negative}. Since supermodular functions are negations of submodular functions and that non-negative weighted sum of submodular functions preserve submodular property \cite{fujishige2005}, $s(\mathcal{P}^{\prime})$ is submodular. Note that adding a constant (\ie $\frac{2}{3}$) does not change the submodular property~\cite{bhowmik2014} and ensures that $s(\mathcal{P}^{\prime})$ is non-negative. The scaling factors of $\alpha_{f_{cov}}=\alpha_{f_{sim}}=\alpha_{f_{cog}}=\frac{1}{3|\mathcal{P}^{\prime}|}$ further bounds $s(\mathcal{P}^{\prime})$ within the range $[0,1]$.

\vspace{1ex}\noindent\textbf{Proof of Lemma~\ref{lem:crossing}.} Consider a graph $G=(V,E)$ with $cr$ crossings. Since each crossing can be removed by removing an edge from $G$, a graph with $|E|-cr$ edges and $|V|$ vertices containing no crossings (\ie planar graph). Since $|E|\leq 3|V|-6$ for planar graph (\ie Euler's formula), hence, $|E|-cr\leq 3|V|-6$ for $|V|\geq 3$. Rewriting the inequality, we have $cr\geq|E|-3|V|+6$.

\vspace{1ex}\noindent\textbf{Proof of Theorem~\ref{thm:greedyApprox}.} Let $A_i$ be an event fixing all the random decisions of \textit{Greedy} for every iteration $i$ and $\mathcal{A}_i$ be the set of all possible $A_i$ events. We denote $s(\mathcal{P}_{i-1}\bigcup \{p_i\})-s(\mathcal{P}_{i-1})$ as $s_{p_i}(\mathcal{P}_{i-1})$. Further, let the desired size of $\mathcal{P}$ be $\gamma$, $1\leq i\leq\gamma$ and $A_{i-1}\in\mathcal{A}_{i-1}$. Unless otherwise stated, all the probabilities, expectations and random quantities are implicitly conditioned on $A_{i-1}$. Consider a set $M^{\prime}_i$ containing the patterns of $OPT\setminus\mathcal{P}_{i-1}$ plus enough dummy patterns to make the size of $M^{\prime}_i$ exactly $\gamma$.

Note that $\mathbb{E}[s_{p_i}(\mathcal{P}_{i-1})]=\gamma^{-1}\cdot\sum_{p\in M_i} s_{p}(\mathcal{P}_{i-1})\geq \gamma^{-1}\cdot\sum_{p\in M^{\prime}_i}\linebreak s_{p} (\mathcal{P}_{i-1})=\gamma^{-1}\cdot\sum_{p\in OPT\setminus\mathcal{P}_{i-1}}s_{p}(\mathcal{P}_{i-1})\geq\frac{s(OPT\bigcup\mathcal{P}_{i-1})-s(\mathcal{P}_{i-1})}{\gamma}$ \cite{buchbinder2014}, where the first inequality follows from the definition of $M_i$ (\ie set of ``good'' candidate patterns) and the second from the submodularity of $s(.)$ Unfixing the event $A_{i-1}$ and taking an expectation over all possible such events, $\mathbb{E}[s_{p_i}(\mathcal{P}_{i-1})]\geq\frac{\mathbb{E}[s(OPT\bigcup\mathcal{P}_{i-1})]-\mathbb{E}[s(\mathcal{P}_{i-1})]}{\gamma}
\geq\frac{(1-\frac{1}{\gamma})^{i-1}\cdot s(OPT)-\mathbb{E}[s(\mathcal{P}_{i-1})]}{\gamma}$, where the second inequality is due to observation that for every $0\geq i\geq\gamma$, $\mathbb{E}[s(OPT\bigcup\mathcal{P}_i)]\geq(1-\frac{1}{\gamma})^i\cdot s(OPT)$\cite{buchbinder2014}.

We now prove by induction that $\mathbb{E}[s(\mathcal{P}_i)]\geq\frac{i}{\gamma}\cdot(1-\frac{1}{\gamma})^{i-1}\cdot s(OPT)$. Note that this is true for $i=0$ since $s(\mathcal{P}_0)\geq 0=\frac{0}{\gamma}\cdot(1-\frac{1}{\gamma})^{-1}\cdot s(OPT)$. Further, we assume that the claim holds for every $i^{\prime}<i$. Now, we prove it for $i>0$. $\mathbb{E}[s(\mathcal{P}_i)]=\mathbb{E}[s(\mathcal{P}_{i-1})]+\mathbb{E}[s_{p_i}(\mathcal{P}_{i-1})]
\geq\mathbb{E}[s(\mathcal{P}_{i-1})]+\frac{(1-\frac{1}{\gamma})^{i-1}\cdot s(OPT)-\mathbb{E}[s(\mathcal{P}_{i-1})]}{\gamma}=(1-\frac{1}{\gamma})\cdot\mathbb{E}[s(\mathcal{P}_{i-1})]
+\gamma^{-1}(1-\frac{1}{\gamma})^{i-1}\cdot s(OPT)\geq(1-\frac{1}{\gamma})\cdot[\frac{i-1}{\gamma}\cdot(1-\frac{1}{\gamma})^{i-2}\cdot s(OPT)]+\gamma^{-1}(1-\frac{1}{\gamma})^{i-1}\cdot s(OPT)=[\frac{i}{\gamma}]\cdot(1-\frac{1}{\gamma})^{i-1}\cdot s(OPT)$. Hence, $\mathbb{E}[s(\mathcal{P}_k)]\geq\frac{\gamma}{\gamma}\cdot(1-\frac{1}{\gamma})^{\gamma-1}\cdot s(OPT)\geq e^{-1}\cdot s(OPT)$. That is, Alg.~\ref{alg:greedy} achieves $\frac{1}{e}$-approximation of \textsc{cps}.

\vspace{1ex}\noindent\textbf{Proof of Theorem~\ref{thm:greedyComplexity} (Sketch).} Let $G_{max}=(V_{max},E_{max})$ be the largest candidate pattern in $P_{all}$. In the worst-case, time complexity of Algorithm~\ref{alg:greedy} is $O(|P_{all}|\gamma|V_{max}|!|V_{max}|)$ since there are $|P_{all}|$ candidate patterns and the while-loop in Algorithm~\ref{alg:greedy} iterates at most $\gamma$ times. For each iteration, the score function requires computation of coverage, cognitive load and redundancy which requires $O(|V_{max}|!|V_{max}|)$, $O(|V_{max}|+|E_{max}|)$ and $O(|V_{max}|+|V_{max}|log(|V_{max}|)$~\cite{berlingerio2013}, respectively. Note that $|V_{max}|log(|V_{max}|)\approx|E_{max}|$ in real-world graphs~\cite{berlingerio2013}. The space complexity is due to storage of all candidate patterns. Hence, Algorithm~\ref{alg:greedy} has space complexity of $O(|P_{all}|(|V_{max}|+|E_{max}|))$.

\eat{	\vspace{1ex}\noindent\textbf{Proof of Lemma~\ref{lem:ccpMinSize} (Sketch).} The smallest \textsc{ccp} is formed from merging a 3-\textsc{cp} and a 4-\textsc{cp} since \textsc{ccp} formed from two 3-\textsc{cp}s is actually a 4-\textsc{cp} and not exactly a \textsc{ccp} (Lemma~\ref{lem:combinedSimple3TrussPattern}). Since \textsc{ccp} is formed from merging an edge of two \textsc{cp}s, the size of the smallest \textsc{ccp} is 7 (\textit{i.e.}, $|\textsc{ccp}(3,4)|=|E_{c3}|+|E_{c4}|-1=7$)
}

\end{document}